\documentclass[10pt,journal,compsoc]{IEEEtran}
\pdfoutput=1
\usepackage[nocompress]{cite}
\usepackage{amssymb,amsmath,amsthm,amsfonts}

\newtheorem{property}{Property}

\newtheorem{theorem}{Theorem}
\newtheorem{lemma}{Lemma}

\newtheorem{condition}{Condition}

\newcommand{\bx}{\mathbf{x}}

\newcommand{\bp}{\mathbf{p}}

\newcommand{\bu}{\mathbf{u}}

\newcommand{\cX}{\mathcal{X}}

\newcommand{\cB}{\mathcal{B}}

\newcommand{\bP}{\mathbf{P}}
\newcommand{\cC}{\mathcal{C}}

\def\argmax{\operatorname*{argmax\,}}

\usepackage{mathtools}
\usepackage{algorithm}
\usepackage{algorithmic}
\usepackage{subfigure}
\usepackage{color}
\usepackage{soul}
\usepackage{enumitem}

\begin{document}

\title{Coresets for Triangulation}

\author{Qianggong~Zhang and~Tat-Jun~Chin
\IEEEcompsocitemizethanks{\IEEEcompsocthanksitem The authors are with the School of Computer Science, The University of Adelaide, Adelaide,
SA, 5000, Australia.\protect\\
E-mail: \{qianggong.zhang, tat-jun.chin\}@adelaide.edu.au
}
}

%
%

\markboth{Journal of \LaTeX\ Class Files,~Vol.~14, No.~8, August~2015}%
{Shell \MakeLowercase{\textit{et al.}}: Bare Demo of IEEEtran.cls for Computer Society Journals}
%




\IEEEtitleabstractindextext{%
\begin{abstract}
Multiple-view triangulation by $\ell_\infty$ minimisation has become established in computer vision. State-of-the-art $\ell_\infty$ triangulation algorithms exploit the quasiconvexity of the cost function to derive iterative update rules that deliver the global minimum. Such algorithms, however, can be computationally costly for large problem instances that contain many image measurements, e.g., from web-based photo sharing sites or long-term video recordings. In this paper, we prove that $\ell_\infty$ triangulation admits a \emph{coreset approximation} scheme, which seeks small representative subsets of the input data called \emph{coresets}. A coreset possesses the special property that the error of the $\ell_\infty$ solution on the coreset is within known bounds from the global minimum. We establish the necessary mathematical underpinnings of the coreset algorithm, specifically, by enacting the stopping criterion of the algorithm and proving that the resulting coreset gives the desired approximation accuracy. On large-scale triangulation problems, our method provides theoretically sound approximate solutions. Iterated until convergence, our coreset algorithm is also guaranteed to reach the true optimum. On practical datasets, we show that our technique can in fact attain the global minimiser much faster than current methods.
\end{abstract}

}

\maketitle

\IEEEdisplaynontitleabstractindextext

%
\IEEEpeerreviewmaketitle

\IEEEraisesectionheading{\section{Introduction}\label{sec:introduction}}

\IEEEPARstart{W}{ith} the basic principles and algorithms of structure-from-motion well established, researchers have begun to consider large-scale reconstruction problems involving millions of input images. Arguably such large-scale problems, which arise from, e.g., photo sharing websites or long-term video observations in robotic exploration, are more common and practical. The significant problem sizes involved in such settings, however, compel practitioners to either use distributed computational architectures (e.g., GPU) to perform the required optimisation, or accept approximate solutions for the reconstruction.

This paper contains a \emph{theoretical contribution} under the second paradigm. We introduce a \emph{coreset approximation} scheme (more below) and prove its validity for multiple view 3D reconstruction, specifically for triangulation.

Triangulation is the task of estimating the 3D coordinates of a scene point from multiple 2D image observations of the point, given that the pose of the cameras are known~\cite{hartley1997triangulation}. The task is of fundamental importance to 3D vision, since it enables the recovery of the 3D structure of a scene. Whilst in theory structure and motion must be obtained simultaneously, there are many settings, such as large-scale reconstruction~\cite{snavely2008modeling,furukawa2010accurate} and SLAM~\cite{mur2015probabilistic}, where the camera poses are first estimated with a sparse set of 3D points, before a denser scene structure is produced by triangulating other points using the estimated camera poses.

An established approach for triangulation is by $\ell_\infty$ minimisation~\cite{hartley2004minimization}. Specifically, we seek the 3D coordinates that minimise the maximum reprojection error across all views. Unlike the sum of squared error function which contains multiple local minima, the maximum reprojection error function is quasiconvex and thus contains a single global minimum. Algorithms that take advantage of this property have been developed to solve such quasiconvex problems exactly~\cite{kahl2005multiple,ke2007quasiconvex,seo2007fast,olsson2007efficient,agarwal2008fast,dai2012novel,eriksson2014pseudoconvex,donne2015point}. In particular, Agarwal et al.~\cite{agarwal2008fast} showed that some of the most effective algorithms belong to the class of generalised fractional programming (GFP) methods~\cite{dinkelbach1967nonlinear,gugat1996fast}.

Although algorithms for $\ell_\infty$ triangulation have steadily improved, there is still room for improvement. In particular, on large-scale reconstruction problems or SLAM where there are usually a significant number of views per point (recall that the size of a triangulation problem is the number of 2D observations of a scene point), the computational cost of many of the algorithms can be considerable; we will demonstrate this in Section~\ref{sec:results}. A major reason is that the algorithms need to repeatedly solve convex programs to determine the update direction, which is of cubic complexity in worst case. It is thus of interest to investigate effective approximate algorithms.

\subsection{Contributions} 

As alluded above, our main contribution in this paper is theoretical. Specifically, we prove that the $\ell_\infty$ triangulation problem admits a \emph{coreset} approximation scheme~\cite{agarwal2005geometric,buadoiu2008optimal}. A coreset is a small representative subset of the data that approximates the overall distribution of the data. In the context of $\ell_\infty$ triangulation from $N$ views, our algorithm iteratively accumulates a coreset, such that the error from solving the problem on the coreset is bounded within a factor of $(1+\epsilon)$ from the theoretically achievable minimum. Given a desired $\epsilon$, we establish a stopping criterion for the algorithm such that the output coreset gives the required approximation accuracy. This provides a mathematically justified way to deal with large-scale problems where considering all available data may not be desirable or worthwhile.


Iterated until convergence, the coreset algorithm is guaranteed to attain the globally optimal solution. We experimentally demonstrate that the algorithm can in fact find the global minimiser much faster than many state-of-the-art $\ell_\infty$ triangulation methods. This superior performance was established on publicly available large scale 3D reconstruction datasets. From a practical standpoint, our algorithm thus provides a useful \emph{anytime} behaviour, i.e., the algorithm can simply be run until convergence, or until the time budget is exhausted. In  the latter case, we have a guaranteed bound of the approximation error w.r.t.~the optimum.

The existence of coresets for quasiconvex vision problems was speculated by Li~\cite{li2009efficient}. However, little progress has been made on this subject since. We provide a positive answer on one such problem. Our work is also one of the first to extend the idea of coresets in computational geometry~\cite{agarwal2005geometric,buadoiu2008optimal} to computer vision.
 
\section{Background}\label{sec:background}

Let $\{ \bP_i, \bu_i \}^{N}_{i=1}$ be a set of data for triangulation, consisting of camera matrices $\bP_i \in \mathbb{R}^{3\times 4}$ and observed image positions $\bu_i \in \mathbb{R}^{2}$ of the same scene point $\bx \in \mathbb{R}^3$. In this paper, by a ``datum" we mean a specific camera and image point $\{ \bP_i, \bu_i\}$. Let $\cX = \{1,\dots,N\}$ index the set of data. The $\ell_\infty$ technique estimates $\bx$ by minimising the maximum reprojection error
\begin{align}\label{equ:triang}
\min_{\bx} \max_{i \in \cX}~~&r(\bx \mid \bP_i,\bu_i),\\
\nonumber \text{subject to} \;\; &\bP^{3}_i \tilde{\bx} > 0~~\forall~i \in \cX.
\end{align}
where
\begin{align}\label{equ:reproj2norm}
r(\bx \mid \bP_i,\bu_i) =  \left\| \bu_i - \frac{\bP^{1:2}_i \tilde{\bx}}{\bP^{3}_i \tilde{\bx}} \right\|_2
\end{align}
is the reprojection error. Here, $\bP^{1:2}_i$ and $\bP^3_i$ respectively denote the first-two rows and third row of $\bP_i$, and $\tilde{\bx}$ is $\bx$ in homogeneous coordinates. The reprojection error is basically the Euclidean distance between the observed point $\bu_i$ and the projection of $\bx$ onto the $i$-th image plane. The cheirality constraints $\bP^{3}_i \tilde{\bx} > 0~\forall i \in \cX$ ensure that the estimated point lies in front of all the cameras. 

Problem~\eqref{equ:triang} belongs to a broader class of problems called generalised linear programs (GLP)~\cite{amenta1994helly}. Two properties of GLPs that will be useful later in this paper, are stated in the context of~\eqref{equ:triang} as follows.

\begin{property}[Monotonicity]\label{prop:mono}
	For any $\cC \subseteq \cX$,
	\begin{align}
	\min_\bx \max_{i \in \cC}~r(\bx \mid \bP_i,\bu_i) \le \min_\bx \max_{i \in \cX}~r(\bx \mid \bP_i,\bu_i)
	\end{align}
	given the appropriate cheirality contraints on both sides. \qed
\end{property}

\begin{property} [Support set]\label{prop:basis}
	Let $\bx^*$ and $\delta^*$ respectively be the minimiser and minimised objective value of~\eqref{equ:triang}. There exists a subset $\cB \subseteq \cX$ with $|\cB| \le 4$, such that for any $\cC$ that satisfies $\cB \subseteq \cC \subseteq \cX$, the following holds
	\begin{align}\label{equ:equiv}
	\begin{split}
	\delta^* &= \min_\bx \max_{i \in \cB}~r(\bx \mid \bP_i,\bu_i)\\
	&= \min_\bx \max_{i \in \cC}~r(\bx \mid \bP_i,\bu_i) = \min_\bx \max_{i \in \cX}~r(\bx \mid \bP_i,\bu_i)
	\end{split}
	\end{align}
	given the appropriate cheirality contraints. In fact, the three problems in~\eqref{equ:equiv} have the same minimiser $\bx^*$. Further,
	\begin{align}\label{equ:equiv2}
	r(\bx^* \mid \bP_i,\bu_i) = \delta^* \;\;\;\; \text{for any}~i \in \cB.
	\end{align}
	The subset $\cB$ is called the ``support set" of the problem.\qed
\end{property}

See~\cite{li2009efficient,amenta1994helly,sim2006removing} for details and proofs related to the above properties. Intuitively,~\eqref{equ:equiv2} states that, at the solution of~\eqref{equ:triang}, the minimised maximum error occurs at the support set $\cB$. Fig.~\ref{fig:triang} illustrates this property. Further,~\eqref{equ:equiv} states that solving~\eqref{equ:triang} amounts to solving the same problem on $\cB$. Many classical algorithms in computational geometry~\cite{seidel1991small,clarkson1995vegas,matouvsek1996subexponential} exploit this property to solve GLPs.

\begin{figure}[t]\centering
	\subfigure[]{\includegraphics[width=0.33\columnwidth]{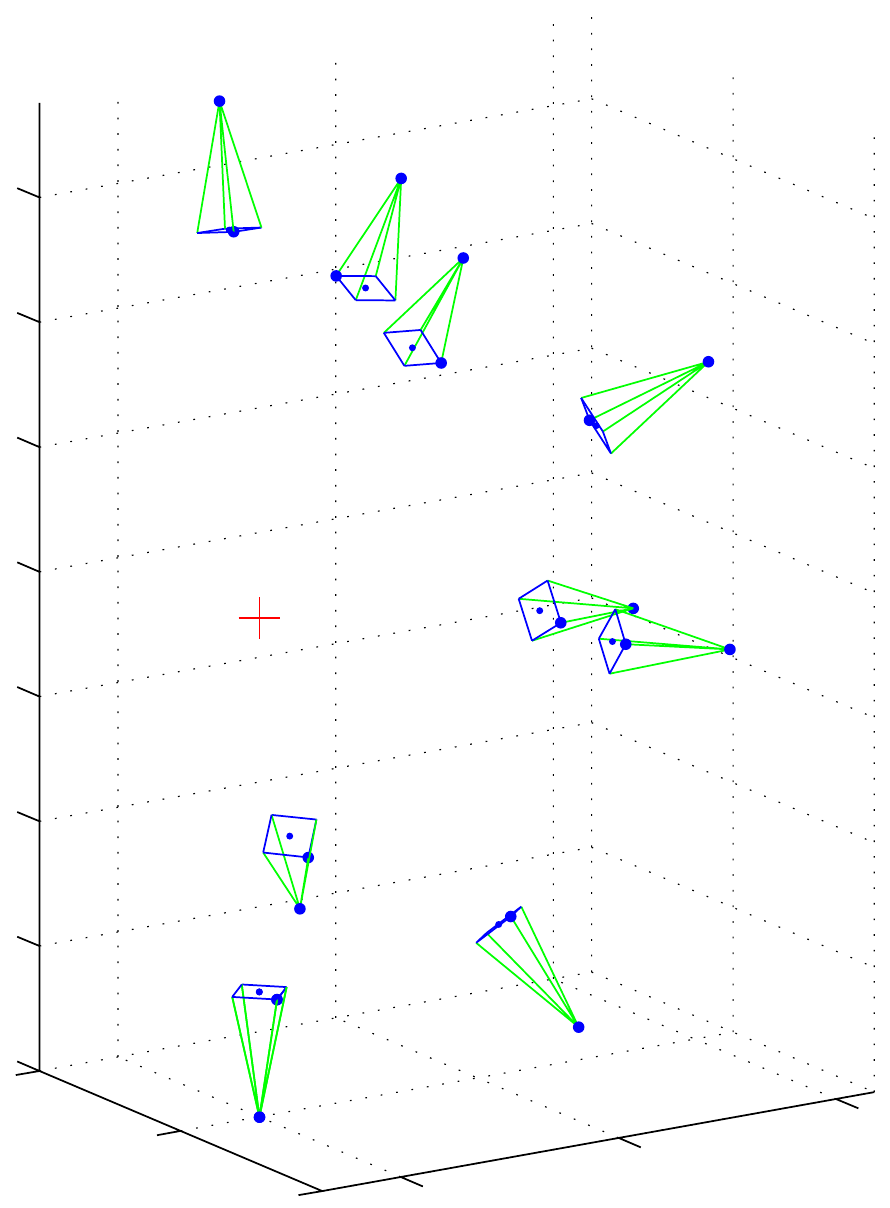}}
	\subfigure[]{\includegraphics[width=0.64\columnwidth]{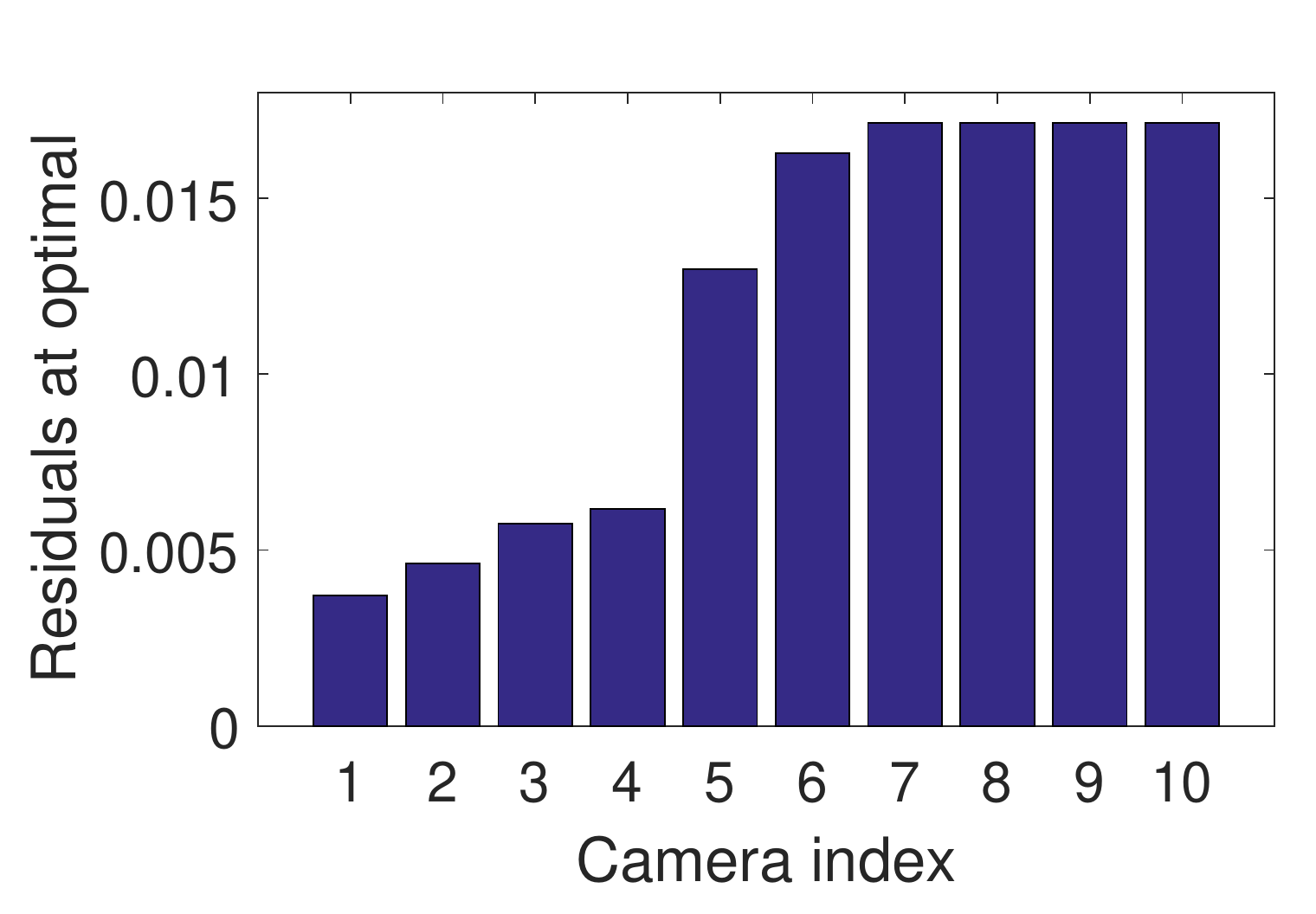}}
	\caption{Triangulating a point $\bx$ observed in 10 views. The red `+' is the $\ell_\infty$ solution $\bx^*$. Observe that there are four views/measurements with the same residual at $\bx^*$. The index of the support set is thus $\cB = \{7,8,9,10\}$.}
	\label{fig:triang}
\end{figure}


\section{Coreset Algorithm}\label{sec:coreset}

We first describe the coreset algorithm and focus on its operational behaviour, before embarking on a discussion of its convergence properties in Sec.~\ref{sec:conv} and the derivation of the coreset approximation bound in Sec.~\ref{sec:approx}.

\subsection{Main Operation}\label{sec:alg}

The coreset algorithm for $\ell_\infty$ triangulation is listed in Algorithm~\ref{alg:coreset}. The primary objective is to seek a representative subset $\cC_s \subseteq \cX$ of the data. This is accomplished by iteratively accumulating the data that should appear in the subset, where the datum that is selected for inclusion at each iteration is the most violating datum; see Step~\ref{step:maxerr}. The size of the subset, and equivalently the runtime of the algorithm, is controlled by the desired approximation error $\epsilon$. To achieve, for e.g., a $1\%$ approximation error, set $\epsilon = 0.01$.

\begin{algorithm}[ht]\centering
	\begin{algorithmic}[1]
		\REQUIRE Input data $\{ \bP_i, \bu_i \}^{N}_{i=1}$, approximation error $\epsilon$.
		\STATE Randomly permute the order of $\{ \bP_i, \bu_i \}^{N}_{i=1}$, and define $\cX = \{1,\dots,N\}$.
		\STATE $s \leftarrow 0$,~~$\gamma \leftarrow \infty$,~~$g \leftarrow 0$,~~$\cC_1 \leftarrow \{1,2,3,4\}$.
		\STATE $(\bx_1,\delta_1) \leftarrow$ Minimiser and minimised value of \eqref{equ:triang} on data indexed by $\cC_1$ \label{step:solver1}
		\STATE $t \leftarrow 2$
		\WHILE{$t \le \lceil2/\epsilon\rceil$}\label{step:startwhile}
		\STATE $q \leftarrow \argmax_{i \in \cX} r(\bx_{t-1} \mid \bP_i,\bu_i)$.\label{step:maxerr}
		\IF{$r(\bx_{t-1} \mid \bP_q, \bu_q) \le \delta_{t-1}$}\label{step:conv}
		\STATE /* Found global minimum */ \\ $s \leftarrow t-1$,~~$g \leftarrow 1$,~~exit while loop.\label{step:early}
		\ENDIF
		\IF{$r(\bx_{t-1} \mid \bP_q, \bu_q) < \gamma$}\label{step:best}
		\STATE /* Found a better coreset */ \\ $s \leftarrow t-1$,~~$\gamma \leftarrow r(\bx_{t-1} \mid \bP_q, \bu_q)$.
		\ENDIF
		\STATE $\cC_t \leftarrow \cC_{t-1} \cup \{ q \}$.\label{step:updateCt}
		\STATE $(\bx_t,\delta_t) \leftarrow$ Minimiser and minimised value of \eqref{equ:triang} on data indexed by $\cC_t$.\label{step:solver2}
		\STATE $t \leftarrow t+1.$ \label{step:tplus1}
		\ENDWHILE\label{step:endwhile}
		\IF{$g = 0$}\label{step:startrefine}
		\STATE $q \leftarrow \argmax_{i \in \cX} r(\bx_{\lceil 2/\epsilon \rceil} \mid \bP_i,\bu_i)$.
		\IF{$r(\bx_{\lceil 2/\epsilon \rceil} \mid \bP_q, \bu_q) < \gamma$}
		\STATE $s \leftarrow \lceil 2/\epsilon \rceil$.
		\ENDIF
		\ENDIF\label{step:endrefine}
		\RETURN $\cC_s$, $\bx_s$ and $\delta_s$.
	\end{algorithmic}
	\caption{Coreset algorithm for $\ell_\infty$ triangulation~\eqref{equ:triang}.}
	\label{alg:coreset}
\end{algorithm}

Observe that Algorithm~\ref{alg:coreset} is a \emph{meta-algorithm}, since it requires executing a solver for \eqref{equ:triang} on the data subset indexed by the current subset $\cC_t$ (see Steps~\ref{step:solver1} and~\ref{step:solver2}). Any of the previous $\ell_\infty$ triangulation algorithms~\cite{kahl2005multiple,ke2007quasiconvex,seo2007fast,olsson2007efficient,agarwal2008fast,dai2012novel} can be applied as the solver.

There are two terminating conditions for Algorithm~\ref{alg:coreset}:
\begin{enumerate}
	\vspace{0.5em}
	\item Iteration counter $t$ reaches $\lceil 2/\epsilon \rceil$.\\
	
	\vspace{-0.5em}
	In this case, the output $\cC_s$ indexes a coreset with the desired approximation accuracy $\epsilon$. Section~\ref{sec:approx} will establish the error bound for approximating \eqref{equ:triang} using the data indexed by $\cC_s$.
	
	\vspace{0.75em}
	
	\item The global minimiser has been found (Step~\ref{step:early}).\\

	\vspace{-0.5em}
	The satisfaction of the condition in Step~\ref{step:conv} implies that $\cC_{t-1}$ already contains the support set $\cB$, since the largest error across all $\cX$ is not larger than the value of \eqref{equ:triang} on the data indexed by $\cC_{t-1}$; see Property~\ref{prop:basis}.
	
	\vspace{0.75em}
\end{enumerate}

To aid intuition, a sample partial run of Algorithm~\ref{alg:coreset} is shown in Fig.~\ref{fig:samplerun}.

\begin{figure*}[ht]\centering
	\subfigure[$t = 1$ (initialisation)]{\includegraphics[width=0.20\textwidth]{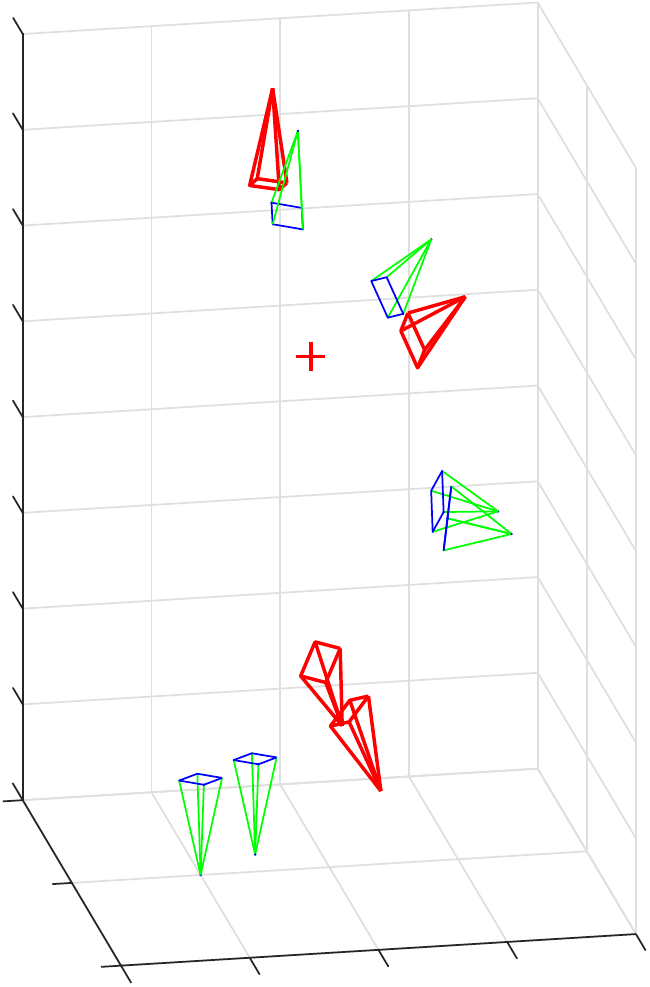}}\hfill
	\subfigure[$t = 2$]{\includegraphics[width=0.20\textwidth]{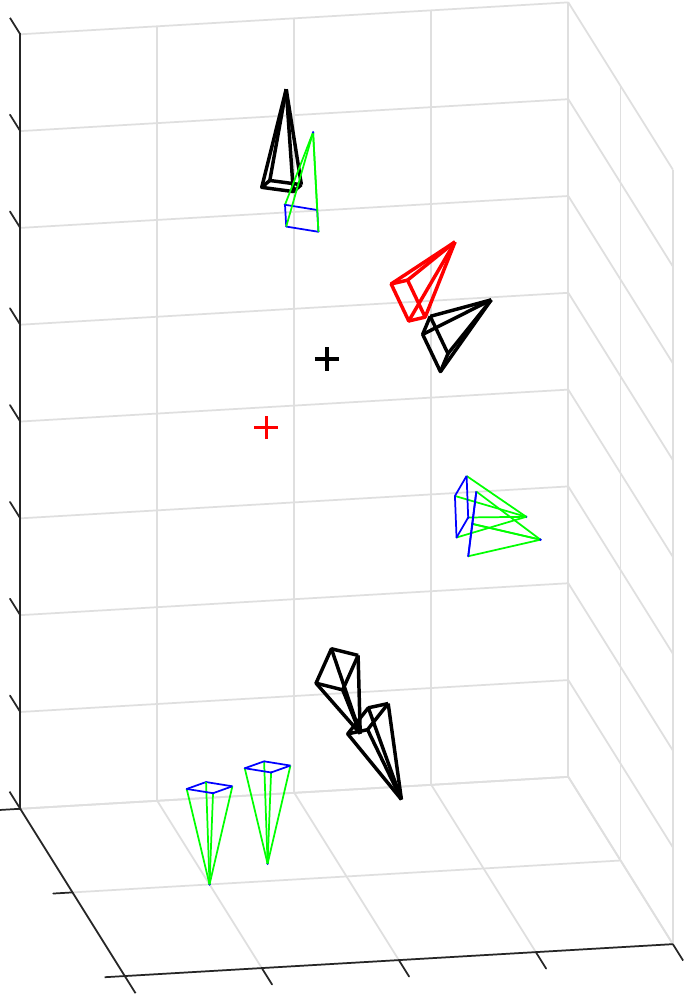}}\hfill
	\subfigure[$t = 3$]{\includegraphics[width=0.20\textwidth]{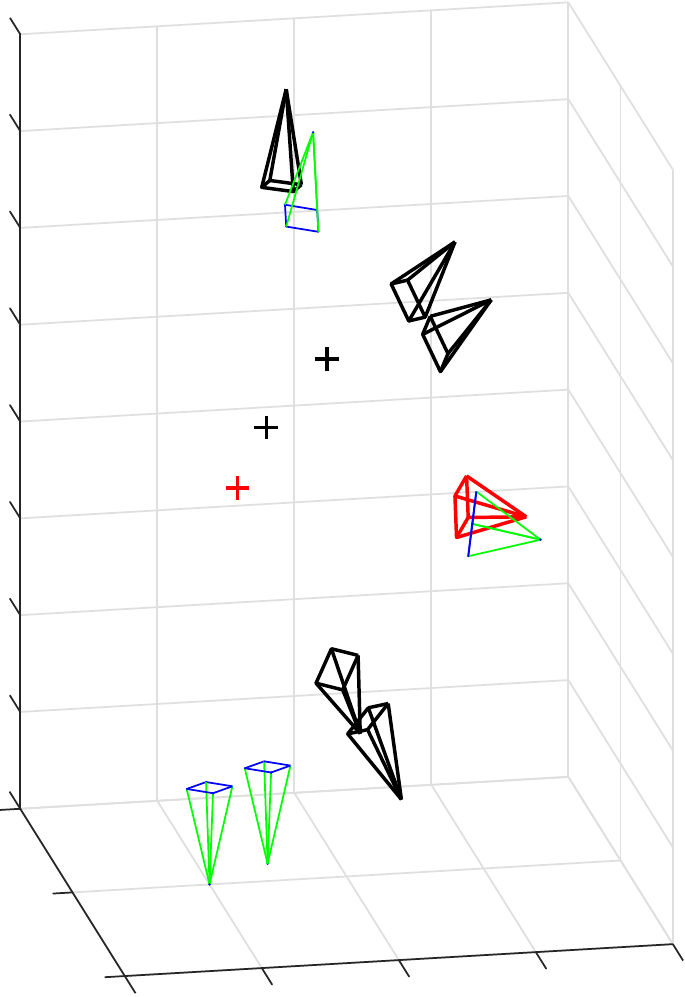}}\hfill
	\subfigure[$t = 4$]{\includegraphics[width=0.20\textwidth]{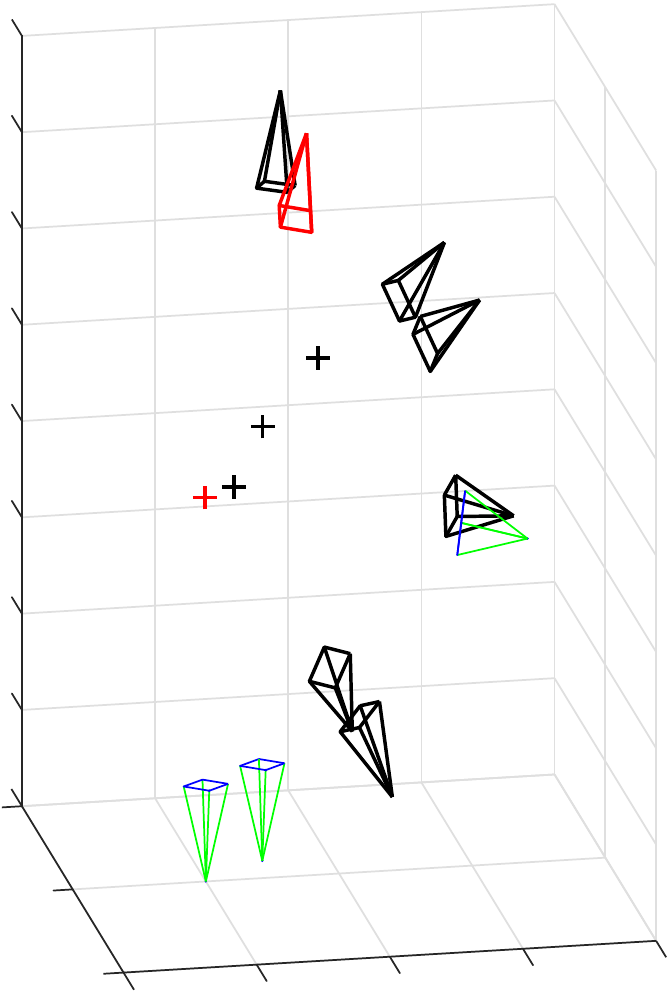}}\hfill
	\caption{A sample run of Algorithm \ref{alg:coreset} on the data displayed in Fig.~\ref{fig:triang}. (a) Four image measurements/camera viewpoints (in red) were selected to form the initial coreset $\cC_1$. The current solution $\bx_1$ is shown as a red cross. (b)--(d) Algorithm~\ref{alg:coreset} progressively inserts new data into the coreset. Data in the current coreset is shown in black, and the newly inserted datum (chosen according to Step~\ref{step:maxerr}) is shown in red. Similary, the previous solutions $\bx_s$ are shown in black, and the current estimate is shown in red. If terminated at $t = \lceil 2/\epsilon \rceil$, the estimate is a $\epsilon$-approximation of the true optimum. Iterated until convergence, the global optimum is achieved. For anytime behaviour, the error bound can be backtracked (see Sec.~\ref{sec:backtrack}) to obtain the approximation error of the last estimate at termination.}
	\label{fig:samplerun}
\end{figure*}

\subsection{Convergence to Global Minimum}\label{sec:conv}

If we are only interested in the global minimiser $\bx^*$, then $\epsilon$ should be set to $0$ (or a value small enough such that $\lceil 2/\epsilon \rceil \ge N - 3$). We prove that with this setting Algorithm~\ref{alg:coreset} will always find $\bx^*$ in a finite number of steps.

\begin{theorem}\label{thm:converge}
	If $\lceil 2/\epsilon \rceil \ge N - 3$, then Algorithm~\ref{alg:coreset} finds $\bx^*$ in finite time.
\end{theorem}
\begin{proof}
	Let $q$ be obtained according to Step~\ref{step:maxerr}.
	\begin{itemize}[parsep=0.75pt,topsep=0.25em,itemsep=2pt]
		\item If $q \in \cC_{t-1}$, then, by how $\bx_{t-1}$ and $\delta_{t-1}$ were calculated in Step~\ref{step:solver2}, the condition in Step~\ref{step:conv} must be satisfied and $\bx_{t-1}$ is the global minimiser.
		\item If $q \notin \cC_{t-1}$ and the condition in Step~\ref{step:conv} is satisfied, then equation~\eqref{equ:equiv} is implied and $\bx_{t-1}$ is the global minimiser.
		\item If $q \notin \cC_{t-1}$ and the condition in Step~\ref{step:conv} is not satisfied, then Algorithm~\ref{alg:coreset} will insert $q$ into $\cC_{t-1}$. There are at most $N$ of such insertions (including the initial four insertions into $\cC_1$). If $\lceil 2/\epsilon \rceil \ge N - 3$, in the worst case all of $\cX$ will finally be inserted, and $\cC_{\lceil 2/\epsilon \rceil} = \cX$ and $\bx_{\lceil 2/\epsilon \rceil} = \bx^*$. 
	\end{itemize}
\end{proof}

Note that, whilst Algorithm~\ref{alg:coreset} needs to repeatedly call an $\ell_\infty$ solver, it only invokes the solver on a small subset $\cC_t$ of the data. Second, the way a new datum is selected (Step~\ref{step:maxerr}) to be inserted into $\cC_{t-1}$---basically by choosing the most violating datum w.r.t.~the current solution---enables $\cB$ to be found quickly. Section~\ref{sec:results} demonstrates that Algorithm~\ref{alg:coreset} can in fact find the global minimiser much more efficiently than invoking an $\ell_\infty$ solver~\cite{kahl2005multiple,ke2007quasiconvex,seo2007fast,olsson2007efficient,agarwal2008fast,dai2012novel} in ``batch mode" on the whole input data $\cX$.

Utilised as a global optimiser (i.e., set $\epsilon = 0$), Algorithm~\ref{alg:coreset} can be viewed as a \emph{Las Vegas} style randomised algorithm, since it always finds the correct result but a non-deterministic runtime. In addition, as indicated in the proof of Theorem~\ref{thm:converge}, in the worst case Algorithm~\ref{alg:coreset} takes $N$ iterations since it considers each measurement at most once.

\subsection{Coreset Approximation}\label{sec:approx}

Our primary contribution is to show that the subset $\cC_s$ output by Algorithm~\ref{alg:coreset} is a coreset of the $\ell_\infty$ triangulation problem~\eqref{equ:triang}. This is conveyed by the following theorem, which bounds the error of approximating~\eqref{equ:triang} using $\cC_s$.

\begin{theorem}\label{thm:coreset}
Let $\cC_s$, $\bx_s$ and $\delta_s$ be the output of Algorithm~\ref{alg:coreset}. Then
\begin{align}\label{equ:bound}
\max_{i\in \cX}~r(\bx_s \mid \bP_i, \bu_i) \le (1+\epsilon) \delta^*,
\end{align}
where $\delta^*$ is the minimised objective value for problem~\eqref{equ:triang} on the full data $\cX$.
\end{theorem}

Intuitively, the above theorem states that the error of approximating $\bx^*$ with $\bx_s$ (the latter was computed using the $\cC_s$ output by Algorithm \ref{alg:coreset}) is at most $(1+\epsilon)$-times of the smallest possible error. This provides a mathematically justified way of dealing with large scale problems. The rest of this subsection is devoted to proving Theorem~\ref{thm:coreset}.

First, we define the set of geometrical quantities in Fig.~\ref{fig:proj}. For an arbitrary camera matrix $\bP$ with measurement $\bu$, the reprojection error of a given $\bx$ is
\begin{align}
r(\bx \mid \bP,\bu) = \left\| \bu - \frac{\bP^{1:2} \tilde{\bx}}{\bP^{3} \tilde{\bx}} \right\|_2 = \left\| \bu - f_{\bP}(\bx) \right\|_2,
\end{align}
where $f_{\bP}(\bx)$ is the projection of $\bx$ onto the image. Given a set of data $\{ \bP_i, \bu_i \}^{N}_{i=1}$, let $\bx^*$ be the global minimiser of~\eqref{equ:triang}. Define a disc on the image plane with centre $\bu$ and radius $r(\bx^* \mid \bP,\bu)$; backprojecting this disc creates a solid elliptic cone $c_{\bP}(\bx^*)$. Define $h_\bP(\bx^*)$ as the tangent plane on the surface of $c_{\bP}(\bx^*)$ that contains $\bx^*$.

\begin{figure}[ht]\centering
	\includegraphics[width=0.35\textwidth]{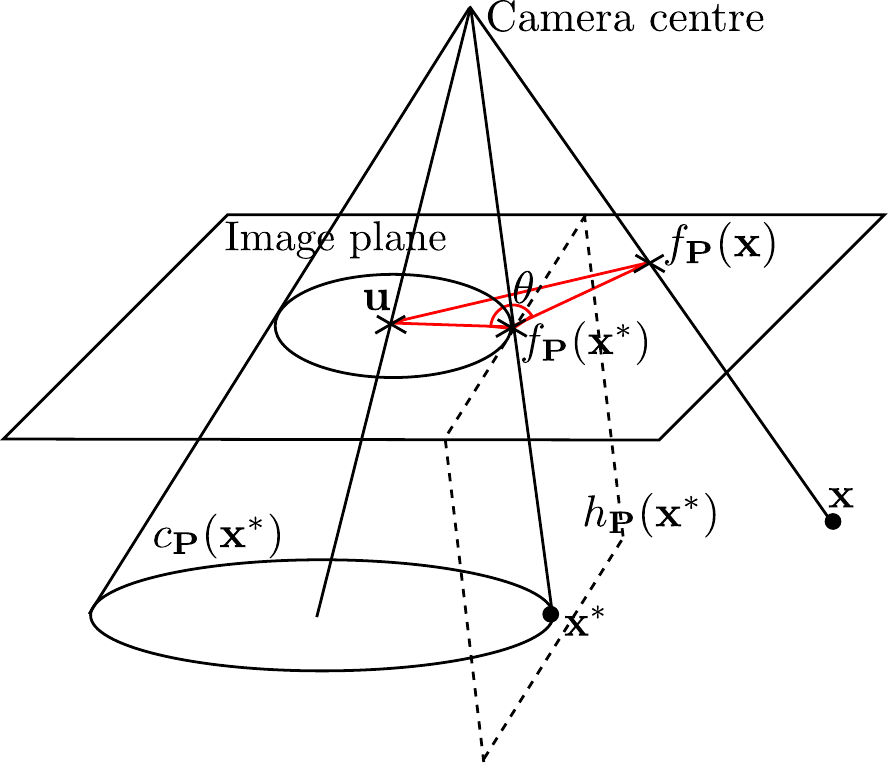}
	\caption{Definition of several geometrical quantities for $\ell_\infty$ triangulation.}
	\label{fig:proj}
\end{figure}

We now establish several intermediate results. In the following, we consider only $\bx \in \mathbb{R}^3$ that lies in front of the camera, i.e., $\bx$ is never on the same side of the image plane as the camera centre.

\begin{lemma}\label{lem:inside}
	$\bx$ is inside $c_{\bP}(\bx^*)$ iff $r(\bx \mid \bP,\bu) < r(\bx^* \mid \bP,\bu)$.
\end{lemma}
\begin{lemma}\label{lem:acute}
	$\bx$ is on the same side of $h_\bP(\bx^*)$ as $\bu$ iff the angle $\theta$ formed by the three points $\bu:f_\bP(\bx^*):f_\bP(\bx)$ is acute, i.e., $\theta < 90^\circ$.
\end{lemma}
\begin{lemma}\label{lem:opposite}
	$\bx$ is on the opposite side of $h_\bP(\bx^*)$ as $\bu$ iff the angle $\theta$ formed by the three points $\bu:f_\bP(\bx^*):f_\bP(\bx)$ is obtuse, i.e., $\theta > 90^\circ$.
\end{lemma}

The above three lemmata can be proven easily by inspecting Fig.~\ref{fig:proj}. As an extension of Lemma~\ref{lem:acute}, the following statement can be made.

\begin{lemma}\label{lem:acute2}
The angle $\angle(\bu:f_\bP(\bx^*):f_\bP(\bx))$ is acute iff there is a line segment
\begin{align}\label{equ:seg}
S = \{ \bx^\prime \mid \bx^\prime = \bx^* + \alpha(\bx - \bx^*), 0 \le \alpha <1 \}
\end{align}
(i.e., $S$ has a start point at $\bx^*$ and lies along vector $\bx - \bx^*$) such that any point $\bx^\prime$ on $S$ will give a strictly smaller reprojection error than $\bx^*$, i.e.,
\begin{align}
r(\bx^\prime \mid \bP,\bu) < r(\bx^* \mid \bP,\bu) \;\;\;\; \forall \; \bx^\prime \in S.  
\end{align}
\end{lemma}
\begin{proof}
If $\angle(\bu:f_\bP(\bx^*):f_\bP(\bx))$ is acute, then from Lemma~\ref{lem:acute}, $\bx$ must be on the same side of $h_\bP(\bx^*)$ as $\bu$. The line segment joining $\bx$ and $\bx^*$ must thus intersect the inside of $c_\bP(\bx^*)$; this intersection gives $S$. Since $S$ is inside $c_\bP(\bx^*)$, from Lemma~\ref{lem:inside} any $\bx^\prime \in S$ must give a strictly smaller reprojection error than $\bx^*$.
	
The reverse direction can be proven by realising that any $\bx^\prime$ which gives a strictly smaller reprojection error than $\bx^*$ must lie in $c_\bp(\bx^*)$. Any line segment that joins $\bx^*$ and $\bx$ with $\bx^\prime$ in the middle must lie on the same side of $h_\bP(\bx^*)$ as $\bu$. From Lemma~\ref{lem:acute}, $\angle(\bu:f_\bP(\bx^*):f_\bP(\bx))$ must be acute.
\end{proof}

Of central importance is the following result.

\begin{lemma}\label{lem:obtuse}
Let $\{ \bP_i, \bu_i \}^{N}_{i=1}$ be a set of data, $\bx^*$ be the global minimiser of~\eqref{equ:triang} on the data, and $\delta^*$ be the minimised value of~\eqref{equ:triang}. For an arbitrary $\bx \in \mathbb{R}^3$ in front of the camera, there exists a datum $\{\bP_j, \bu_j\}$ such that
\begin{align}
\angle(\bu_j:f_{\bP_j}(\bx^*):f_{\bP_j}(\bx)) > 90^\circ.
\end{align}
Via the cosine rule, the above inequality can be re-expressed as
\begin{align}\label{equ:ineq}
r(\bx \mid \bP_j,\bu_j)^2 \ge \left\| f_{\bP_j}(\bx) - f_{\bP_j}(\bx^*) \right\|^2_2 + r(\bx^* \mid \bP_j,\bu_j)^2.
\end{align}
\end{lemma}
\begin{proof}
From~\eqref{equ:equiv2}, at the solution $\bx^*$ there must exist a support set $\cB$ such that the data indexed by $\cB$ attain the minimised maximum residual $\delta^*$.
	
It is sufficient to consider $\cB$. We aim to contradict the following assumption:
\begin{align}\label{equ:ass1}
\exists\bx \;\;\; \text{s.t.} \;\;\; \angle(\bu_i:f_{\bP_i}(\bx^*):f_{\bP_i}(\bx)) < 90^\circ \;\;\;\; \forall \;\; i\in \cB.
\end{align}
Given an $\bx$ that satisfies~\eqref{equ:ass1}, Lemma~\ref{lem:acute2} states that for each $i \in \cB$, there is a line segment $S_i$ that lies completely inside $c_{\bP_i}(\bx^*)$.
	
Amongst all the segments $S_i$, $i \in \cB$, pick the shortest one and call it $\bar{S}$. The segment $\bar{S}$ must lie simultaneously in all of the cones $c_{\bP_i}(\bx^*)$, $i \in \cB$ (recall from~\eqref{equ:seg} that all $S_i$'s begin at $\bx^*$ and lie along vector $\bx - \bx^*$). Any $\bx^\prime \in \bar{S}$ must thus yield a strictly smaller reprojection error than $\bx^*$ for all $\{ \bP_i,\bu_i\}$, $i \in \cB$. This contradicts that $\cB$ is a support set, thus falsifying~\eqref{equ:ass1}.
	
The falsity of~\eqref{equ:ass1} implies that for an arbitrary $\bx$, there must be an $i \in \cB$ such that $\angle(\bu_i:f_{\bP_i}(\bx^*):f_{\bP_i}(\bx)) > 90^\circ$---set $j$ as that $i$. 
\end{proof}

Given the above results, we adapt B\v{a}doiu and Clarkson's derivation~\cite{buadoiu2008optimal} to yield the inequality~\eqref{equ:bound} for triangulation. Define
\begin{align}\label{equ:def}
\begin{split}
&\bar{\delta} := (1+\epsilon)\delta^*,\\
&\lambda_t := \delta_t/\bar{\delta},\\
&k^{\bP} := \left\| f_{\bP}(\bx_1) - f_{\bP}(\bx_2)  \right\|_2.
\end{split}
\end{align}
Note that $0 \le \lambda_t \le 1$. Further, since $\cC_t \subseteq \cX$, from Property~\ref{prop:mono}, $\delta_t \le \delta^*$, thus
\begin{align}\label{equ:lambda}
\lambda_t \le \delta^*/\bar{\delta} = 1/(1+\epsilon).
\end{align}

We aim to disprove the following assumption:
\begin{align}\label{equ:ass2}
\nexists t \ge 2 \;\; \text{such that} \;\; \max_{i \in \cX}~r(\bx_t \mid \bP_i,\bu_i) \le (1+\epsilon)\delta^*.
\end{align}
In words,~\eqref{equ:ass2} effectively states that none of the $\cC_t$ accumulated throughout the iterations in Algorithm~\ref{alg:coreset} gives a $(1+\epsilon)$ approximation to~\eqref{equ:triang}.

For any $t \ge 2$, Lemma~\ref{lem:obtuse} states that there exists a $j \in \cC_{t-1}$ such that
\begin{align}
\begin{split}
r(\bx_t &\mid \bP_j,\bu_j)^2 \\
&\ge \left\| f_{\bP_j}(\bx_t) - f_{\bP_j}(\bx_{t-1}) \right\|^2_2 + r(\bx_{t-1} \mid \bP_j,\bu_j)^2 \\
&=(k^{\bP_j}_t)^2 + \delta^2_{t-1}
\end{split}
\end{align}
(recall that $j$ indexes a datum in the support set of the data indexed by $\cC_{t-1}$, thus $r(\bx_{t-1} \mid \bP_j,\bu_j) = \delta_{t-1}$). Then
\begin{align}\label{equ:cineq1}
r(\bx_t \mid \bP_j,\bu_j) \ge \sqrt{\lambda^2_{t-1}\bar{\delta}^2 + (k^{\bP_j}_t)^2}.
\end{align}
For the $q$ chosen in iteration $t$ (Step~\ref{step:maxerr} in Algorithm~\ref{alg:coreset}), via triangle inequality
\begin{align}
\begin{split}
r(\bx_{t-1} &\mid \bP_q, \bu_q) \\
&\le r(\bx_{t} \mid \bP_q, \bu_q) + \left\| f_{\bP_q}(\bx_t) - f_{\bP_q}(\bx_{t-1}) \right\|_2,
\end{split}
\end{align}
implying that
\begin{align}
\begin{split}
r(\bx_{t} &\mid \bP_q, \bu_q)\\
&\ge r(\bx_{t-1} \mid \bP_q, \bu_q) - \left\| f_{\bP_q}(\bx_t) - f_{\bP_q}(\bx_{t-1}) \right\|_2\\
&= r(\bx_{t-1} \mid \bP_q, \bu_q) - k^{\bP_q}_t > \bar{\delta} - k^{\bP_q}_t.\label{equ:cineq2}
\end{split}
\end{align}
The last inequality follows from the assumption in~\eqref{equ:ass2} which states that none of the $\cC_t$ for $t \ge 2$ gives a $(1+\epsilon)$ approximation of~\eqref{equ:triang}.

Since both $j$ and $q$ are in $\cC_t$, by combining~\eqref{equ:cineq1} and~\eqref{equ:cineq2} we obtain
\begin{align}\label{equ:themax}
\begin{split}
\lambda_t \bar{\delta} = \delta_t &\ge \max\left(r(\bx_t \mid \bP_j,\bu_j),~r(\bx_{t} \mid \bP_q, \bu_q)\right)\\
&\ge \max\left( \sqrt{\lambda^2_{t-1}\bar{\delta}^2 + (k^{\bP_j}_t)^2},~~\bar{\delta} - k^{\bP_q}_t\right).
\end{split}
\end{align}

Recall the definition of $k^{\bP_j}_t$ and $k^{\bP_q}_t$:
\begin{align}
k^{\bP_j}_t &= \left\| f_{\bP_j}(\bx_t) - f_{\bP_j}(\bx_{t-1}) \right\|_2,\\
k^{\bP_q}_t &= \left\| f_{\bP_q}(\bx_t) - f_{\bP_q}(\bx_{t-1}) \right\|_2.
\end{align}
Geometrically, these quantities represent the 2D projection, respectively on cameras $j$ and $q$, of the 3D shift $\|\bx_t - \bx_{t-1}\|_2$ between the current and previous estimates.

At this juncture, the rest of the proof diverges based on the following conditions:
\begin{align}
k^{\bP_j}_t \ge k^{\bP_q}_t \;\;\;\; \text{or} \;\;\;\; k^{\bP_j}_t < k^{\bP_q}_t.
\end{align}

\begin{condition}\label{cond1}
$k^{\bP_j}_t \ge k^{\bP_q}_t$.
\end{condition}

Under Condition~\ref{cond1}, and following from~\eqref{equ:themax},
\begin{align}\label{equ:themax2}
\begin{split}
\lambda_t \bar{\delta} &\ge \max\left( \sqrt{\lambda^2_{t-1}\bar{\delta}^2 + (k^{\bP_j}_t)^2},~~\bar{\delta} - k^{\bP_q}_t\right)\\
&\ge \max\left( \sqrt{\lambda^2_{t-1}\bar{\delta}^2 + (k^{\bP_j}_t)^2},~~\bar{\delta} - k^{\bP_j}_t\right),
\end{split}
\end{align}
where the second inequality follows since $k_t^{\bP_j}$ and $k_t^{\bP_q}$ are both non-negative quantities. Interpreting the arguments in the second $\max$
\begin{align}
\sqrt{\lambda^2_{t-1}\bar{\delta}^2 + (k^{\bP_j}_t)^2} \;\;\; \text{and} \;\;\; \bar{\delta} - k^{\bP_j}_t
\end{align}
as two functions of $k^{\bP_j}_t$, observe that the first function increases with $k^{\bP_j}_t$ whilst the second decreases with $k^{\bP_j}_t$. Therefore, the RHS of~\eqref{equ:themax2} achieves its minimum when 
\begin{align}\label{equ:equal_max}
\sqrt{\lambda^2_{t-1}\bar{\delta}^2 + (k^{\bP_j}_t)^2} = \bar{\delta} - k^{\bP_j}_t.
\end{align}
Solving \eqref{equ:equal_max} for $k^{\bP_j}_t$ and replacing it in \eqref{equ:themax2}, we arrive at 
\begin{align}\label{equ:cineq3}
\lambda_t \bar{\delta} \ge \frac{1+\lambda^2_{t-1}}{2}\bar{\delta} \implies 1 - \lambda_t \le \frac{1-\lambda^2_{t-1}}{2}.
\end{align}
The second inequality in~\eqref{equ:cineq3} can ``inverted" as
\begin{align}
\frac{1}{1-\lambda_t} &\ge \frac{2}{(1-\lambda_{t-1})(1+\lambda_{t-1})}\\
&= \frac{1}{1-\lambda_{t-1}} + \frac{1}{1+\lambda_{t+1}} > \frac{1}{1-\lambda_{t-1}} + \frac{1}{2},
\end{align}
where the last step is due to $\lambda_{t-1} < 1$. By recursively expanding the above from $t, t-1, \dots, 2$, and recalling that $0 \le \lambda_1 \le 1$, we obtain
\begin{align}
\frac{1}{1-\lambda_t} > \frac{1}{1-\lambda_{1}}  + \frac{t-1}{2} > 1 + \frac{t-1}{2},
\end{align}
which implies
\begin{align}
\lambda_t > 1 - \frac{2}{1+t}.
\end{align}
For $t = \lceil 2/\epsilon \rceil + 1$ the last inequality reduces to
\begin{align}
\begin{split}
\lambda_{\lceil 2/\epsilon \rceil + 1} > 1 - \frac{2}{1 + (2/\epsilon+1)} = \frac{1}{1+\epsilon}
\end{split}
\end{align}
which contradicts~\eqref{equ:lambda}. Thus,~\eqref{equ:ass2} cannot be true. Whilst it may be disconcerting that we have chosen an iteration count $t = \lceil 2/\epsilon \rceil + 1$ that does not exist in Algorithm~\ref{alg:coreset}, for the purpose of a theoretical argument we can always arbitrarily extend the algorithm by one iteration.
	
Since~\eqref{equ:ass2} is false, there must be a $2 \le t \le \lceil 2/\epsilon \rceil$ (say $t^*$) that yields a $(1+\epsilon)$ approximation. The set index by $\cC_s$, which satisfies
\begin{align}
\max_{i \in \cX} r(\bx_s \mid \bP_i,\bu_i) \le \max_{i \in \cX} r(\bx_{t^*} \mid \bP_i,\bu_i) \le (1+\epsilon)\delta^*,
\end{align}
is thus a coreset. Assuming Condition~\ref{cond1} is always satisfied, therefore, the proof for Theorem~\ref{thm:coreset} is complete.

\begin{condition}\label{cond2}
$k^{\bP_j}_t < k^{\bP_q}_t$.
\end{condition}

The above derivations for Condition~\ref{cond1} unfortunately do not cover Condition~\ref{cond2}. In other words, if Condition~\ref{cond2} occurs during the iterations in Algorithm~\ref{alg:coreset}, we cannot guarantee that the output coreset $\cC_s$ satisfies Theorem~\ref{thm:coreset}.



Fortunately this deficiency can be rectified by a small tweak to Algorithm~\ref{alg:coreset}. Specifically, we replace Steps~\ref{step:updateCt} to \ref{step:tplus1} in the main algorithm with the slightly more elaborate steps in Algorithm~\ref{alg:coreset_excerpt}. Note that index $j$ in Algorithm~\ref{alg:coreset_excerpt} is appropriately chosen from $\cC_{t-1}$ according to Lemma~\ref{lem:obtuse}. Intuitively, the modification causes Algorithm~\ref{alg:coreset} to ``skip a step" whenever Condition~\ref{cond2} presents itself; namely, the most violating datum $q$ is still inserted into the coreset, but the iteration counter is not incremented.

\begin{algorithm}
\begin{algorithmic}[1]
\IF{$k^{\bP_j}_t \ge k^{\bP_q}_t$}
\STATE $\cC_t \leftarrow \cC_{t-1} \cup \{ q \}$.
\STATE $(\bx_t,\delta_t) \leftarrow$ Minimiser and minimised value of \eqref{equ:triang} on data indexed by $\cC_t$.
\STATE $t \leftarrow t+1.$
\ELSE
\STATE $\cC_{t-1} \leftarrow \cC_{t-1} \cup \{ q \}$.
\STATE $(\bx_{t-1},\delta_{t-1}) \leftarrow$ Minimiser and minimised value of \eqref{equ:triang} on data indexed by $\cC_{t-1}$.
\ENDIF
\end{algorithmic}
\caption{Pseudo-code to replace Steps~\ref{step:updateCt} to~\ref{step:tplus1} in Algorithm~\ref{alg:coreset} to accommodate Condition~\ref{cond2}.}
\label{alg:coreset_excerpt}
\end{algorithm}

With the modification above and by Property~\ref{prop:mono}, we have that two successive coresets $\cC_{t-1}$ and $\cC_t$ give
\begin{align}
\delta_{t-1} \le \delta_t,
\end{align}
since $\cC_{t-1} \subseteq \cC_t$. Thus, all the derivations starting from \eqref{equ:themax} will hold, and the output coreset $\cC_s$ from Algorithm~\ref{alg:coreset} with the above modification will always satisfy Theorem~\ref{thm:coreset}.

\subsection{Size of Output Coreset and Runtime Analysis}\label{sec:sizeofcoreset}

With the modification to account for Condition~\ref{cond2}, we must anticipate that in general 
\begin{align}
|\cC_{t}| - |\cC_{t-1}| \ge 1,
\end{align}
i.e., there could be occurrences of Condition~\ref{cond2} between any two successive increments to the iteration counter. This also implies that the size of the output coreset $\cC_s$ is
\begin{align}
\lceil 2/\epsilon \rceil + 3 + V,
\end{align}
where $V$ is the total number of occurrences of Condition~\ref{cond2} throughout Algorithm~\ref{alg:coreset} (recall that $\cC_1$ is initialised already with four of the available measurements).

To facilitate the analysis of the runtime, define $\alpha$ as the probability of Condition~\ref{cond2} occurring at a particular iteration of the main loop (thus, $V = \alpha\lceil 2/\epsilon \rceil / (1-\alpha)$). Therefore, the number of times the main loop is traversed is
\begin{align}\label{equ:alpha}
\lceil 2/\epsilon \rceil/(1-\alpha).
\end{align}
We argue that the value of $\alpha$ is dependent mainly on the distribution of the cameras and the structure of the scene, rather than on the number of measurements $N$ itself (see evidence in Sec.~\ref{sec:synth}). Under this assumption, the number of effective iterations of Algorithm~\ref{alg:coreset}, and hence the size of the output coreset, depends only on $\epsilon$.

Of course, the actual runtime of Algorithm~\ref{alg:coreset} depends closely on the routine used to solve~\eqref{equ:triang} at each iteration. Many previous studies have shown that~\eqref{equ:triang} can be solved efficiently~\cite{kahl2005multiple,olsson2007efficient,agarwal2008fast}, more so since the solver need only be invoked on a small subset $\cC_t$ at each iteration in Algorithm~\ref{alg:coreset}. Sec.~\ref{sec:results} will investigate the actual runtimes of Algorithm~\ref{alg:coreset} on real image datasets for 3D reconstruction.

\subsection{Error Backtracking for Anytime Operation}\label{sec:backtrack}

By anytime mode, we mean the allowance to stop Algorithm~\ref{alg:coreset} prematurely (i.e., before any of the terminating conditions are met) with the ability to bound the approximation error of the current best estimate $\bx_s$ w.r.t.~$\bx^\ast$.

If Algorithm~\ref{alg:coreset} is run up to $t = \lceil 2/\epsilon \rceil$ (which implies that global convergence has not occurred before that), then the bound~\eqref{equ:bound} holds. To facilitate analysis with non-integral $t$, we can equivalently state that if Algorithm~\ref{alg:coreset} is run until \emph{at least} $t = 2/\epsilon$, then~\eqref{equ:bound} holds. Inverting the relationship between $t$ and $\epsilon$, we can restate the bound as
\begin{align}\label{equ:boundloose}
\max_{i\in \cX}~r(\bx_s \mid \bP_i, \bu_i) \le (1+2/t) \delta^*,
\end{align}
i.e., if Algorithm~\ref{alg:coreset} is run up until an arbitrary $t$, we can expect a $(2/t)$-factor approximation to~\eqref{equ:triang} (note that $\bx_s$ in this case is the best estimate up to iteration $t$).

Care must be taken to spell out "running up until an arbitrary $t$". By this, we mean running the main loop (Steps~\ref{step:startwhile} to~\ref{step:endwhile}) in its entirely (including the modification summarised in Algorithm~\ref{alg:coreset_excerpt}) under that particular $t$ value, \emph{and then} conducting the post-hoc refinement (Steps~\ref{step:startrefine} to~\ref{step:endrefine}) such that the estimate $\bx_t$ from the last coreset $\cC_t$ has a chance to be used to update the incumbent $\bx_s$. If this last update is not attempted, then~\eqref{equ:boundloose} is not guaranteed to hold.


\section{Validation on Synthetic Data}\label{sec:synth}

Here we validate our theoretical results above on synthetically generated data for triangulation.





\subsection{Data Generation}

We synthesised four types of camera pose distributions:
\begin{itemize}
\item Type A: Camera positions were on a straight line. 
\item Type B: Camera positions were randomly distributed.
\item Type C: Camera positions were on a circle.
\item Type D: Stereo cameras with fixed baselines and positions randomly distributed.
\end{itemize}
For Types B and D, the camera orientations were randomly generated. For Types A and C, angular noise was added to the orientation/rotation matrices. In all cases, the cameras could observe the 3D point to respect cheirality. See Fig.~\ref{fig:camera_distrib} for sample data instances generated. Type A simulates a robotic exploration scenario where the robot views a scene from a directed trajectory~\cite{mur2015probabilistic}, Type B simulates large-scale 3D reconstruction from crowd-sourced images~\cite{snavely2008modeling}, Type C simulates the usage of a rotating platform for 3D modelling~\cite{web:3dscan}, and Type D simulates large-scale 3D reconstruction using stereo cameras~\cite{alcantarilla2013large}.

\begin{figure*}[ht]
	\centering     
	\subfigure[Type A]{\label{fig:random}\includegraphics[height=0.22\textwidth]{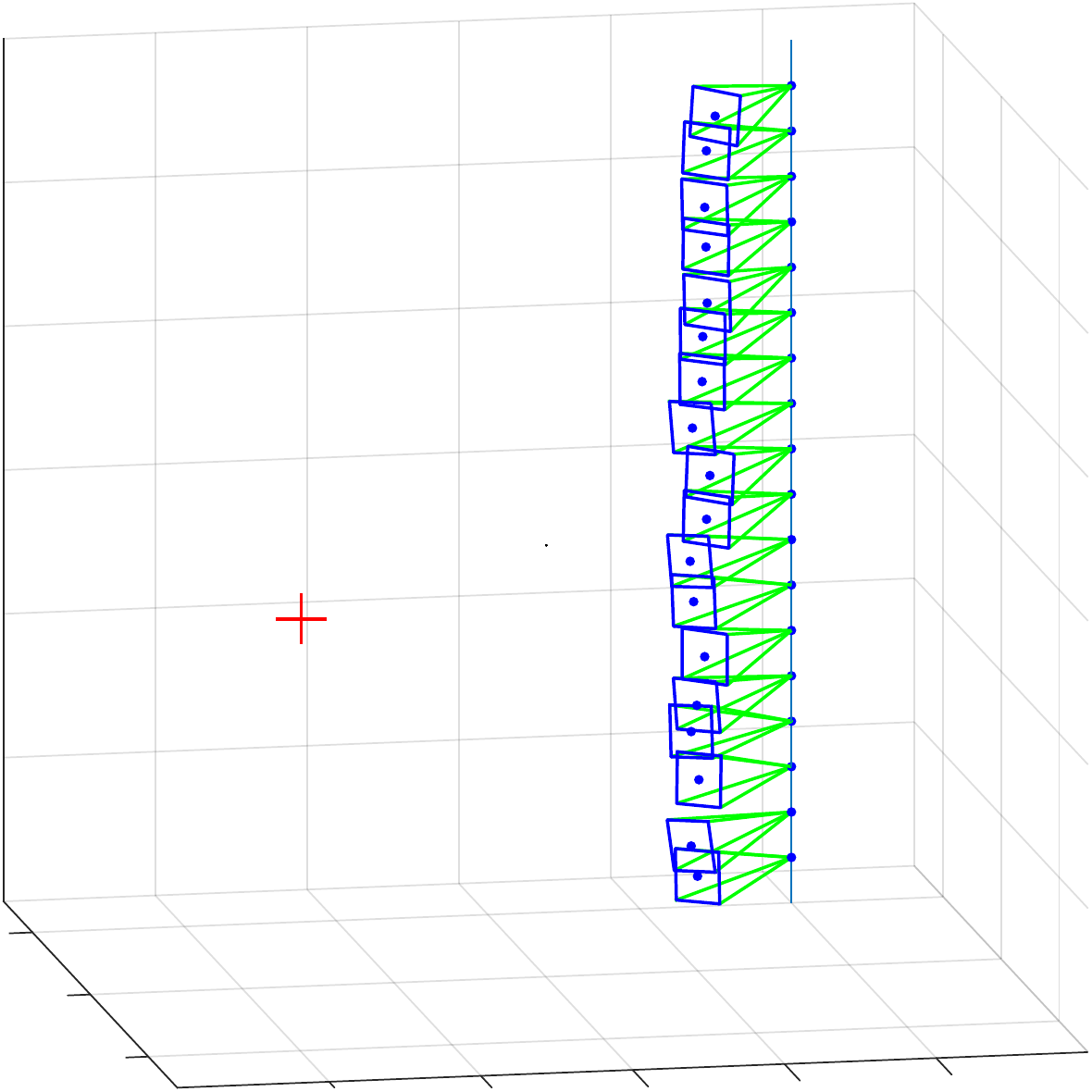}}\hspace{1em}
	\subfigure[Type B]{\label{fig:line}\includegraphics[height=0.22\textwidth]{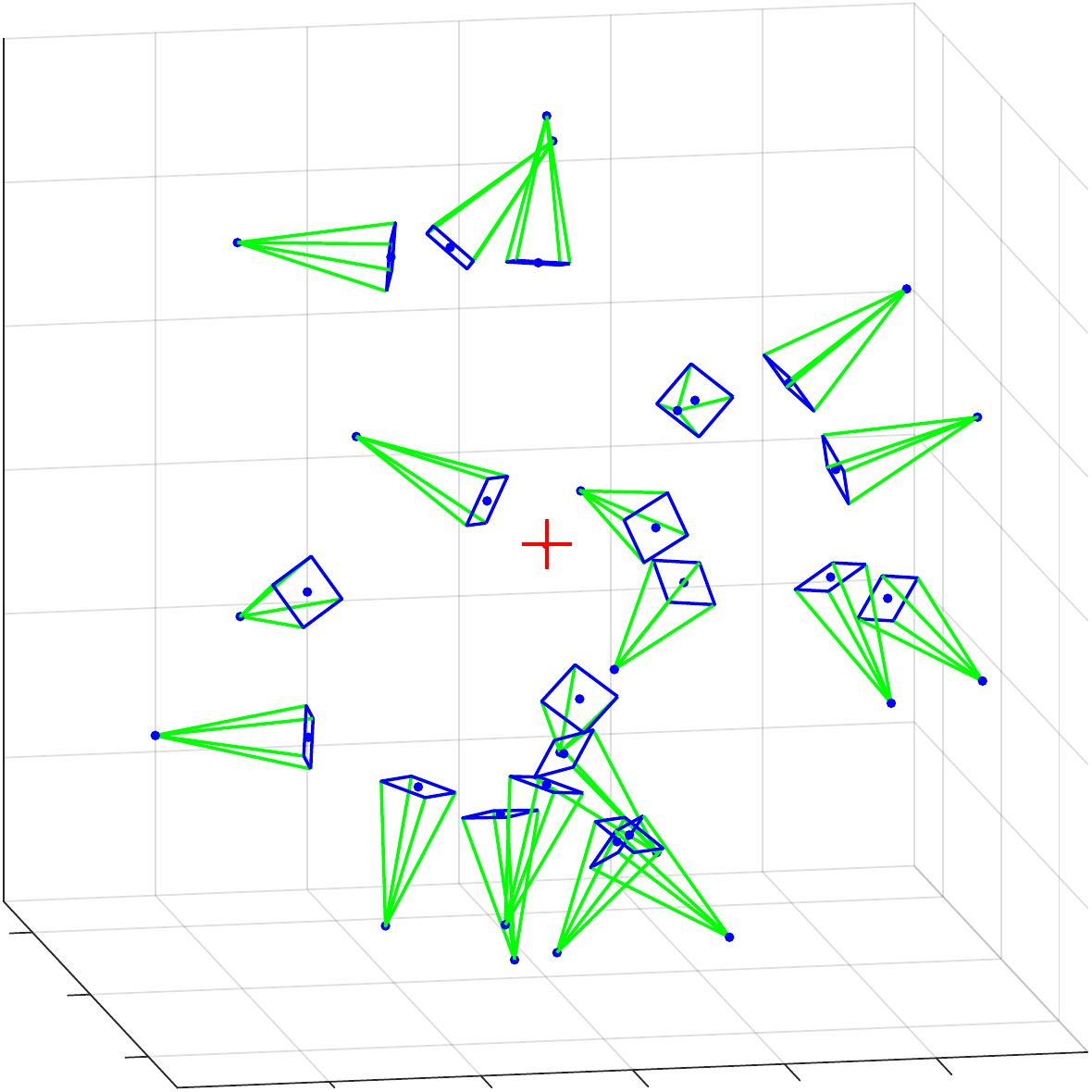}}\hspace{1em}
	\subfigure[Type C]{\label{fig:circle}\includegraphics[height=0.22\textwidth]{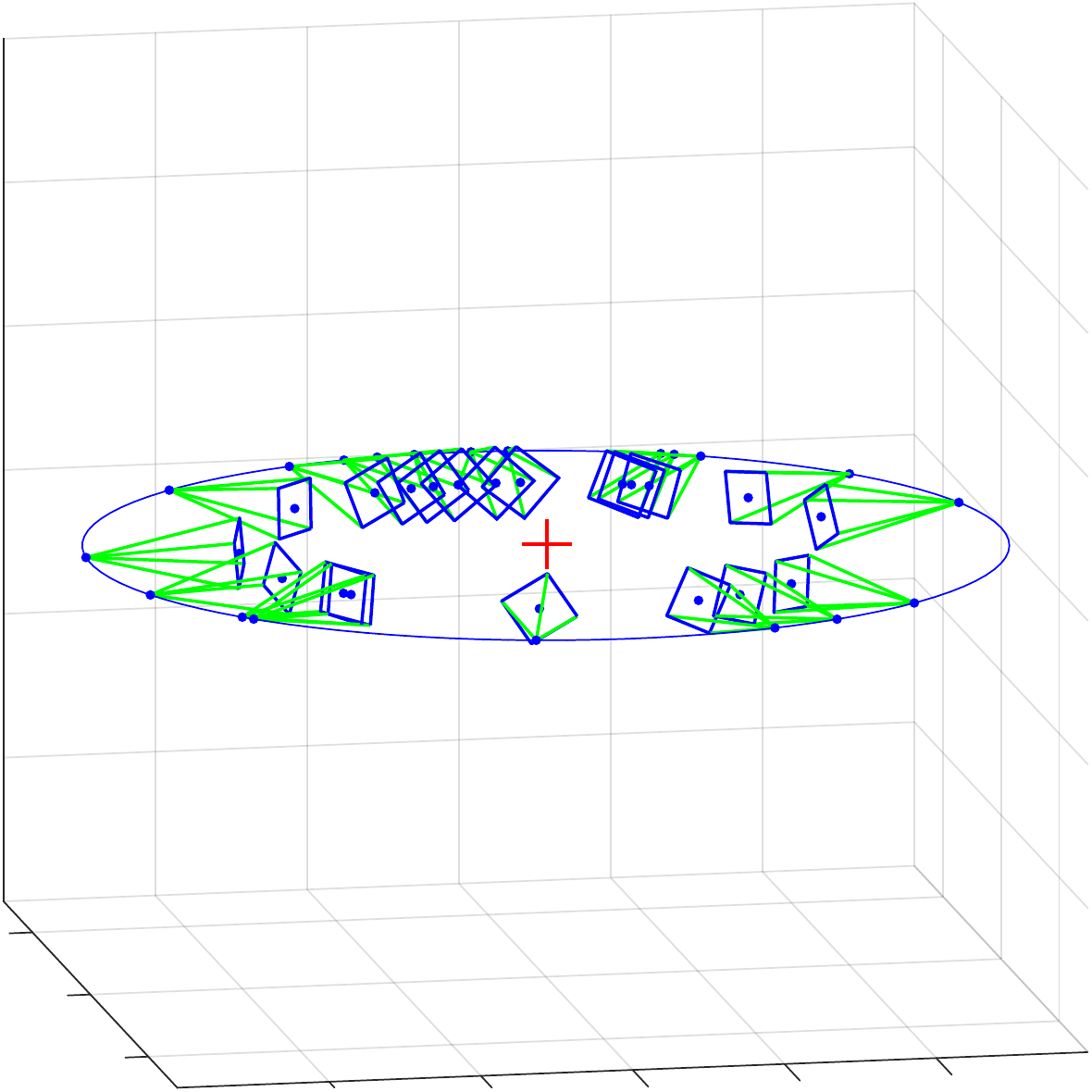}}\hspace{1em}
	\subfigure[Type D]{\label{fig:centred}\includegraphics[height=0.22\textwidth]{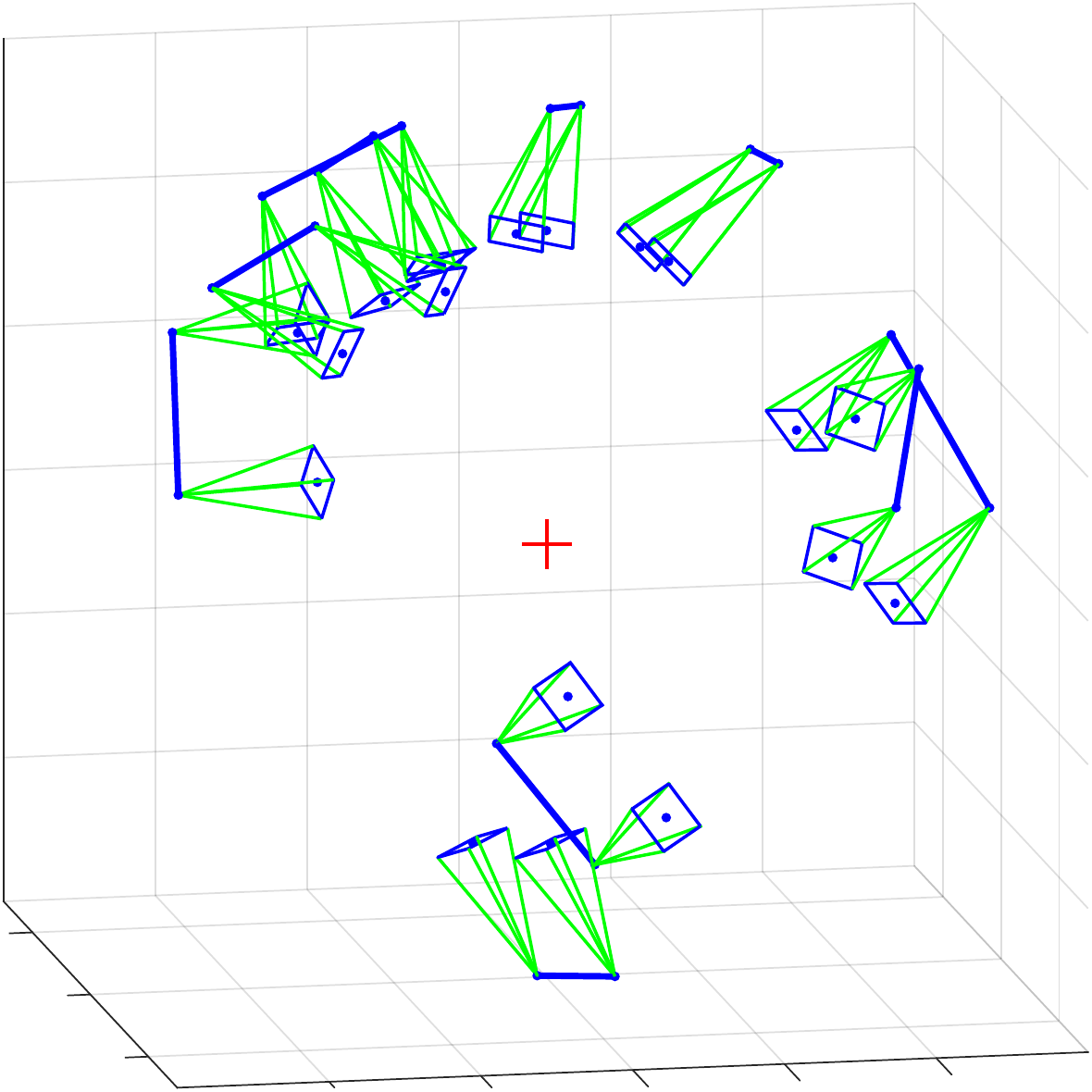}}
	\caption{The four types of synthetically generated triangulation instances. Type A: Camera positions are on a straight line; Type B: Camera positions are randomly distributed; Type C: Camera positions are on a circle; Type D: Stereo cameras with fixed baselines and positions randomly distributed. In all cases, the cameras are roughly oriented towards the 3D point to respect cheirality.}
	\label{fig:camera_distrib}
\end{figure*}

For each camera distribution, $N$ image measurement/camera matrix pairs $\{ \bP_i,\bu_i \}^{N}_{i=1}$ are generated by projecting the 3D point onto each camera and corrupting the projected point with Gaussian noise $\sigma$ = 10 pixels.

\subsection{Validation of Approximation Accuracy}

To experimentally validate Theorem~\ref{thm:coreset}, one instance of each of the camera distribution type in Fig.~\ref{fig:camera_distrib} with $N = 100$ views were generated. For each camera distribution, $200$ 3D scene points were created (in a way that the 3D points are observable in all $N$ cameras) and projected onto the $N$ cameras. This created a total of $200$ triangulation instances~\eqref{equ:triang}.

On each triangulation instance, we executed Algorithm~\ref{alg:coreset}. The approximation error ratio
\begin{align}\label{equ:err_ratio}
\frac{\max_{i\in \cX}~r(\bx_{s} \mid \bP_i, \bu_i)}{\delta^*}
\end{align}
achieved by the current estimate $\bx_s$ at each $t$ is plotted; these are shown as red curves in Fig.~\ref{fig:bnd_wrt_t} (one curve for each triangulation instance). Note that the horizontal axis begins at $t = 2$, since the bound~\eqref{equ:boundloose} is not guaranteed to hold for the initial coreset $\cC_1$. Note also that due to the allowance of ``skipping" by inserting Algorithm~\ref{alg:coreset_excerpt} into Algorithm~\ref{alg:coreset}, each $t$ can involve several updates to $\bx_s$; in Fig.~\ref{fig:bnd_wrt_t} we plotted the error ratio pertaining to final $\bx_s$ in each $t$.

\begin{figure}[h!]\centering     
	\includegraphics[width=0.7\columnwidth]{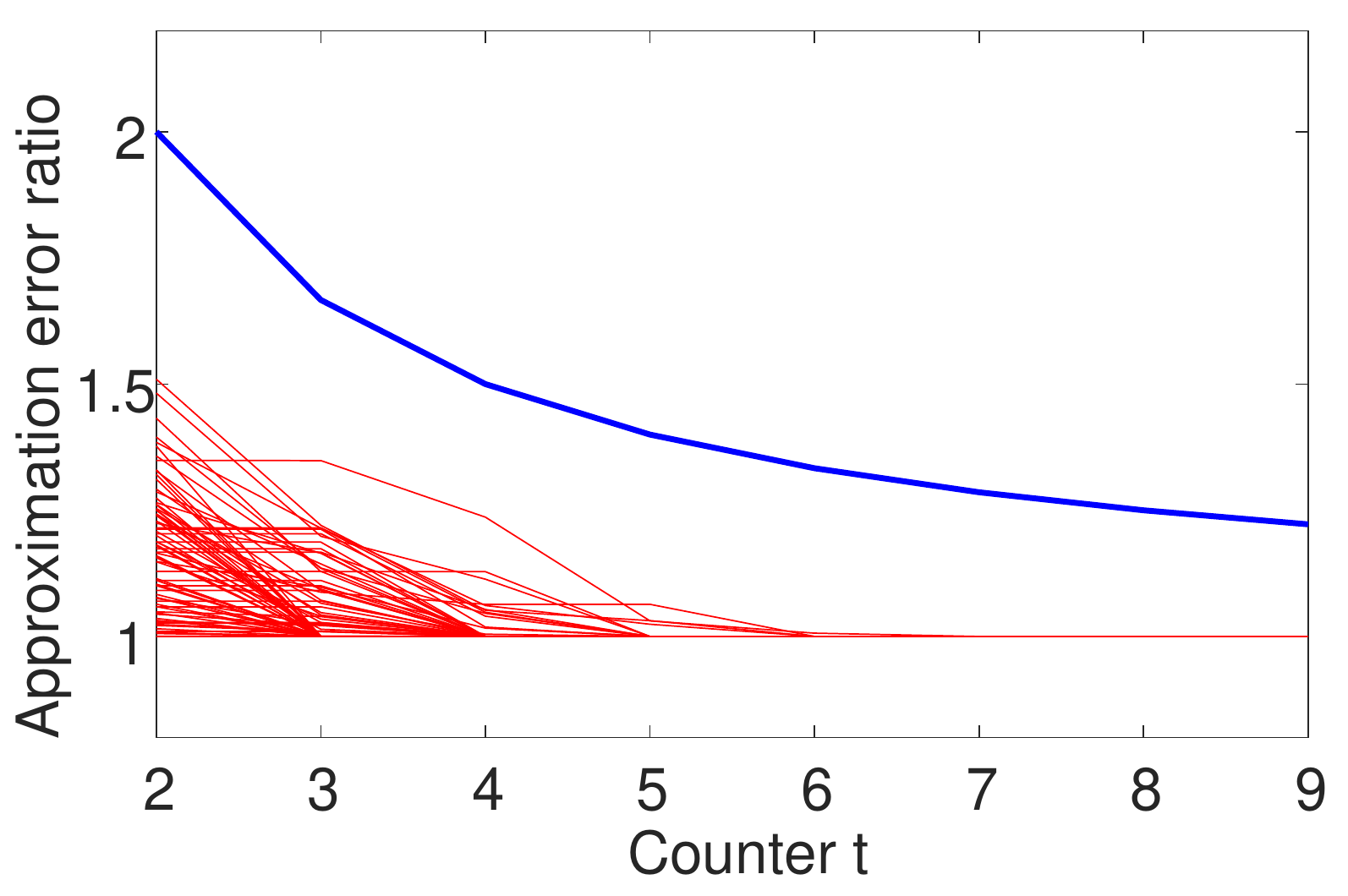}
	\caption{Approximation error ratios~\eqref{equ:err_ratio} across counter $t$ (red curves) for $200$ synthetically generated triangulation instances. The blue curve is the upper bound on the error ratio ($1 + 2/t$) as predicted by Theorem~\ref{thm:coreset}.}
\label{fig:bnd_wrt_t}	
\end{figure}

By the backtracking formula~\eqref{equ:boundloose}, the approximation error ratios should lie below the curve
\begin{align}\label{equ:bnd_ratio}
1 + 2/t;
\end{align}
this curve is plotted in blue in Fig.~\ref{fig:bnd_wrt_t}. As predicted, the bound is respected across all $t$. In Sec.~\ref{sec:results}, we will further validate Theorem~\ref{thm:coreset} using real image data.


\subsubsection{Illustrating Effect of Condition~\ref{cond2}}

The effect of Condition~\ref{cond2} on Algorithm~\ref{alg:coreset} is demonstrated in Fig.~\ref{fig:bnd_all}. On one of the synthetic triangulation instances, Fig.~\ref{fig:bndI} plots the approximation error ratio~\eqref{equ:err_ratio} against the coreset size, as the coreset is being accumulated in Algorithm~\ref{alg:coreset} (the horizontal axis thus begins at $5$ since the initial coreset $\cC_1$ has size $4$). The effects of skipping on the ratio bound~\eqref{equ:bnd_ratio} to account for Condition~\ref{cond2} is also shown; specifically, as the coreset is increased from size $8$ to $9$, Condition~\ref{cond2} occured and $t$ was not incremented. Hence the bound~\eqref{equ:bnd_ratio} does not decrease between these two steps.

Fig.~\ref{fig:bndI2} plots the approximation error ratios for all the synthetic triangulation instances from Fig.~\ref{fig:bnd_wrt_t} against coreset size. In this figure, since Condition~\ref{cond2} occurred at different iterations for the respective problem instances, the bounding curve is not plotted. Observe that all the problem instances converged to a coreset size that is not very much larger than the value of $t$ at the time of convergence, cf.~Fig.~\ref{fig:bnd_wrt_t}. 

\begin{figure}[ht!]\centering     
\subfigure[]{\includegraphics[width=0.7\columnwidth]{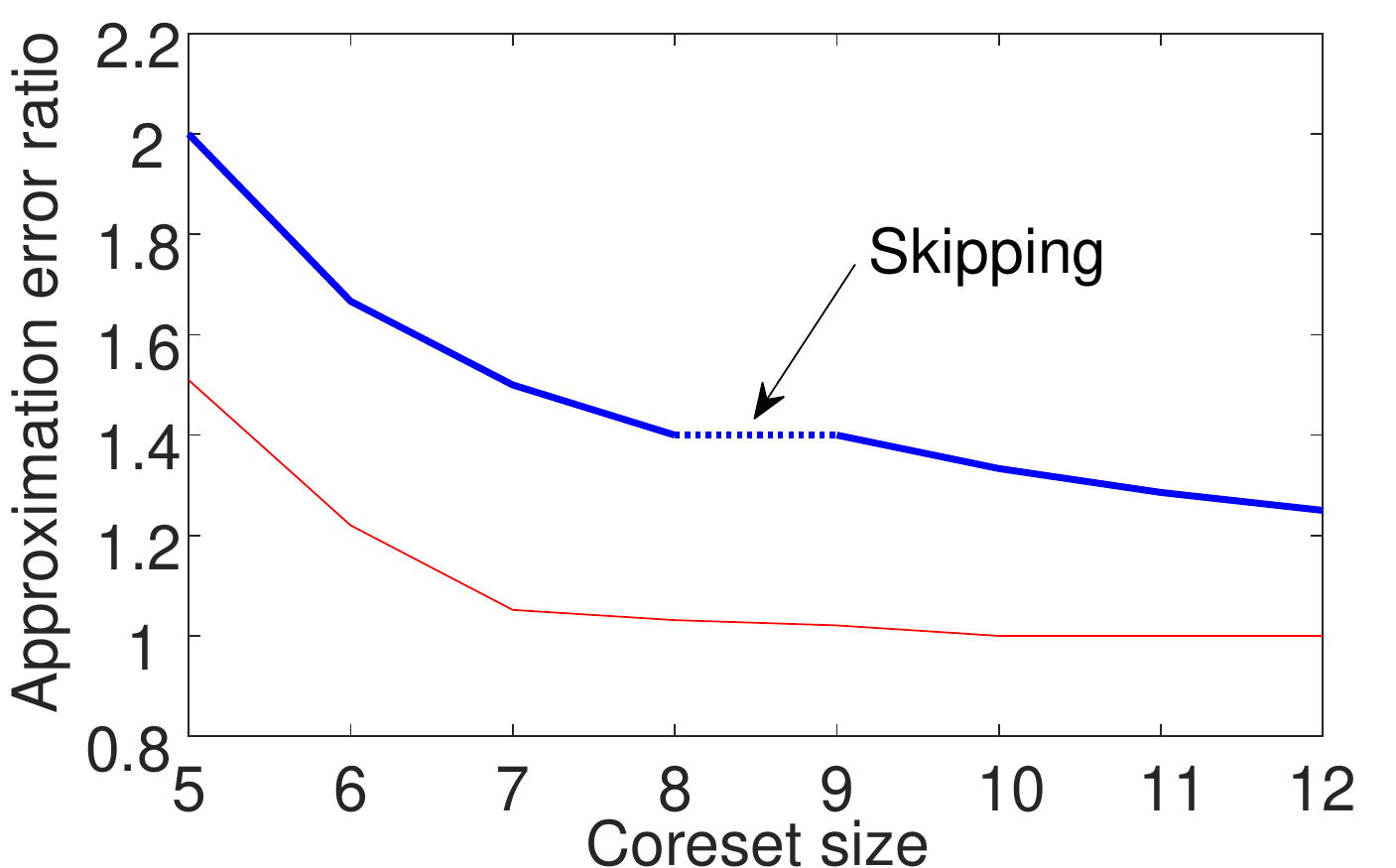}\label{fig:bndI}}
\subfigure[]{\includegraphics[width=0.7\columnwidth]{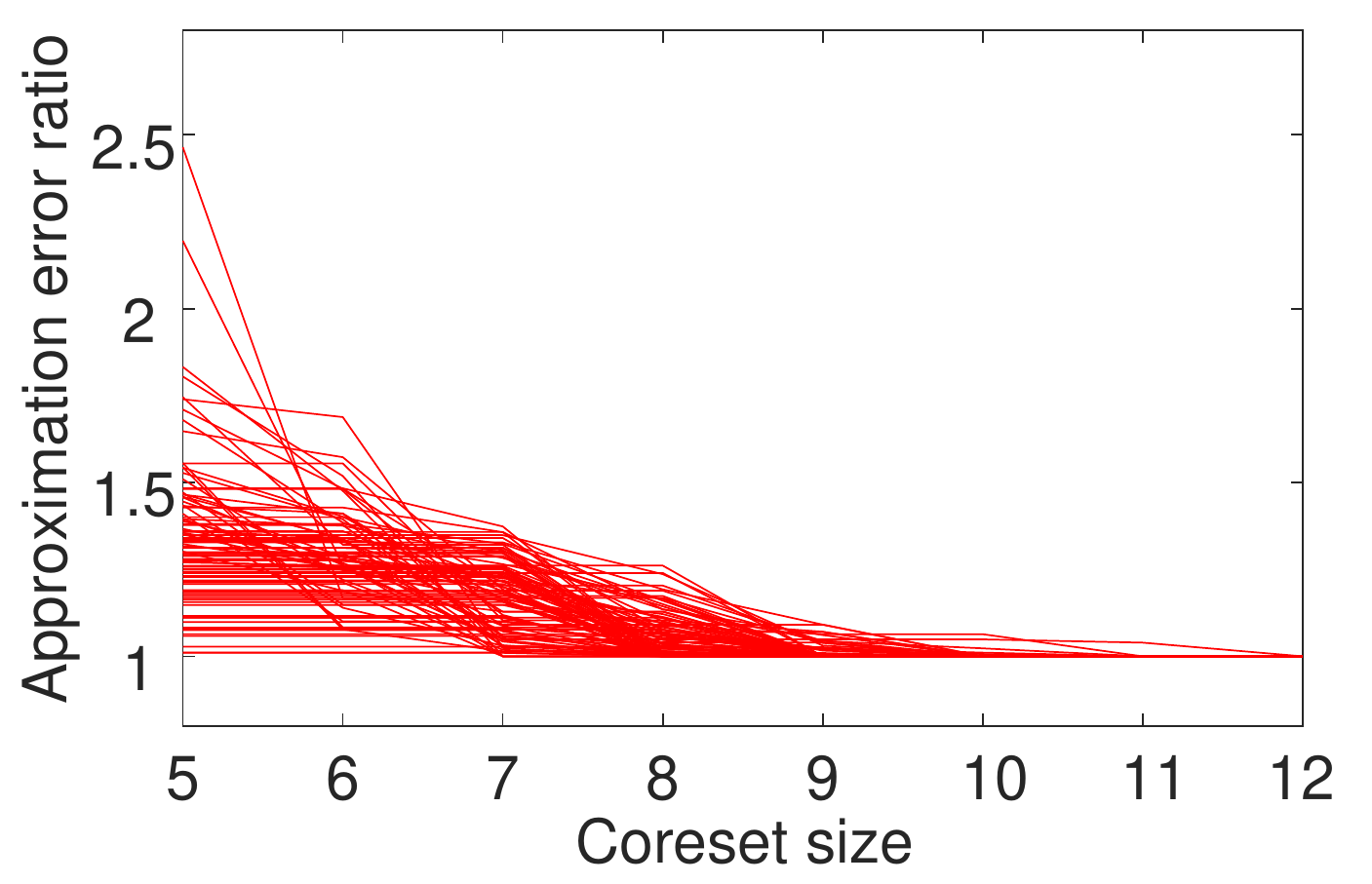}\label{fig:bndI2}}
\caption{(a) Approximation error ratio~\eqref{equ:err_ratio} plotted against coreset size for one of the synthetic data instances. In this instance, Condition~\ref{cond2} occurred as the coreset was increased from size $8$ to $9$, thus the bound~\eqref{equ:bnd_ratio} does not decrease between these two iterations. (b) Approximation error ratio plotted against coreset size for all the synthetic data instances. Since Condition~\ref{cond2} occurred at different iterations for the respective problem instances, the bounding curve is not plotted.}
\label{fig:bnd_all}
\end{figure}


\subsubsection{Probability of Condition~\ref{cond2}}

We demonstrate that the probability $\alpha$ of the occurrence of Condition~\ref{cond2} in any problem instance is mainly affected by the way the camera poses are distributed, and is not a factor of problem size $N$. For a complete execution of Algorithm~\ref{alg:coreset} on the following particular problem instances, we obtain $\alpha$ empirically as the ratio of the number of occurrence of Condition~\ref{cond2} over the effective number of iterations in the algorithm.

The experimental settings were as follows: for each type of camera distribution, $200$ 3D points were randomly generated and projected onto $N$ views, where $N$ was varied from $100$ to $10,000$. On each problem instance, Algorithm~\ref{alg:coreset} was executed $20$ times (with random initialisations) and the $\alpha$ values were recorded. Fig.~\ref{fig:alpha} shows the average $\alpha$ over all instances as a function of $N$. Evidently $\alpha$ is almost constant across $N$, and the biggest factor in the difference in $\alpha$ is the type of camera pose distribution (the curves of Types B and D are similar since they are essentially randomly distributed camera poses). This supports the analysis in Sec.~\ref{sec:sizeofcoreset} that the total runtime of Algorithm~\ref{alg:coreset}, and hence the size of the output coreset, is mainly dependent on the desired approximation factor $\epsilon$.

The next section will further investigate the size of the coreset output by Algorithm~\ref{alg:coreset}.


\begin{figure}[ht]\centering     
	\includegraphics[width=0.7\columnwidth]{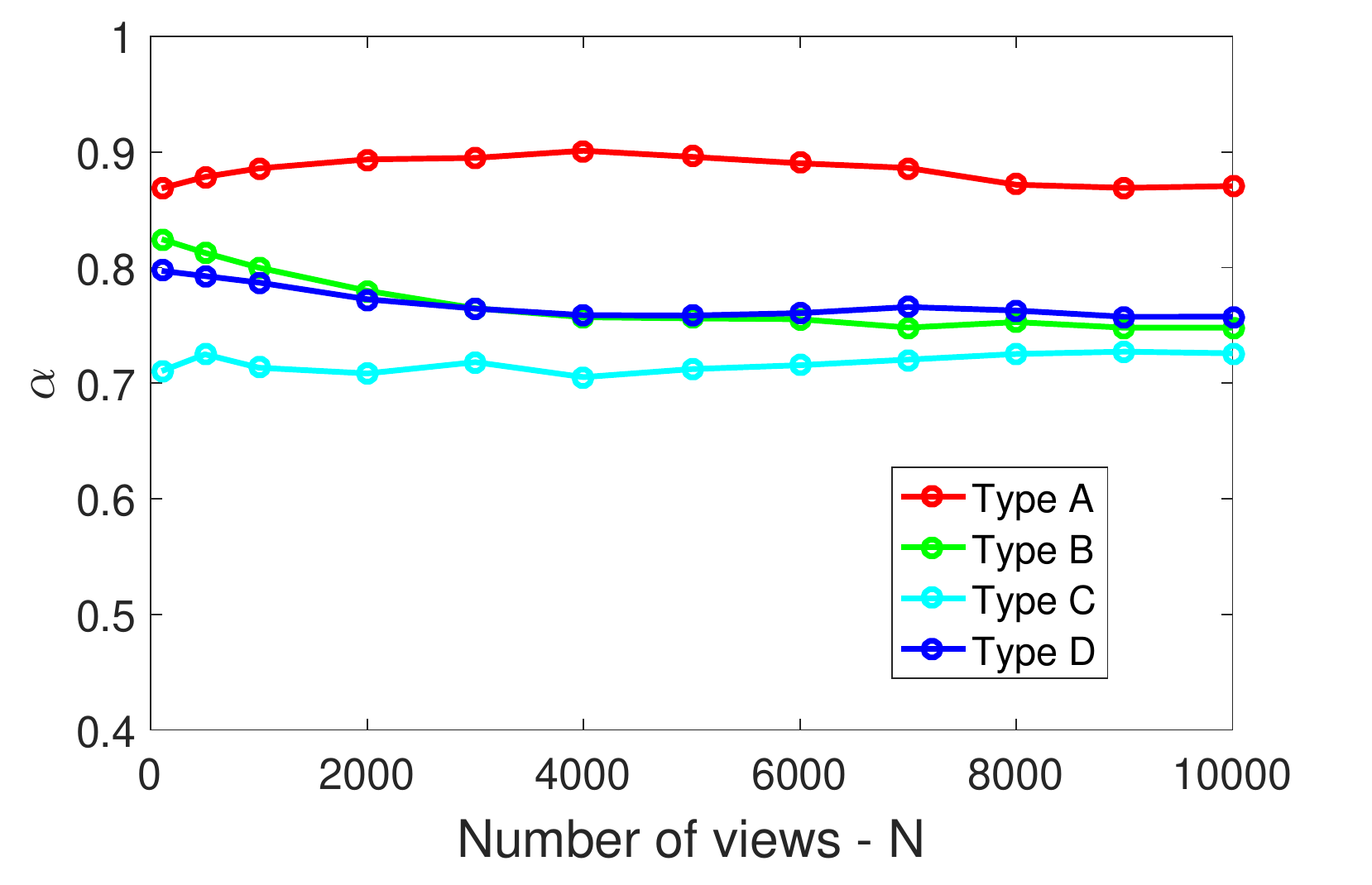}
	\caption{Average $\alpha$ (probability of occurrence of Condition~\ref{cond2}) as the problem size $N$ increases, separated according to the type of camera pose distribution (see Fig.~\ref{fig:camera_distrib}).}
	\label{fig:alpha}
\end{figure}

\subsection{Size of Output Coreset}\label{sec:exp_coreset_size}

To investigate the size of the output coreset produced by Algorithm~\ref{alg:coreset} as a function of the approximation error $\epsilon$, we generated triangulation instances for the four types of camera distribution, with $200$ 3D points in each instance but with varying problem size (number of views) $N \in \{100, 500, 1000, 5000\}$. On each instance, the setting of $\epsilon$ for Algorithm~\ref{alg:coreset} was varied decreasingly from $1$ and the size of the output coreset was recorded.

\begin{figure*}[h]\centering     
	\subfigure[$N$=100]{\includegraphics[width=0.7\columnwidth]{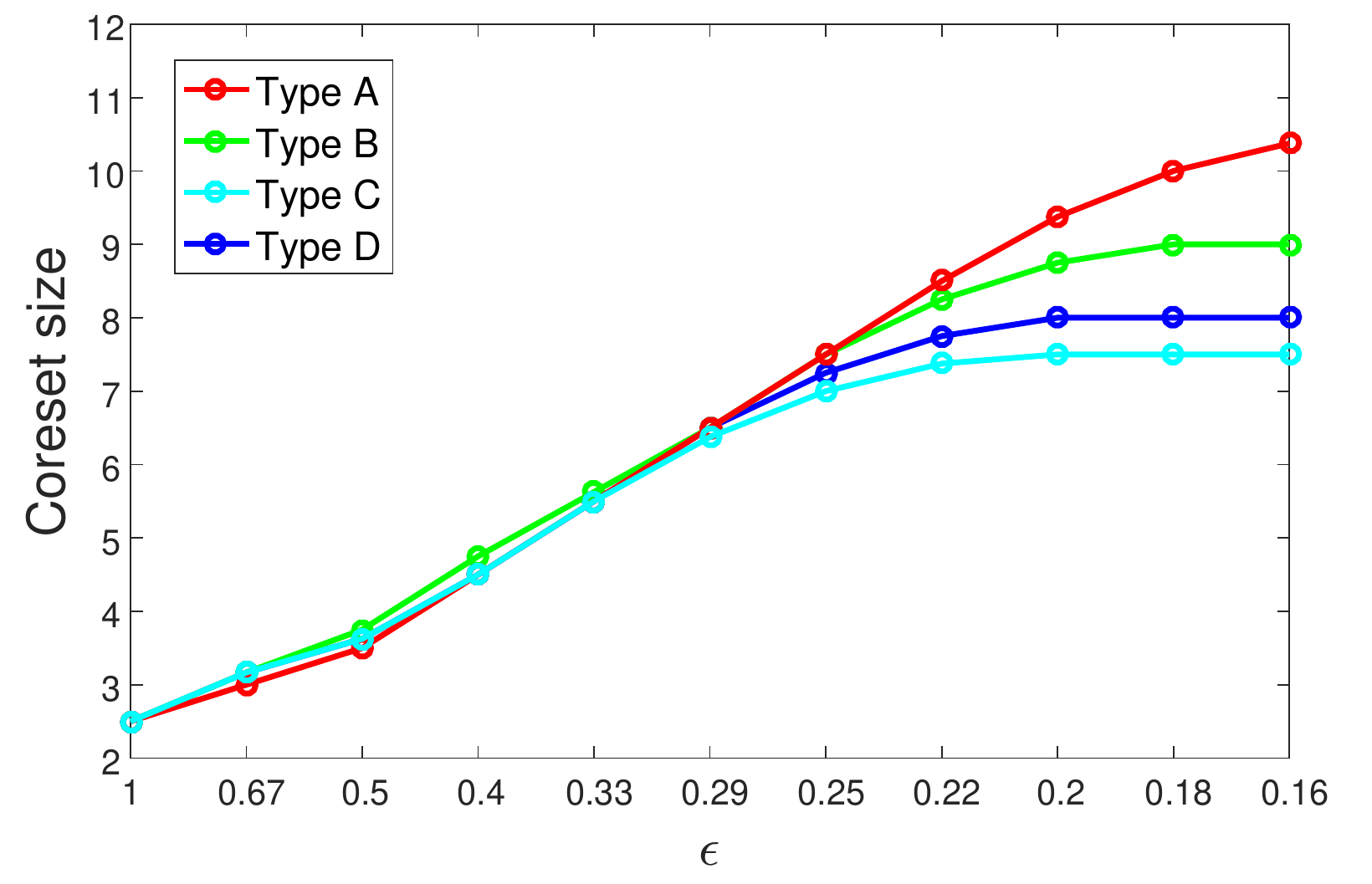}}\hspace{2em}
	\subfigure[$N$=500]{\includegraphics[width=0.7\columnwidth]{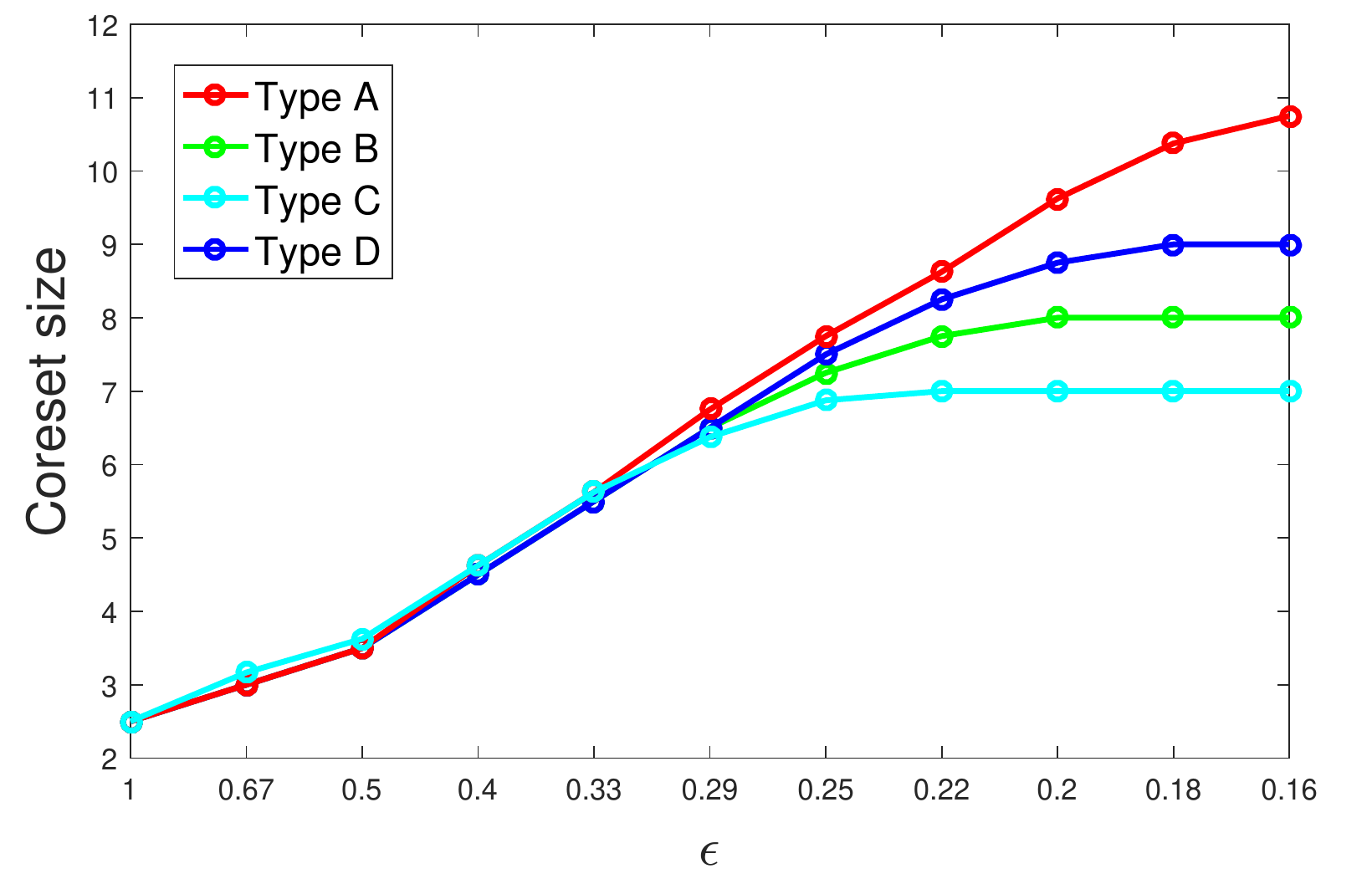}}
    \subfigure[$N$=1000]{\includegraphics[width=0.7\columnwidth]{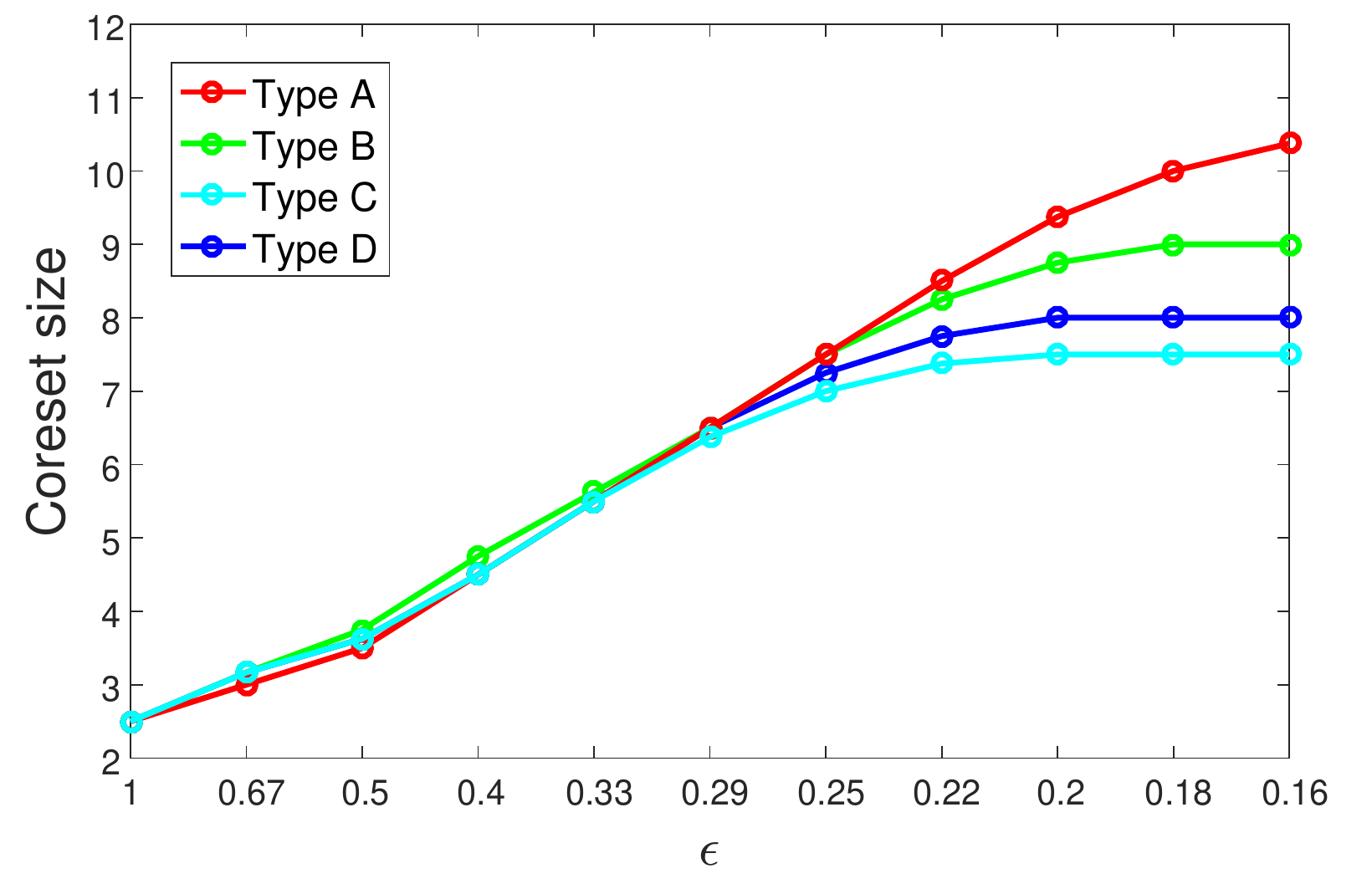}}\hspace{2em}
	\subfigure[$N$=5000]{\includegraphics[width=0.7\columnwidth]{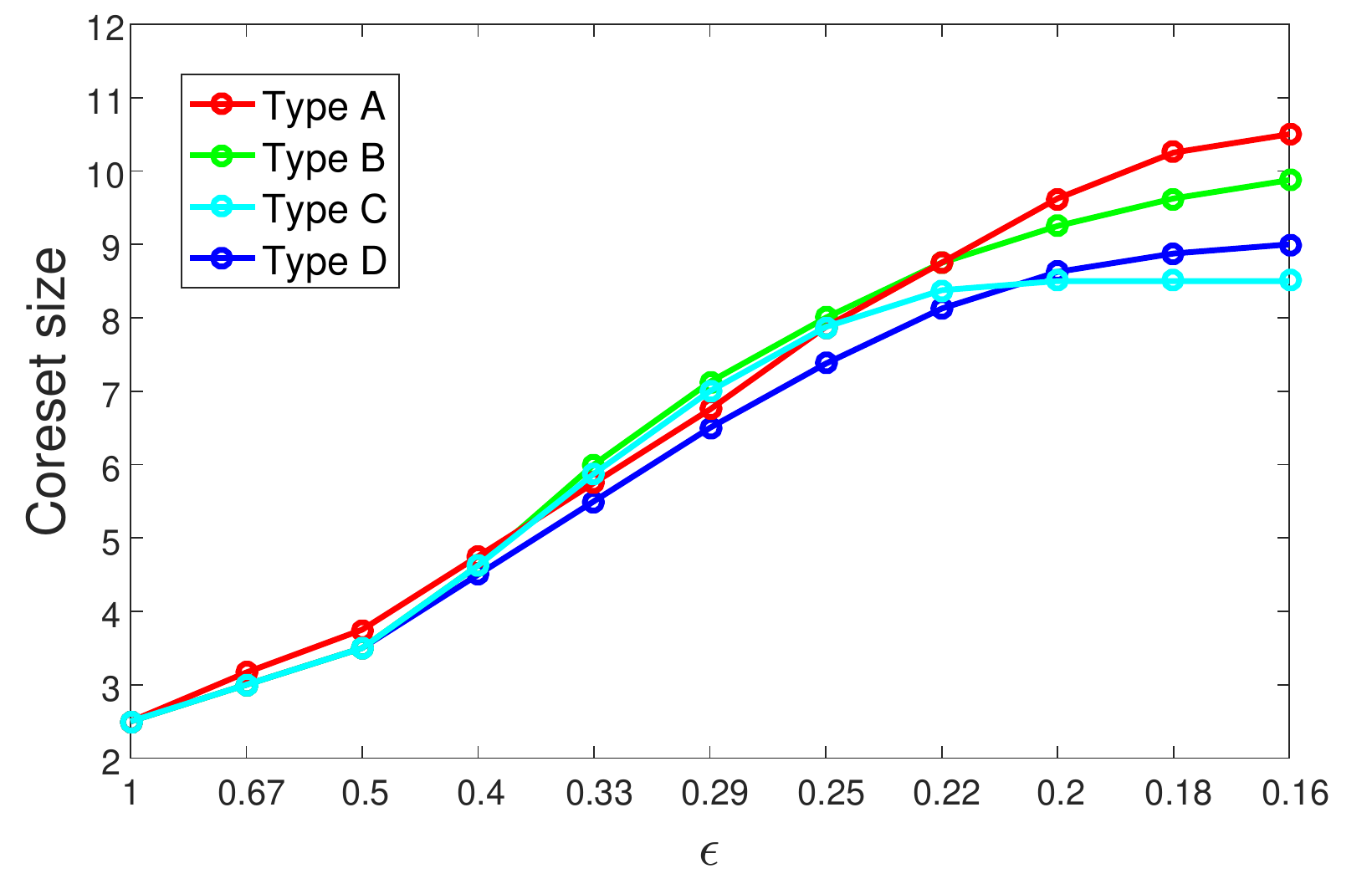}}
	\caption{Average size of coreset plotted against decreasing $\epsilon$, separated according to the four types of camera distribution as in Fig.~\ref{fig:camera_distrib}. Panel (a),(b),(c), and (d) are the results respectively for $N$ = $ 100 $, $ 500 $, $ 1000 $, and $5000$. }
	\label{fig:epsilon}
\end{figure*}

Fig.~\ref{fig:epsilon} plots the maximum coreset size as a function of $\epsilon$, averaged across all instances, but separated according to type of camera pose distributions and problem size $N$. The results show that the coreset size depends mainly on $\epsilon$ and is independent of the problem size $N$. Further, for all the data settings/parameters, the coreset size did not exceed $12$. For some of the types of distribution, the coreset size also converged earlier due to earlier global convergence.

\section{Experiments on Real Data}\label{sec:results}

We conducted experiments on real data to validate Theorem~\ref{thm:coreset} and investigate the performance of Algorithm~\ref{alg:coreset} for $\ell_\infty$ triangulation. We used a standard machine with 3.2 GHz processor and 16 GB main memory.

\subsection{Datasets and Initialisation}

We tested on publicly available datasets for large scale 3D reconstruction, namely, Vercingetorix Statue, Stockholm City Hall, Arc of Triumph, Alcatraz, \"{O}rebro Castle~\cite{enqvist2011stable,olsson2011stable}, and Notre Dame~\cite{snavely2008modeling}. The \emph{a priori} estimated camera poses and intrinsics supplied with these datasets were used to derive camera matrices. For triangulation, the size of an instance is the number of observations of the target 3D point. To avoid excessive runtimes, we randomly sampled 10\% of the scene points in each dataset - this reduces the number of problem instances, but not the size of each of the selected instances. A histogram of the problem sizes for each of the above datasets are shown in the top left panel of Figs.~\ref{fig:Vercingetorix} to~\ref{fig:NotreDame}.

Our coreset method was initialised as shown in the first few steps in Algorithm~\ref{alg:coreset}, which amounts to randomly choosing four data to instantiate $\bx_1$ by solving~\eqref{equ:triang}. For any other algorithm that requires initialisation, the same $\bx_1$ or its current maximum reprojection error were provided as the initial estimates.

\subsection{Validation of Approximation Accuracy}

The top right panel in Figs.~\ref{fig:Vercingetorix} to~\ref{fig:NotreDame} show the actual ratio of errors versus the predicted ratio (using the backtracking formula in Sec.~\ref{sec:backtrack}) across the iterations of Algorithm~\ref{alg:coreset} for all problem instances in the datasets. Again, the results confirm the validity of Theorem~\ref{thm:coreset}.

\subsection{Relative Speed-up of Coreset over Batch}\label{sec:realresults}

Here we investigate the practicality of Algorithm~\ref{alg:coreset} as a \emph{global optimiser} for $\ell_\infty$ triangulation. As described in Section~\ref{sec:coreset}, Algorithm~\ref{alg:coreset} is a meta-algorithm which requires a sub-routine to solve~\eqref{equ:triang} on the subset indexed by $\cC_t$. We thus compared running Algorithm~\ref{alg:coreset} with a specific solver as a sub-routine, and the direct execution of the same solver in ``batch mode" on the whole data. Since the runtime of Algorithm~\ref{alg:coreset} depends on the efficiency of the embedded solver, the key performance indicator here is the \emph{relative speed-up} achieved by coreset over batch.

Based on the investigations in~\cite{agarwal2008fast}, we have chosen to use bisection~\cite{kahl2005multiple} and Dinkelbach's method~\cite{dinkelbach1967nonlinear} (equivalent to~\cite{olsson2007efficient}) to embed into Algorithm~\ref{alg:coreset}. Although the best performing technique in~\cite{agarwal2008fast} was Gugat's algorithm~\cite{gugat1996fast}, our experiments suggested that it did not outperform Dinkelbach's method on the triangulation problem. SeDuMi~\cite{sturm1999using} was used to solve the SOCP sub-problems in~\cite{kahl2005multiple,dinkelbach1967nonlinear}.

The bottom diagrams of Figs.~\ref{fig:Vercingetorix} to~\ref{fig:NotreDame} show the average runtime of coreset and batch as a function of problem size (number of views), for each respective $\ell_\infty$ solver. Observe that the runtime of batch increased linearly and then exponentially, whilst coreset exhibited almost constant runtime---the latter observation is not surprising, since Algorithm~\ref{alg:coreset} usually terminated at $\le 10$ iterations regardless of the problem size at shown in the top right panel.

Of course, on all of the datasets, most of the triangulation instances are small, as shown in the histogram at the upper left panel of Figs.~\ref{fig:Vercingetorix} to~\ref{fig:NotreDame}. However, in these datasets, there are sufficient numbers of moderate to large problem instances, such that the total runtime of coreset is still much smaller than then total runtime of batch. Table~\ref{table:l2} shows the total and average runtime of the variants considered. Evidently, coreset outperformed batch in all the datasets.

\begin{table*}[ht]\centering
	\renewcommand{\arraystretch}{1.3}
	\begin{tabular}{ |c|p{10em}|p{8em}|c|c|c|c| }
		\cline{4-7}
		\multicolumn{3}{c}{} & \multicolumn{4}{|c|}{Total runtime (seconds)}  \\
		\cline{4-7}		
		\multicolumn{3}{c}{} & \multicolumn{2}{|c|}{Bisection solver} & \multicolumn{2}{|c|}{Dinkelbach solver} \\ 
		\hline
		Dataset    & Scene points (number of triang. instances) & Number of views (max. triang. size) & Batch & Coreset & Batch & Coreset     \\ \hline
		Vercingetorix     & \multicolumn{1}{|c|}{594}  & \multicolumn{1}{|c|}{68}  & 57    & \textbf{44} (23\%)  & 23    & \textbf{17} (26\%)  \\ \hline
Stockholm         & \multicolumn{1}{|c|}{2176} & \multicolumn{1}{|c|}{43}  & 258   & \textbf{177} (31\%) & 109   & \textbf{74} (32\%)   \\ \hline
Arc of Triumph    & \multicolumn{1}{|c|}{2744} & \multicolumn{1}{|c|}{173} & 607   & \textbf{243} (59\%) & 204   & \textbf{89} (56\%)   \\ \hline
Alcatraz          & \multicolumn{1}{|c|}{4431} & \multicolumn{1}{|c|}{419} & 1231  & \textbf{520} (57\%) & 452   & \textbf{239} (47\%)  \\ \hline
\"{O}rebro Castle & \multicolumn{1}{|c|}{5943} & \multicolumn{1}{|c|}{761} & 4052  & \textbf{800} (80\%) & 1440  & \textbf{351} (75\%)   \\ \hline
Notre Dame        & \multicolumn{1}{|c|}{7149} & \multicolumn{1}{|c|}{715} & 4148  & \textbf{696} (83\%) & 2399  & \textbf{582} (75\%)\\ \hline
	\end{tabular}
	\vspace{1em}
	\caption{Comparisons between coreset and batch in terms of total runtime in seconds. For coreset, the number in parentheses indicates the percentage of reduction in runtime by using the coreset method over the batch counterpart. }
	\label{table:l2}
\end{table*}

\begin{table*}[ht]\centering
	\renewcommand{\arraystretch}{1.3}
	\begin{tabular}{ |c|p{10em}|p{8em}|c|c|c|c| }
		\cline{4-7}
		\multicolumn{3}{c}{} & \multicolumn{4}{|c|}{Total runtime (seconds)} \\
		\cline{4-7}		
		\multicolumn{3}{c}{} & \multicolumn{2}{|c|}{Dinkelbach ($p = 1$)} & \multicolumn{2}{|c|}{Polyhedron ($p = \infty$)}\\
		\hline
		Dataset    & Scene points (number of triang. instances) & Number of views (max. triang. size) & Batch & Coreset & Batch & Coreset             \\ \hline		
		Vercingetorix     & \multicolumn{1}{|c|}{594}  & \multicolumn{1}{|c|}{68}  & 1.61 & \textbf{1.58 } (13\%)  & 1.82 & \textbf{1.83} (-1\%) \\ \hline
		Stockholm         & \multicolumn{1}{|c|}{2176} & \multicolumn{1}{|c|}{43}  & 10.4 & \textbf{9.87 } (19\%)  & 12.2 & \textbf{11.7} (4\% ) \\ \hline
		Arc of Triumph    & \multicolumn{1}{|c|}{2744} & \multicolumn{1}{|c|}{173} & 20.5 & \textbf{12 } (49\%)    & 23.5 & \textbf{14.1} (40\%) \\ \hline
		Alcatraz          & \multicolumn{1}{|c|}{4431} & \multicolumn{1}{|c|}{419} & 166  & \textbf{40 } (73\%)    & 146  & \textbf{28.9} (80\%) \\ \hline
		\"{O}rebro Castle & \multicolumn{1}{|c|}{5943} & \multicolumn{1}{|c|}{761} & 1011 & \textbf{56 } (94\%)    & 887  & \textbf{52.2} (94\%) \\ \hline
		Notre Dame        & \multicolumn{1}{|c|}{7149} & \multicolumn{1}{|c|}{715} & 1535 & \textbf{116} (87\%)    & 865  & \textbf{31.7} (96\%) \\ \hline
	\end{tabular}
	\vspace{1em}	
	\caption{Comparisons between coreset and batch in terms of total runtime in seconds, under the $\ell_1$ and $\ell_\infty$ reprojection error~\eqref{equ:reprojpnorm}. For $\ell_1$ reprojection error, the Dinkelbach method was used as the embedded solver in Algorithm~\ref{alg:coreset}. For $\ell_\infty$ reprojection error, the state-of-art polyhedron collapse method~\cite{donne2015point} (``Polyhedron" above) was used as the embedded solver in Algorithm~\ref{alg:coreset}. For coreset, the number in parentheses indicates the percentage of reduction in runtime by using the coreset method over the batch counterpart.}
	\label{table:lp}
\end{table*}

Although in principle any $\ell_\infty$ solver for $\ell_2$ norm reprojection error~\cite{kahl2005multiple,ke2007quasiconvex,olsson2007efficient,agarwal2008fast,dai2012novel} can be used in Algorithm~\ref{alg:coreset}, we emphasise again that the primary performance indicator here is the relative speed-up of coreset over batch. If a faster solver is used, it would likely improve both coreset and batch by the same factor.

\subsection{Extensions to $\ell_1$ and $\ell_\infty$ Reprojection Error}

In the literature, apart from the more ``traditional" $\ell_2$ norm used in the reprojection error~\eqref{equ:reproj2norm}, different $p$-norms have been considered~\cite{agarwal2008fast}, i.e.,
\begin{align}\label{equ:reprojpnorm}
r(\bx \mid \bP_i,\bu_i) =  \left\| \bu_i - \frac{\bP^{1:2}_i \tilde{\bx}}{\bP^{3}_i \tilde{\bx}} \right\|_p, \;\;\;\; p \in \{1,2,\infty \}.
\end{align}
Our coreset theory was developed based on $p = 2$, thus the approximation bound (Theorem~\ref{thm:coreset}) will not hold for other $p$.

It is nonetheless feasible to apply Algorithm~\ref{alg:coreset} as a meta-algorithm for solving $\ell_\infty$ triangulation~\eqref{equ:triang} under different reprojection errors. Since problem~\eqref{alg:coreset} remains a GLP (Sec.~\ref{sec:background}) for $p = 1$ and $p = \infty$ in~\eqref{equ:reprojpnorm}, global convergence is guaranteed. It is thus of interest to compare the relative speed-up given by the coreset method over a batch method in finding the globally optimal solution.

Table~\ref{table:lp} shows the results of repeating the experiment in Sec.~\ref{sec:realresults} with $p = 1$ and $p = \infty$. For $p = 1$, the Dinkelbach method was used as the embedded solver for Algorithm~\ref{alg:coreset}. For $p = \infty$, the state-of-the-art polyhedron collapsed method~\cite{donne2015point} was used as the embedded solver. Evidently the results show that the coreset method is able to significantly speed up global convergence.

\section{Dealing with Outliers}

While the existence of outliers in the measurements (e.g., from incorrect feature associations) is transparent to Algorithm~\ref{alg:coreset} (i.e., the error bound in Theorem~\ref{thm:coreset} will still hold), the result will of course be biased by the outliers---after all, the $\ell_\infty$ framework is not inherently robust~\cite{Kahl2008multiview}.

Nonetheless, Sim and Hartley~\cite{sim2006removing} showed how an effective outlier removal scheme can be constructed based on $\ell_\infty$ estimation. Basically, their scheme recursively conducts $\ell_\infty$ estimation and removes the support set (see Property~\ref{prop:basis}) from the input data until the maximum residual is below a pre-determined inlier threshold. The remaining data then forms an inlier set.

Here, in the context of triangulation with outliers, we showed how the efficiency of Sim and Hartley's scheme can be improved by using Algorithm~\ref{alg:coreset} (with a high $\epsilon$) as a fast $\ell_\infty$ solver. Due to the approximation by Algorithm~\ref{alg:coreset}, instead of removing the support set, we remove the $4$ measurements with the largest residuals.

Fig.~\ref{fig:outlier} compares the runtime and number of remaining inliers produced by Sim and Hartley's original scheme (with Dinkelbach's method as the $\ell_\infty$ solver) and our coreset-enabled scheme (with $\epsilon = 0.4$). The results are based on $100$ 3D scene points projected onto $N$ views ($10 \le N \le 500$), and where $90\%$ of the 2D measurements in each problem instance were corrupted with Gaussian noise of $\sigma = 5$ pixels (the inliers), while the remaining $10\%$ were corrupted with larger noise ($\sigma = 30$ pixels) to create outliers. The inlier threshold was set to $10$ pixels. Evidently our coreset modification significantly improved the efficiency of the original scheme, without significantly affecting the quality of the result (number of remaining inliers).

\begin{figure}[h]\centering
	\subfigure[]{\includegraphics[width=0.71\columnwidth]{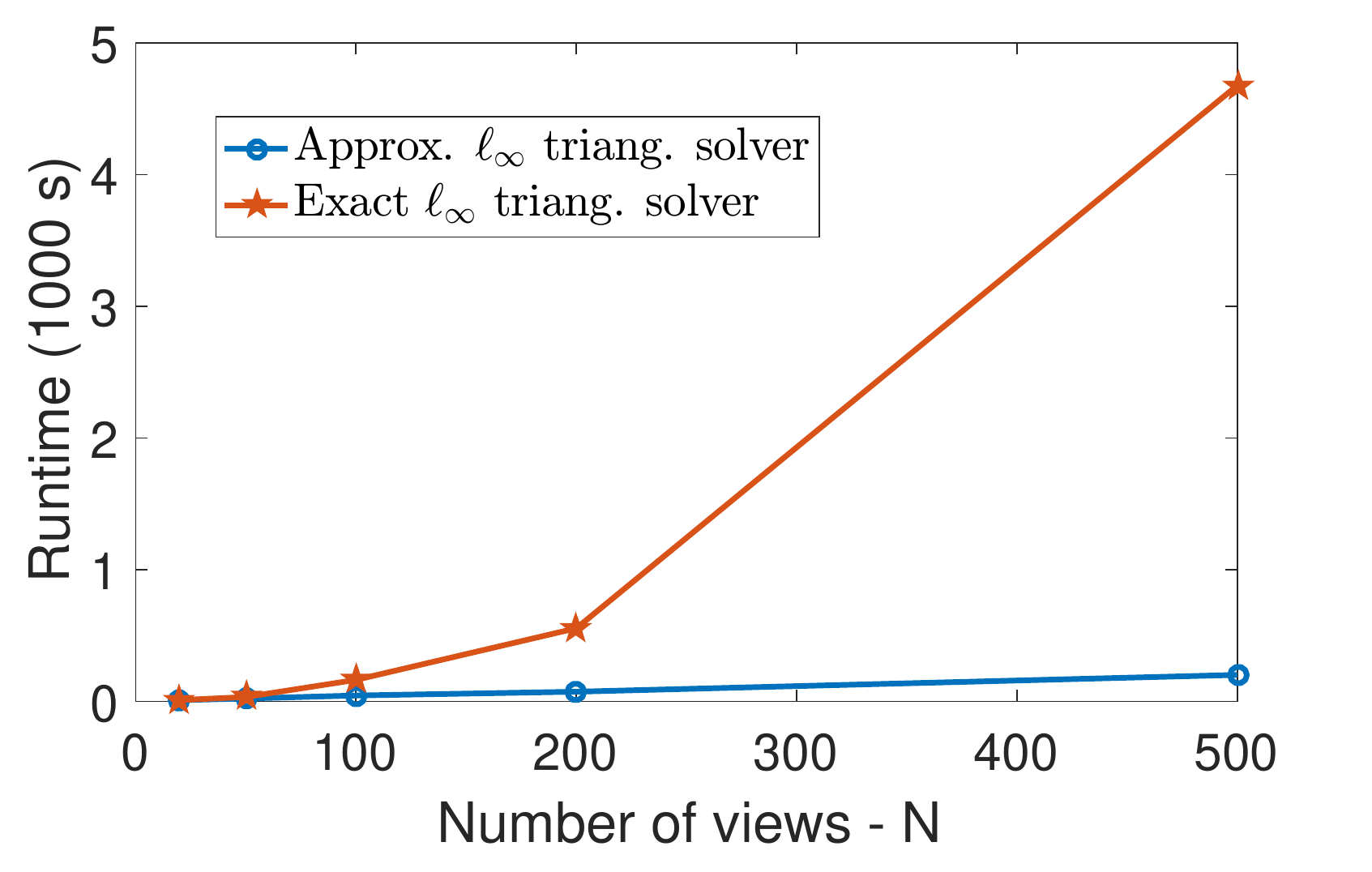}}\label{fig:outlier_time}
	\subfigure[]{\includegraphics[width=0.71\columnwidth]{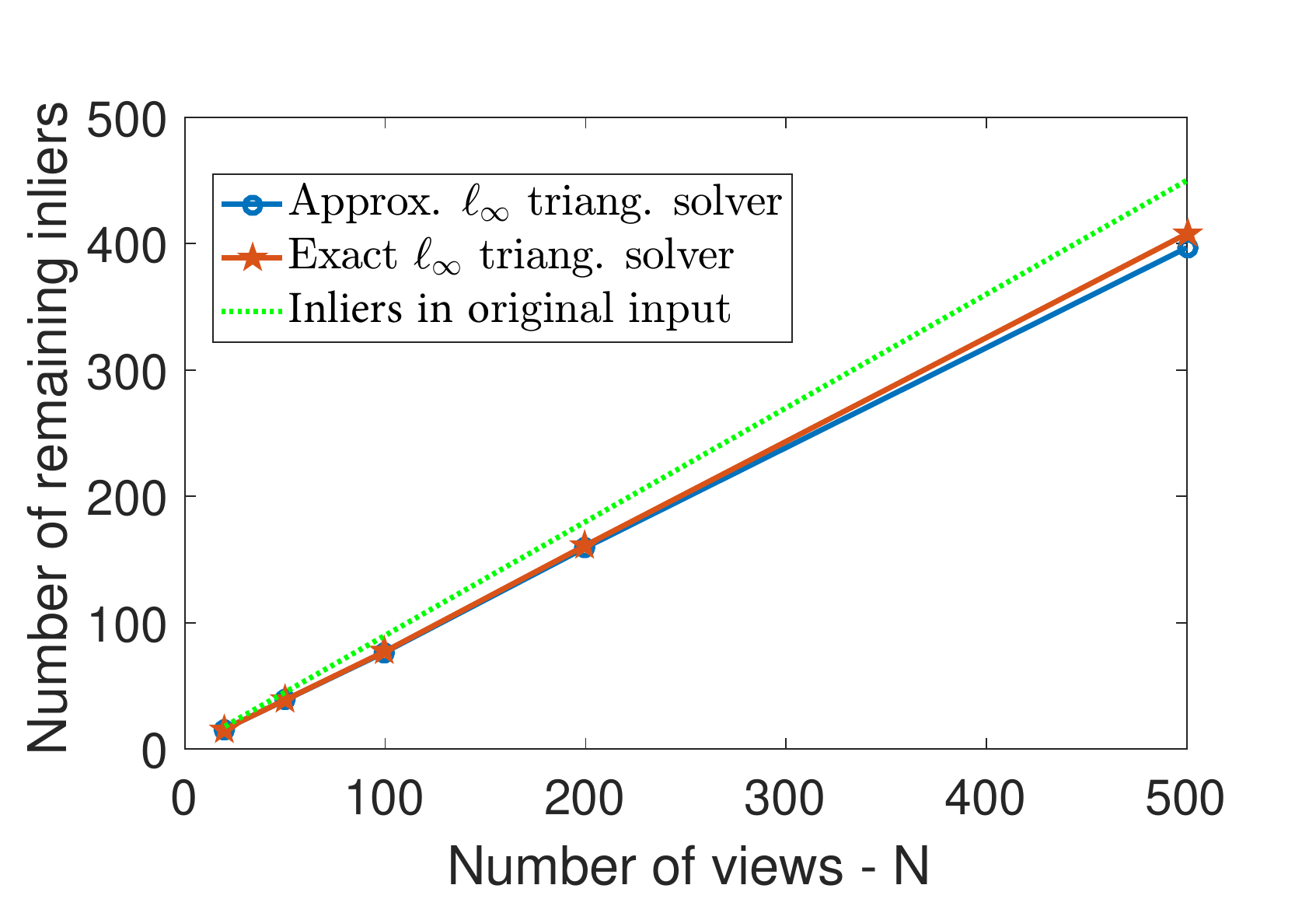}}\label{fig:outlier_in}
	\caption{Comparing Sim and Hartley's outlier removal scheme~\cite{sim2006removing} with an exact solver (Dinkelbach's method) and with Algorithm~\ref{alg:coreset}. (a) Average runtime; (b) Number of inliers remaining in the final inlier set (the dashed green line to indicate the number of inliers in original input).\label{fig:outlier}}
\end{figure}



\section{Conclusions and Open Questions}

In this paper, we show that $\ell_\infty$ triangulation admits coreset approximation. We also provided comprehensive experimental results that establish the practical value of the coreset algorithm on large scale 3D reconstruction datasets. 

There are several open questions:
\begin{itemize}
\item A deeper analysis of Condition~\ref{cond2} to hopefully remove it from consideration in the coreset bound, or at least to better characterise and  predict its occurrence.
\item The proof in Sec.~\ref{sec:approx} was inspired by the work of~\cite{buadoiu2008optimal} on minimum enclosing ball (MEB) problems. There, the coreset size bound $\lceil 2/\epsilon \rceil$ was proven to be tight if the dimensionality $d$ is comparable to $1/\epsilon$. For lower dimensional MEBs,  tighter bounds have been proposed. It would be of interest to construct such tighter bounds for $\ell_\infty$ triangulation $(d = 3)$.
\end{itemize}

Last but not least, we hope that our work encourages more effort to seek theoretically justifiable approximate algorithms for large scale geometric computer vision problems, especially the class of problems surveyed in~\cite{kahl2005multiple}, which can be seen as GLPs~\cite{amenta1994helly}.

\clearpage

\begin{figure*}[h]\centering
	\subfigure{\includegraphics[width=0.38\textwidth]{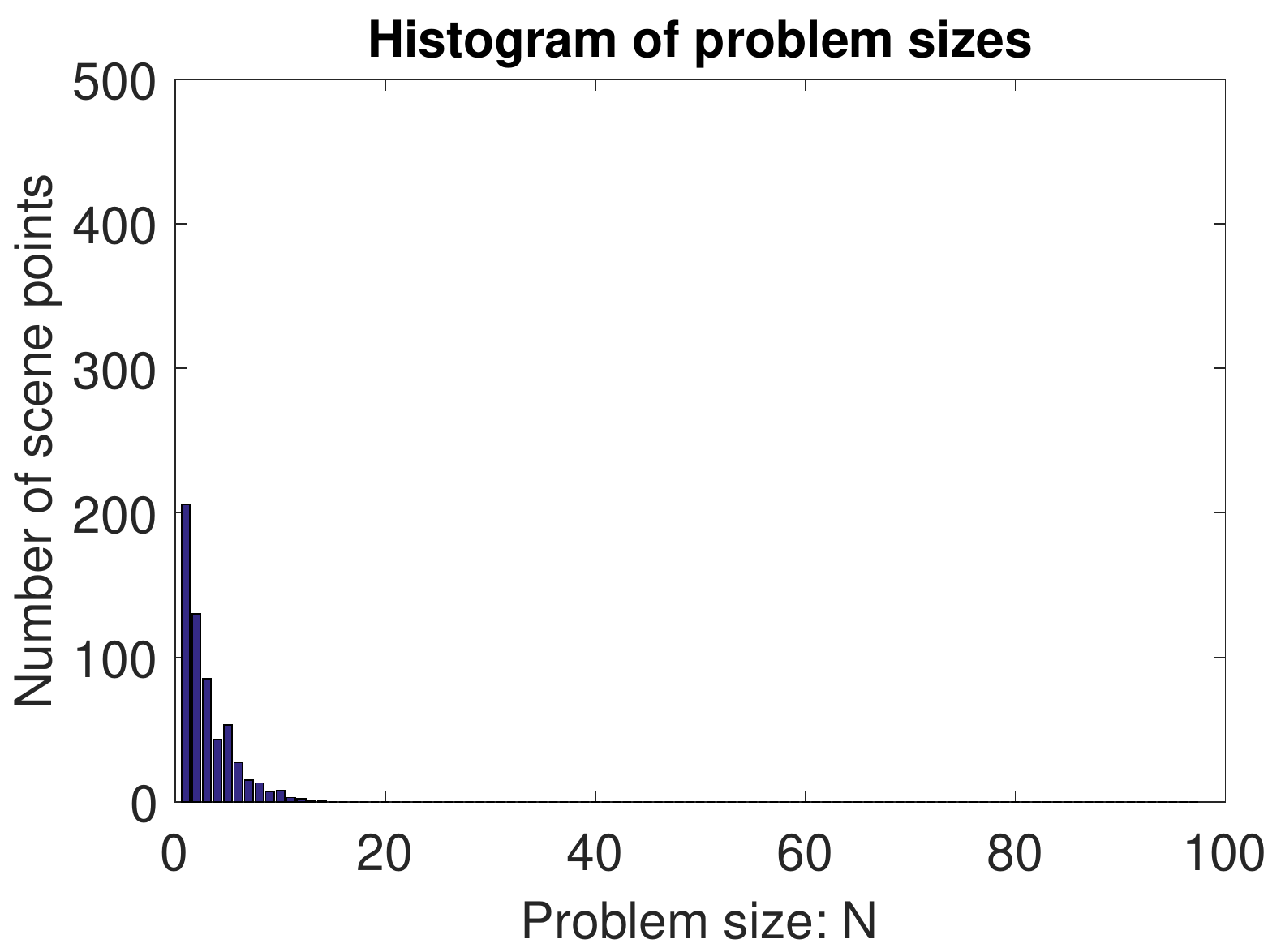}} \hspace{2.3em}
	\subfigure{\includegraphics[width=0.38\textwidth]{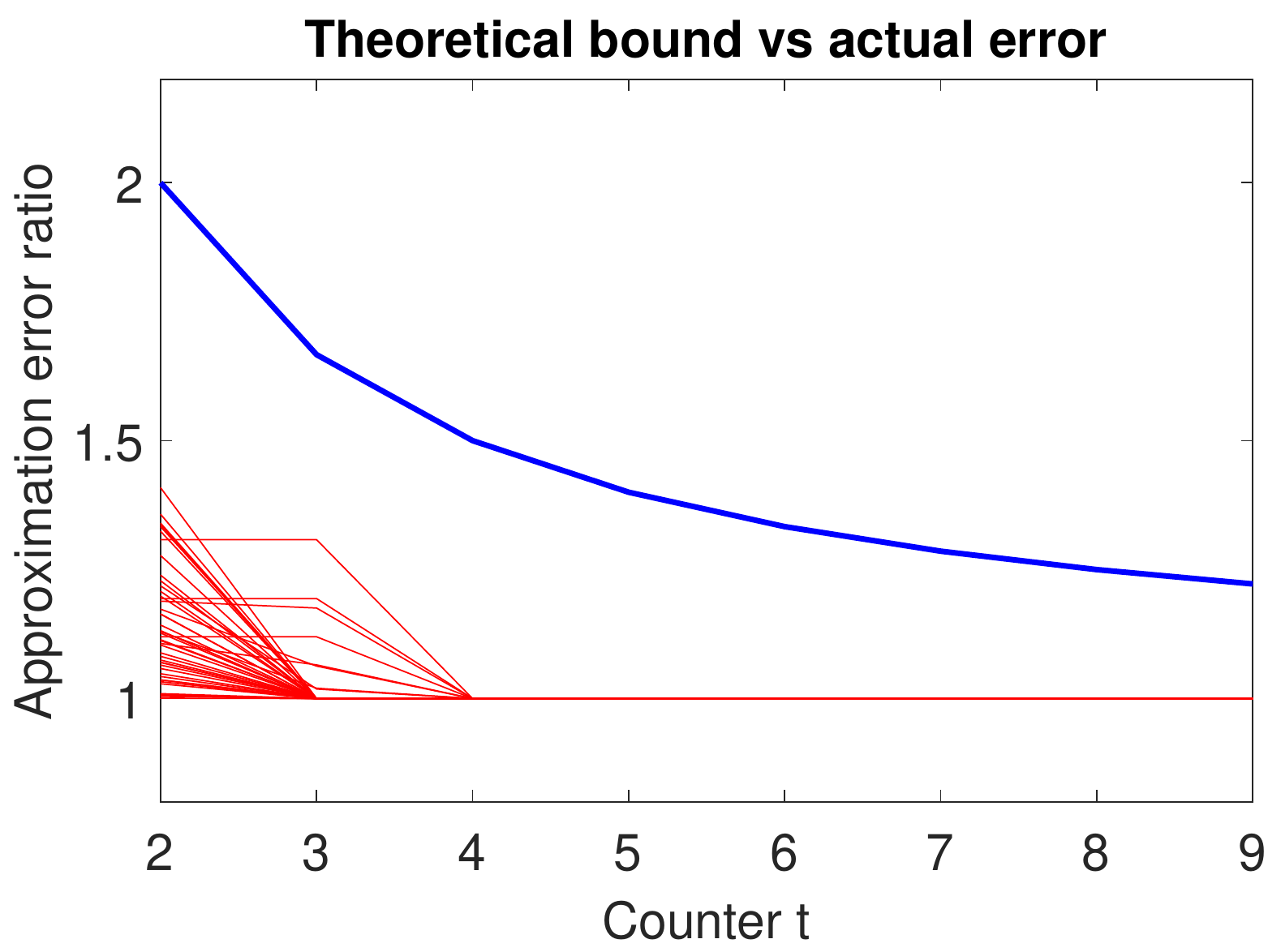}}
	\subfigure{\includegraphics[width=0.38\textwidth]{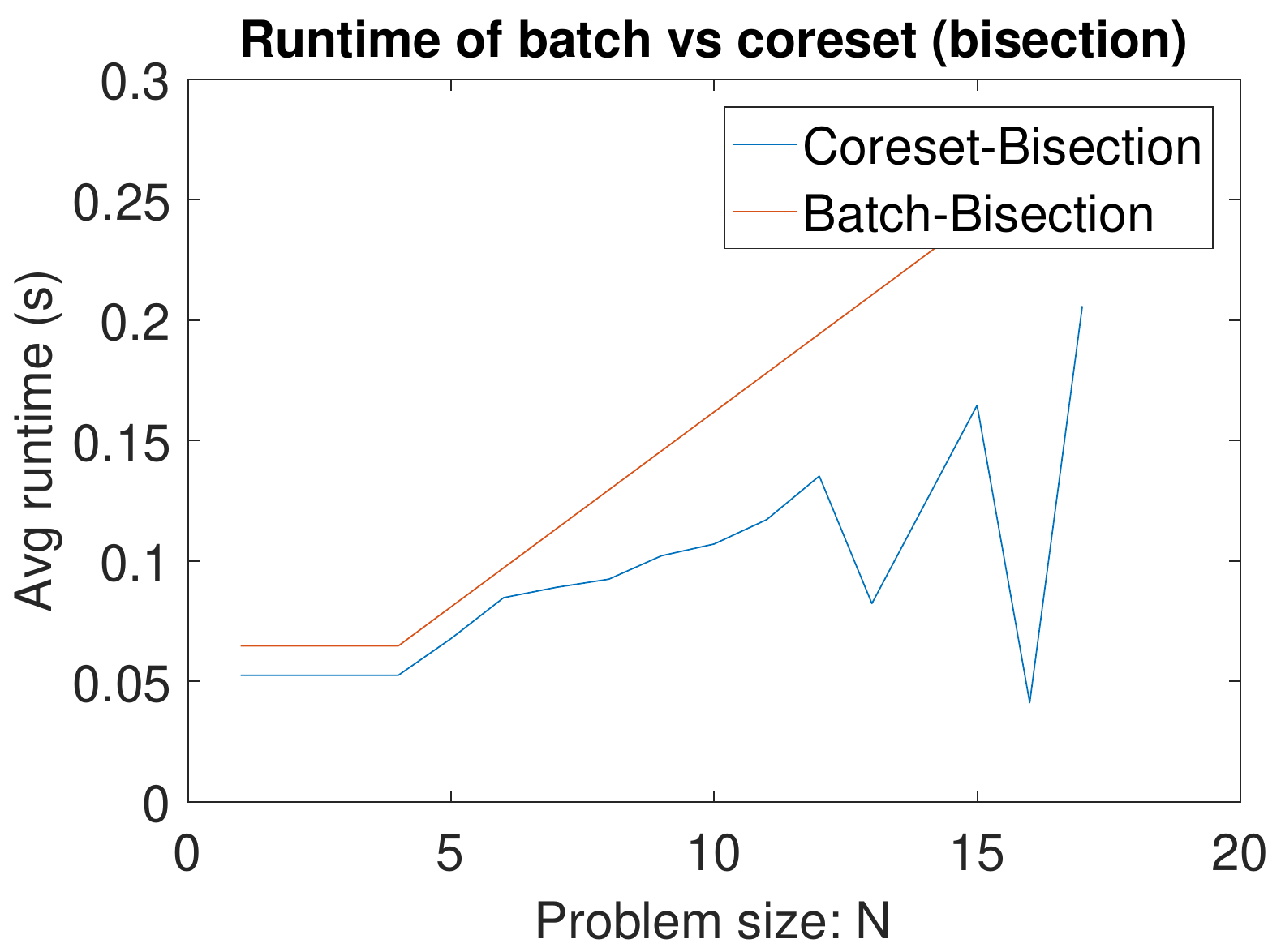}}\hspace{2.3em}
	\subfigure{\includegraphics[width=0.38\textwidth]{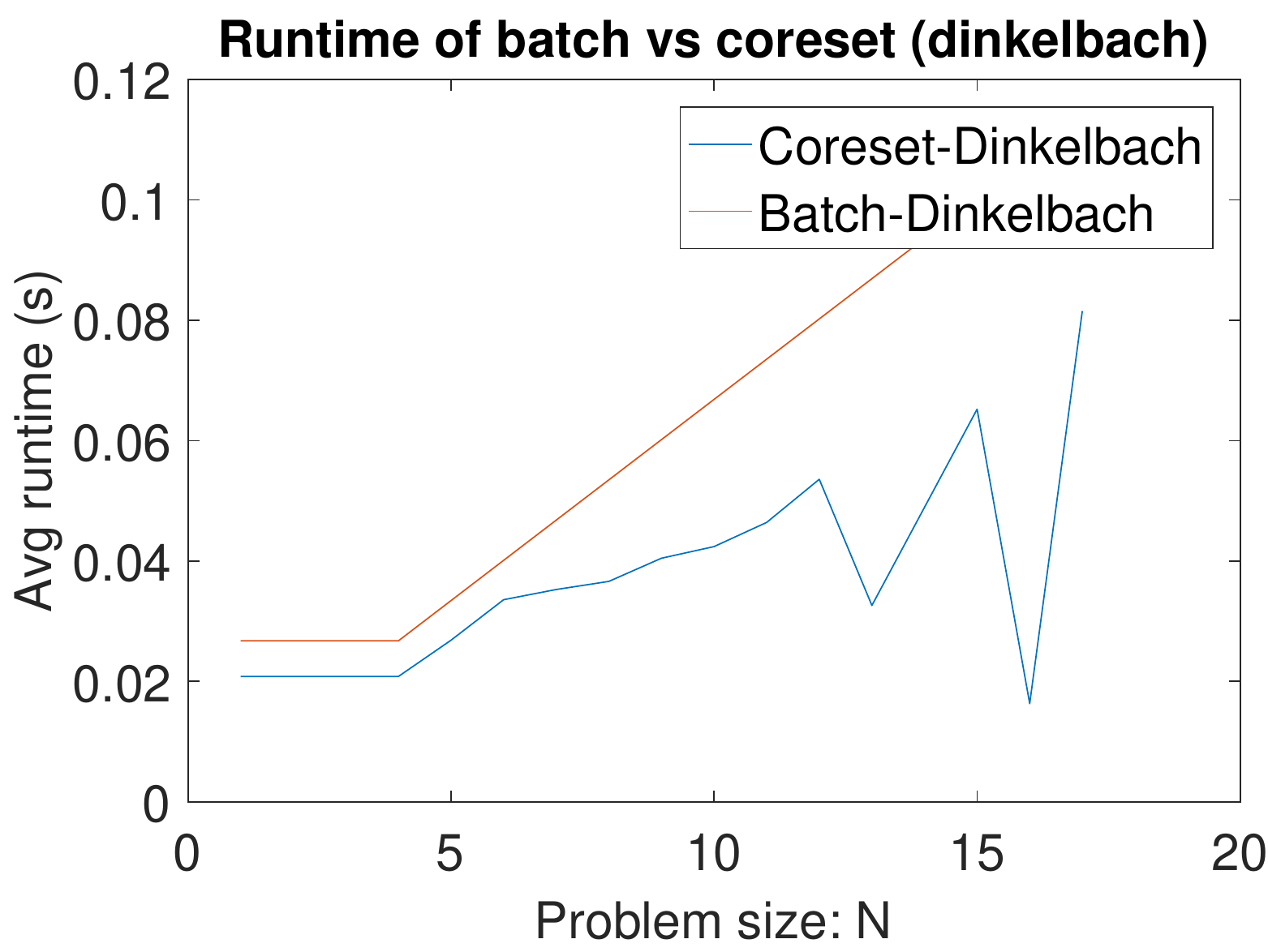}}
	\caption{Results for \textbf{Vercingetorix Statue}. (top left) Histogram of problem sizes. (top right) Approximation error ratio versus error ratio bound. (bottom left) Runtime of coreset vs batch, for bisection solver. (bottom right) Runtime of coreset vs batch, for Dinkelbach solver.}
	\label{fig:Vercingetorix}
\end{figure*}

\begin{figure*}[]\centering
	\subfigure{\includegraphics[width=0.38\textwidth]{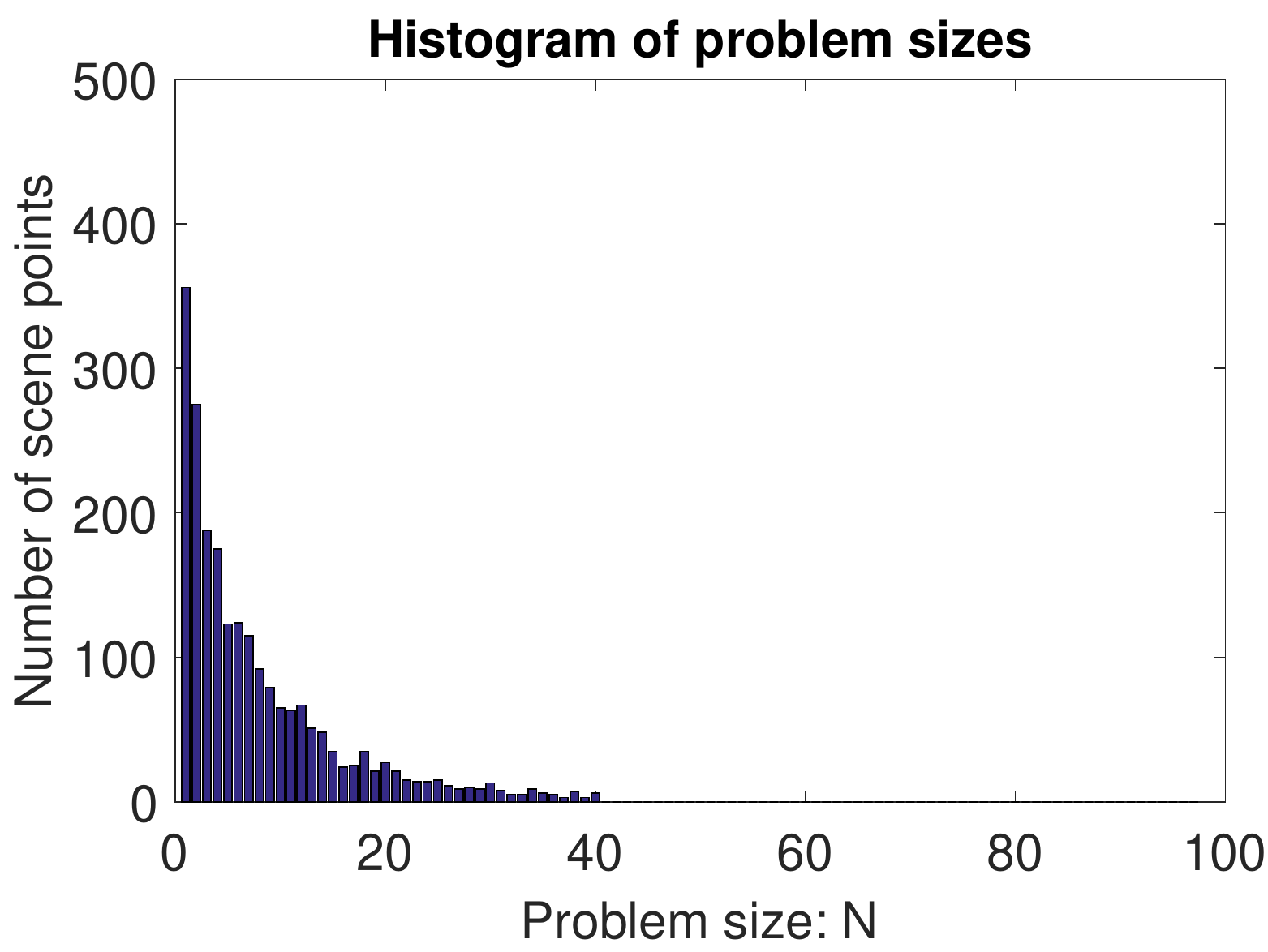}}\hspace{2.3em}
	\subfigure{\includegraphics[width=0.38\textwidth]{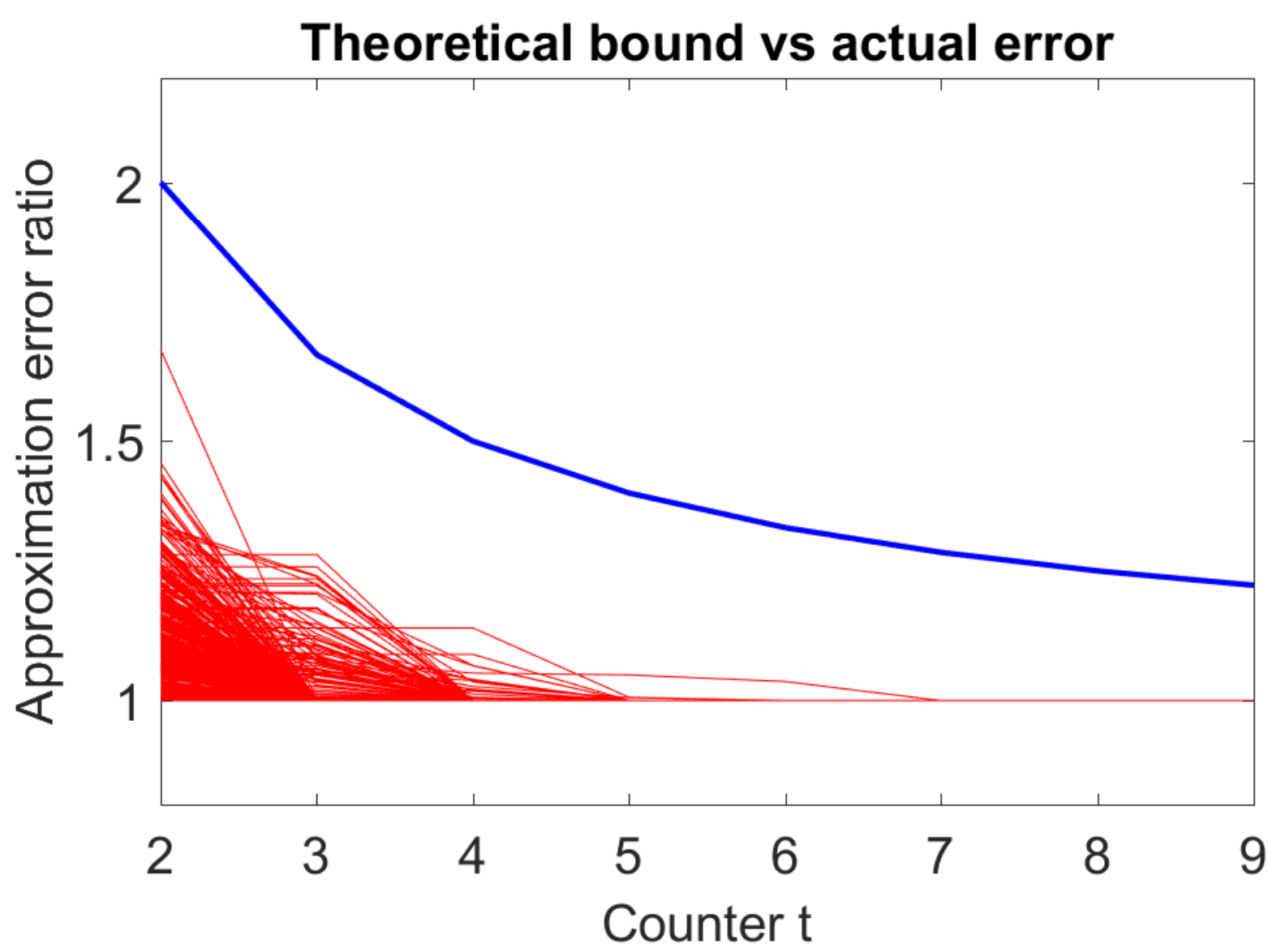}}
	\subfigure{\includegraphics[width=0.38\textwidth]{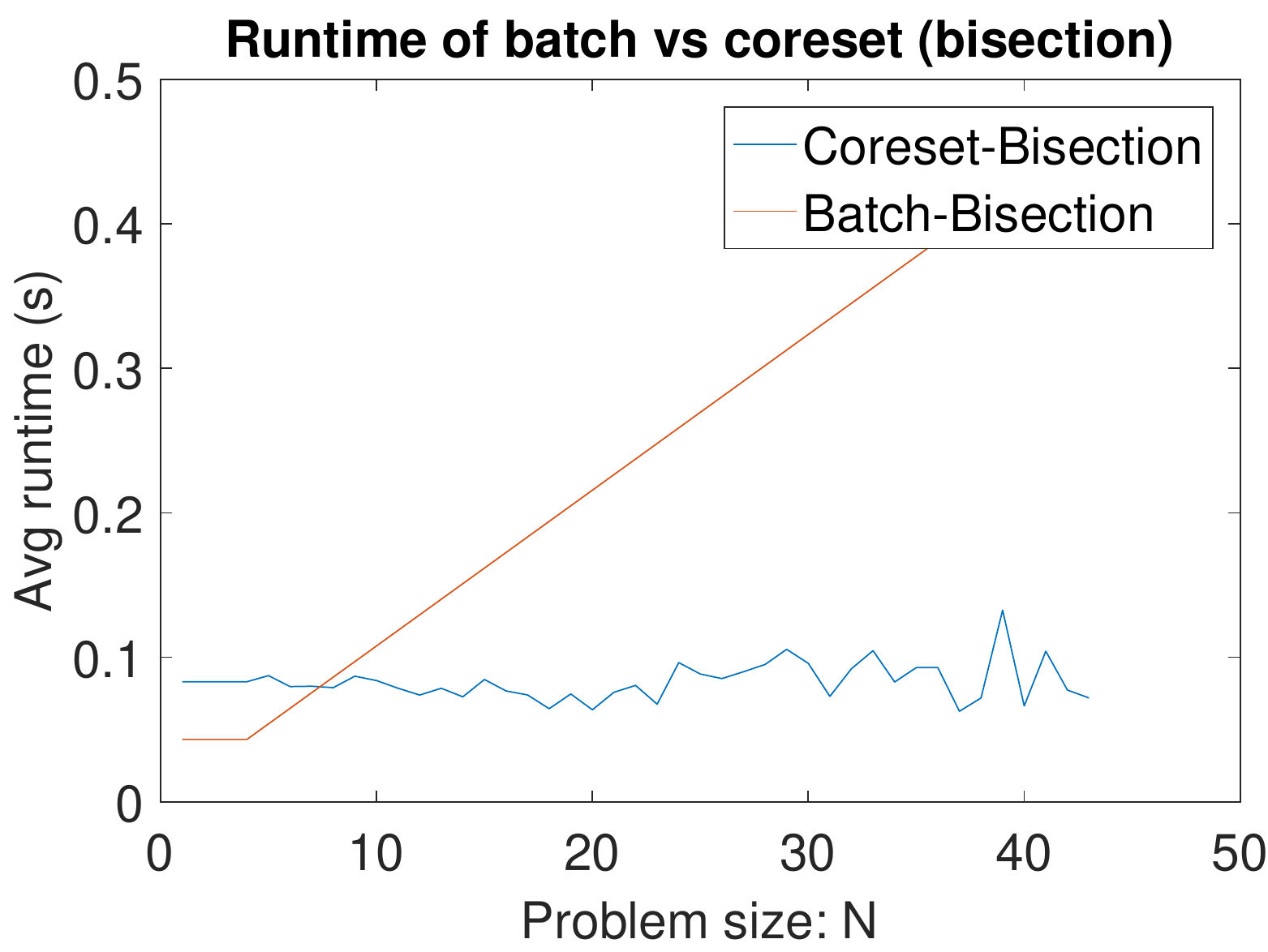}}\hspace{2.3em}
	\subfigure{\includegraphics[width=0.38\textwidth]{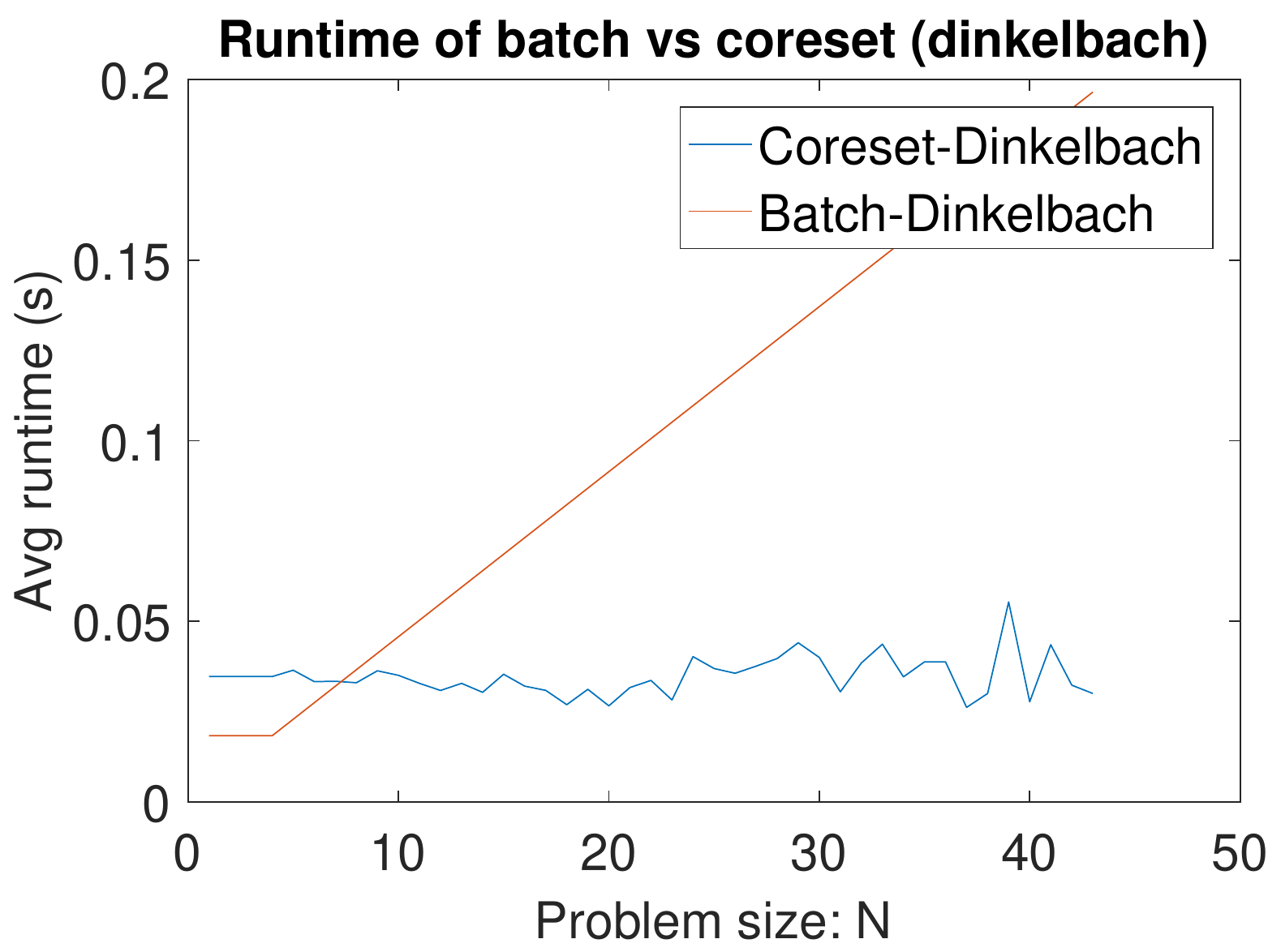}}
	\caption{Results for \textbf{Stockholm City Hall}. (top left) Histogram of problem sizes. (top right) Approximation error ratio versus error ratio bound. (bottom left) Runtime of coreset vs batch, for bisection solver. (bottom right) Runtime of coreset vs batch, for Dinkelbach solver.}
	\label{fig:stockholm}
\end{figure*}

\begin{figure*}[]\centering
	\subfigure{\includegraphics[width=0.38\textwidth]{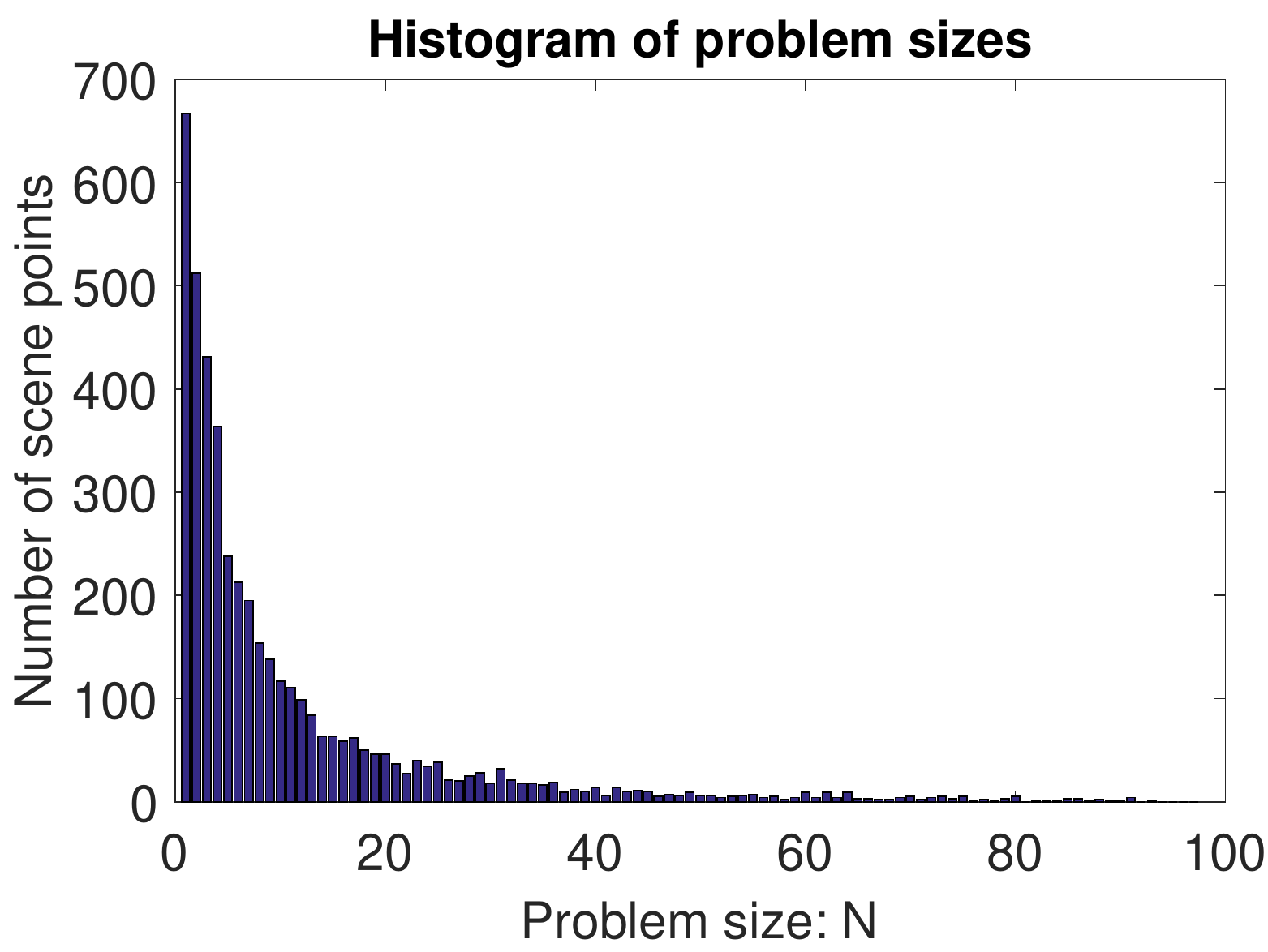}}\hspace{2.3em}
	\subfigure{\includegraphics[width=0.38\textwidth]{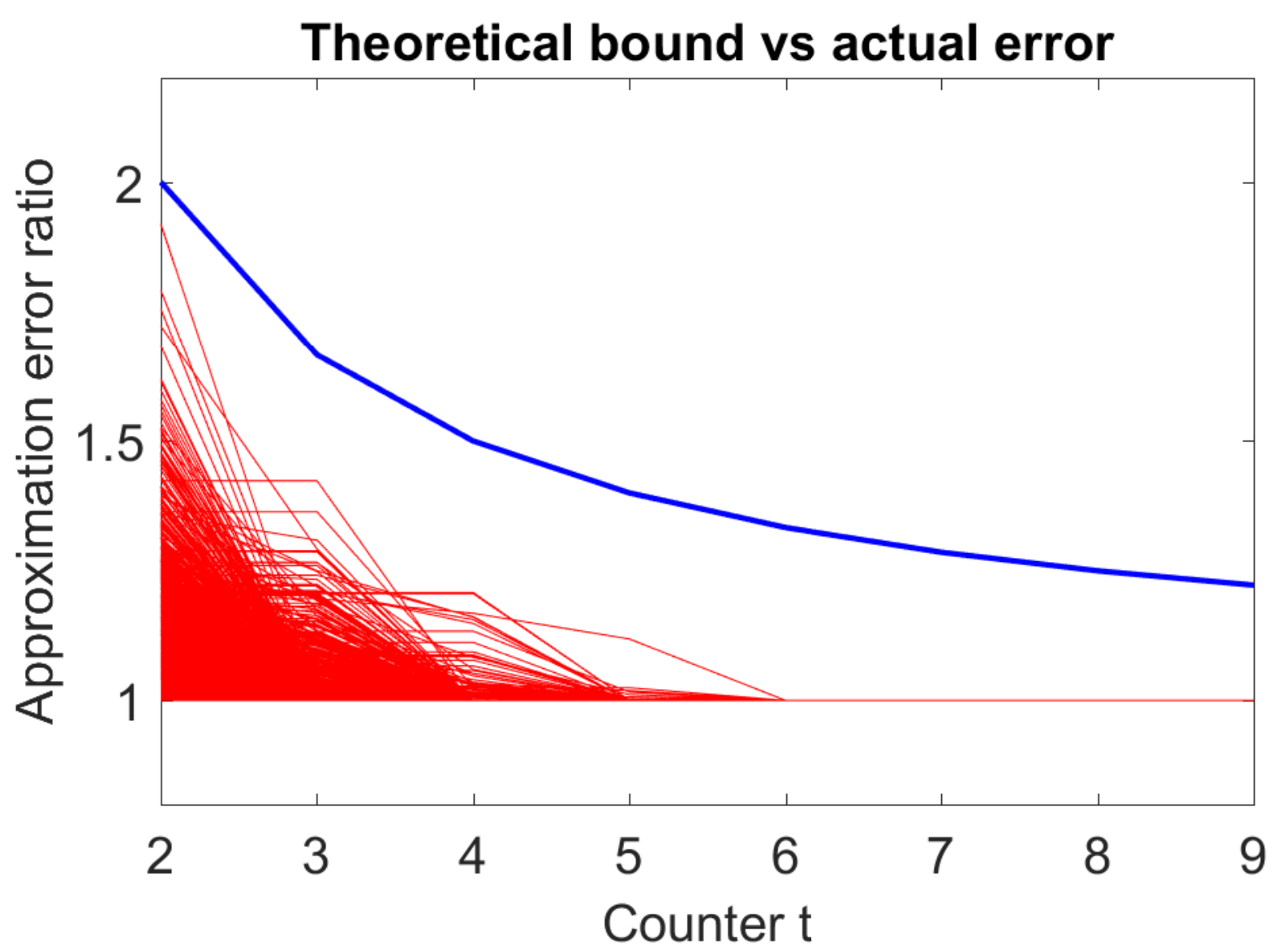}}
	\subfigure{\includegraphics[width=0.38\textwidth]{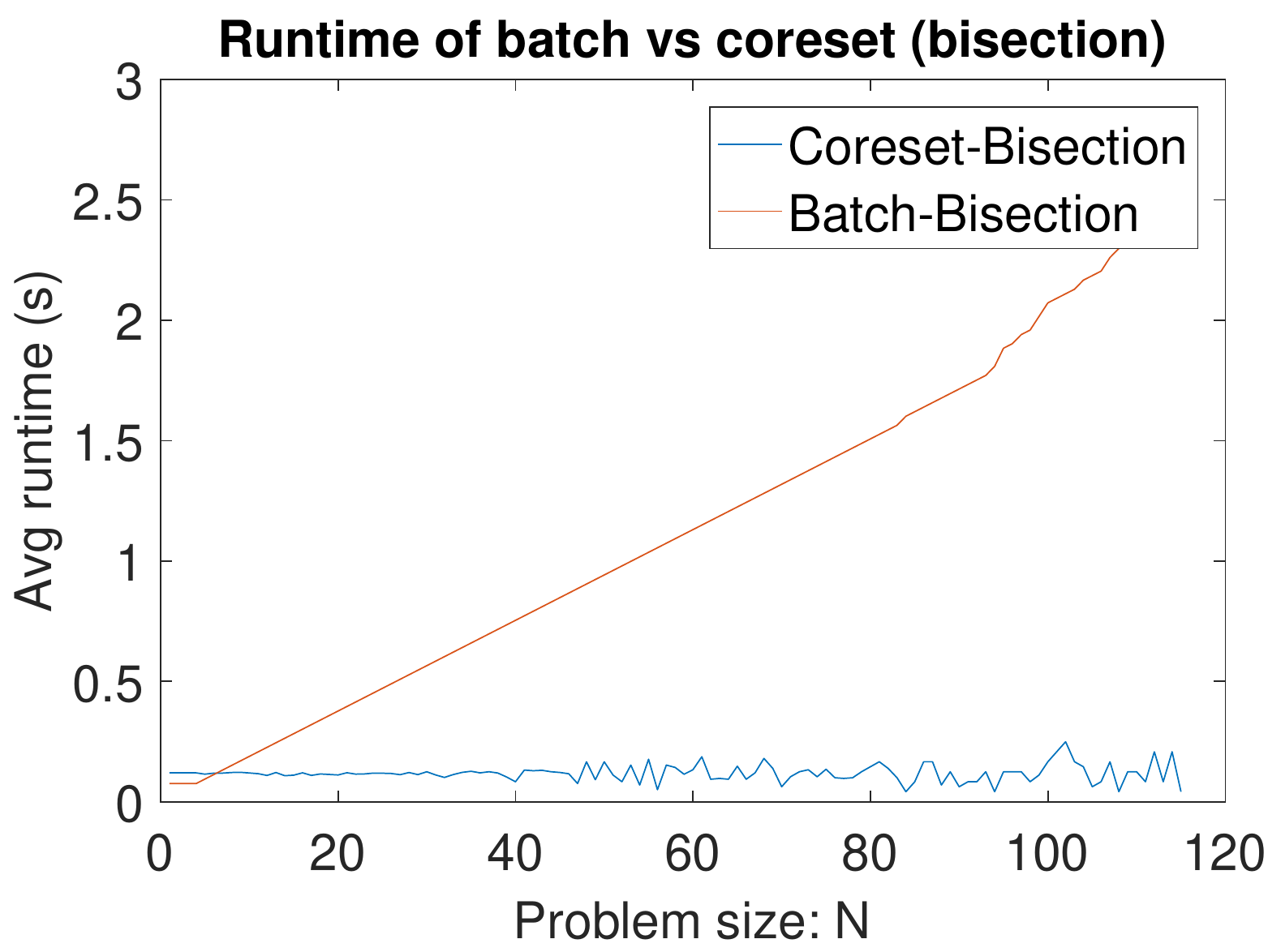}}\hspace{2.3em}
	\subfigure{\includegraphics[width=0.38\textwidth]{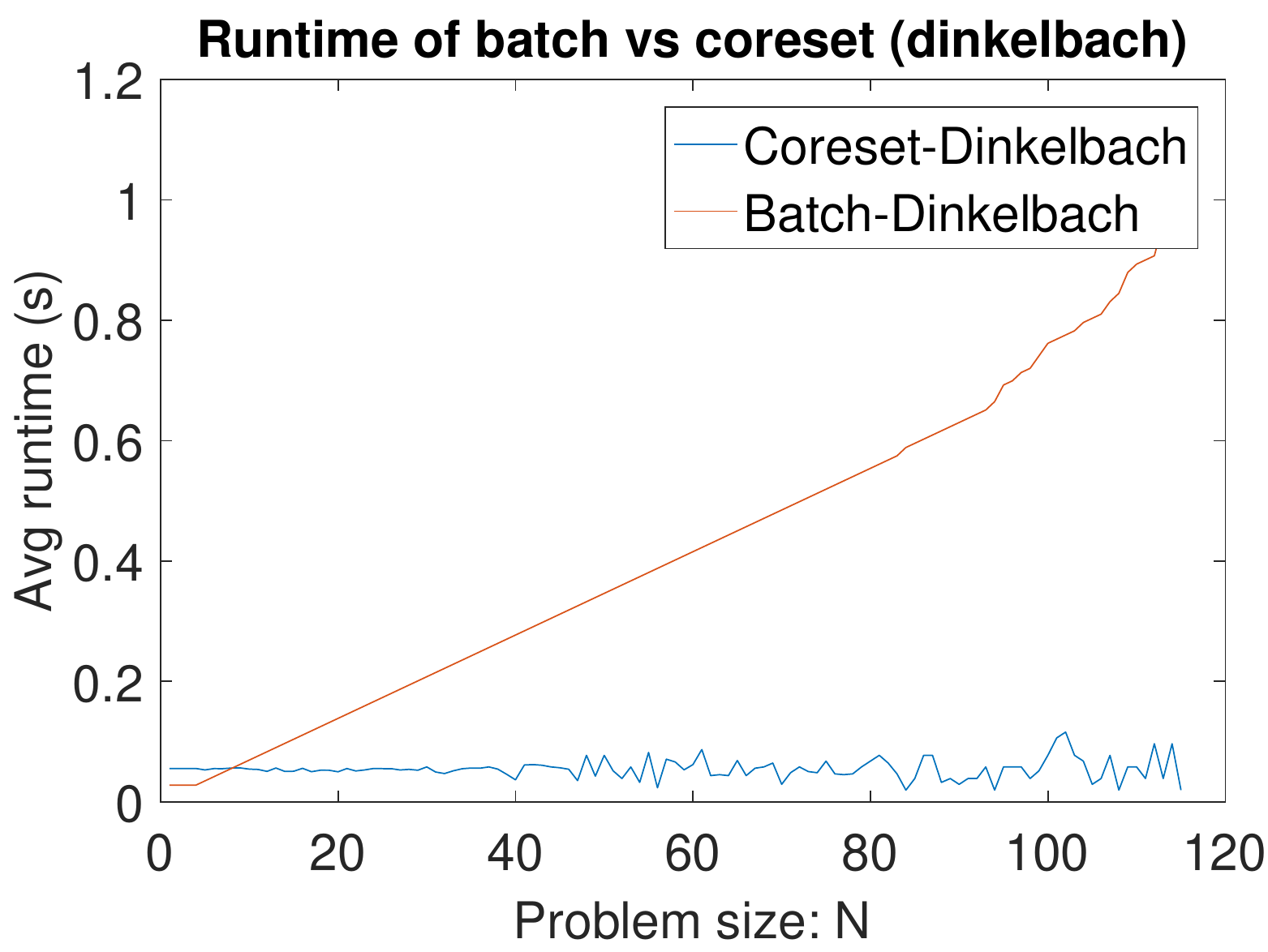}}
	\caption{Results for \textbf{Alcatraz}. (top left) Histogram of problem sizes. (top right) Approximation error ratio versus error ratio bound. (bottom left) Runtime of coreset vs batch, for bisection solver. (bottom right) Runtime of coreset vs batch, for Dinkelbach solver.}
	\label{fig:alcatraz}
\end{figure*}

\begin{figure*}[]\centering
	\subfigure{\includegraphics[width=0.38\textwidth]{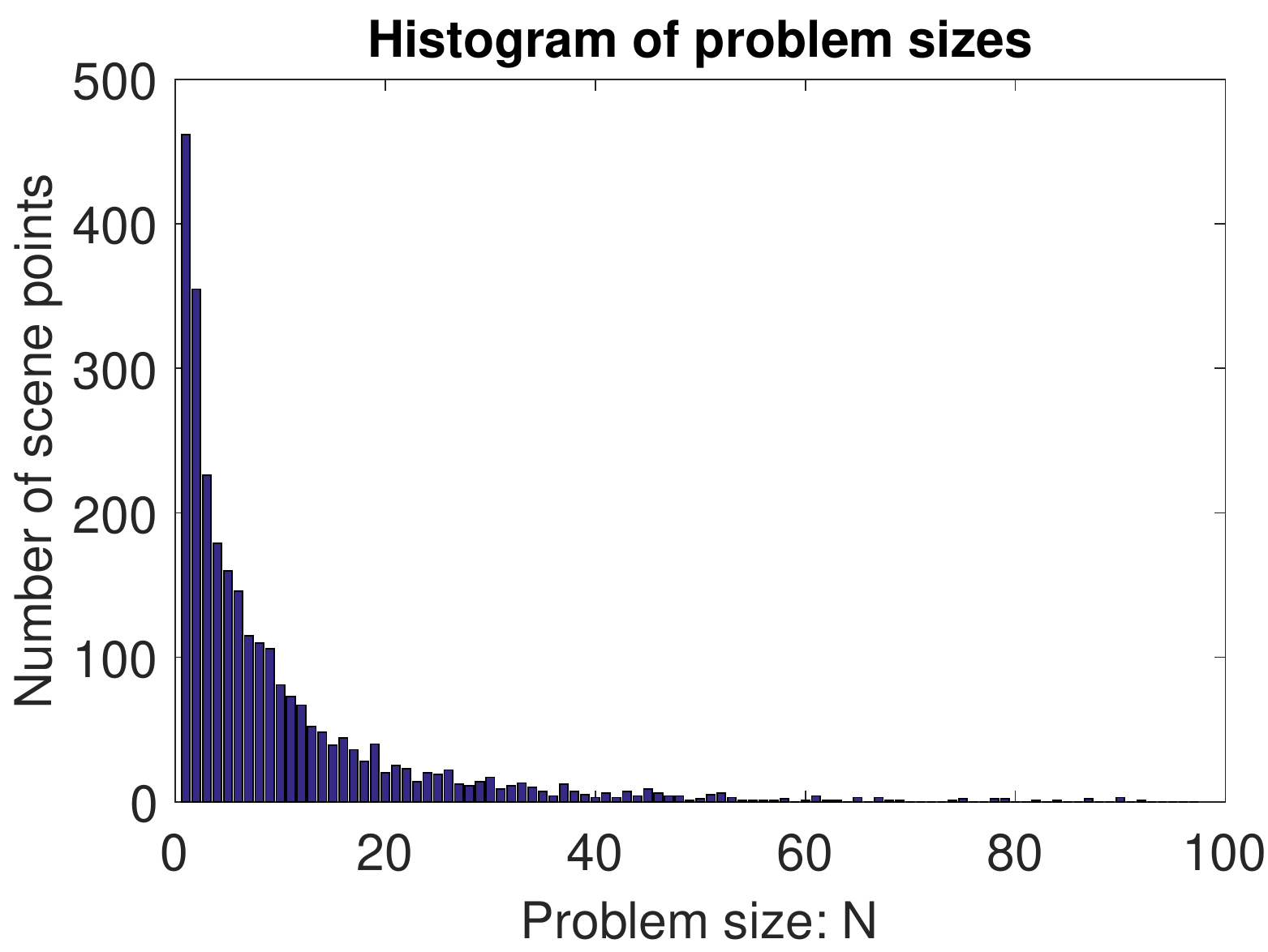}}\hspace{2.3em}
	\subfigure{\includegraphics[width=0.38\textwidth]{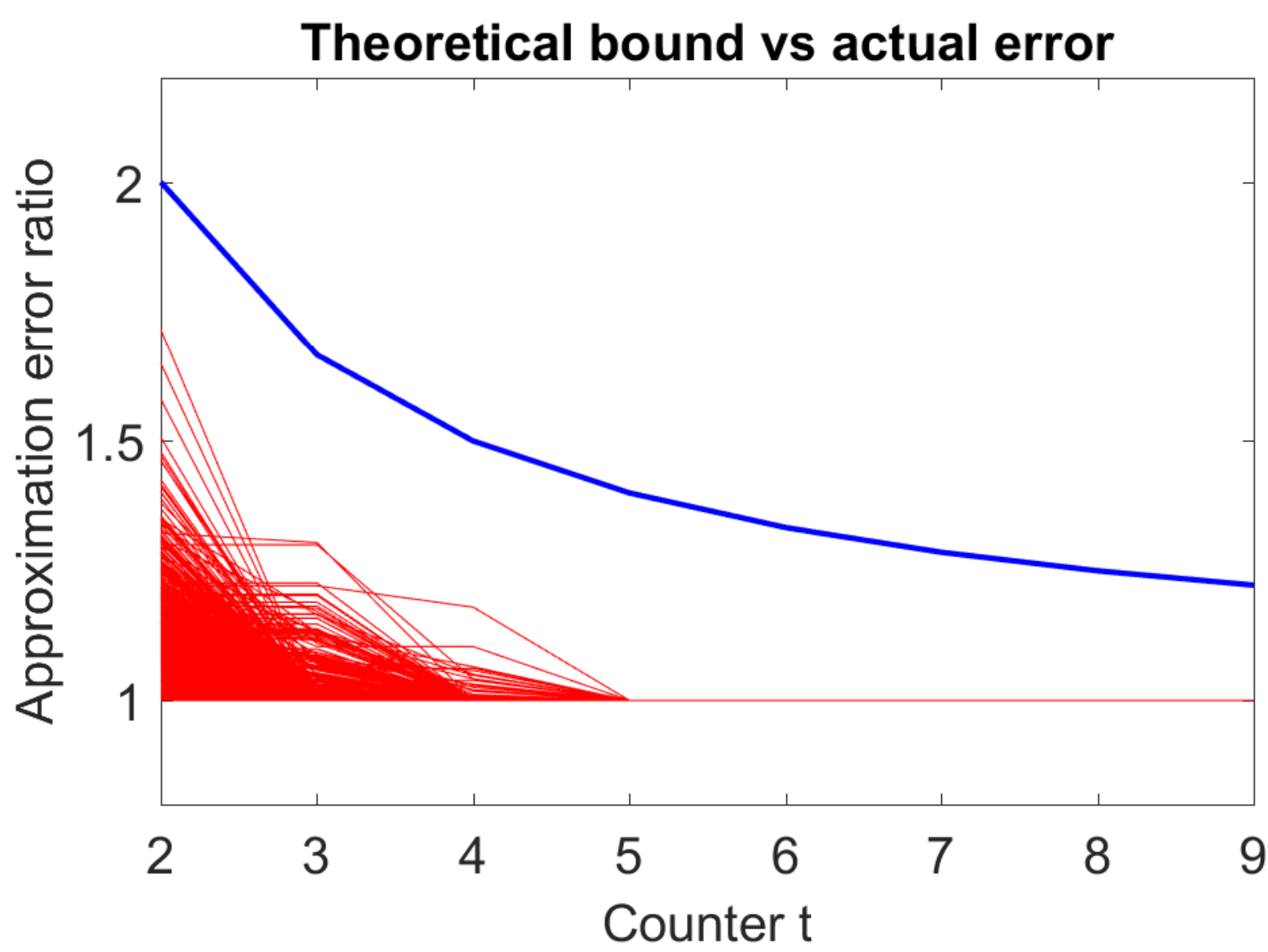}}
	\subfigure{\includegraphics[width=0.38\textwidth]{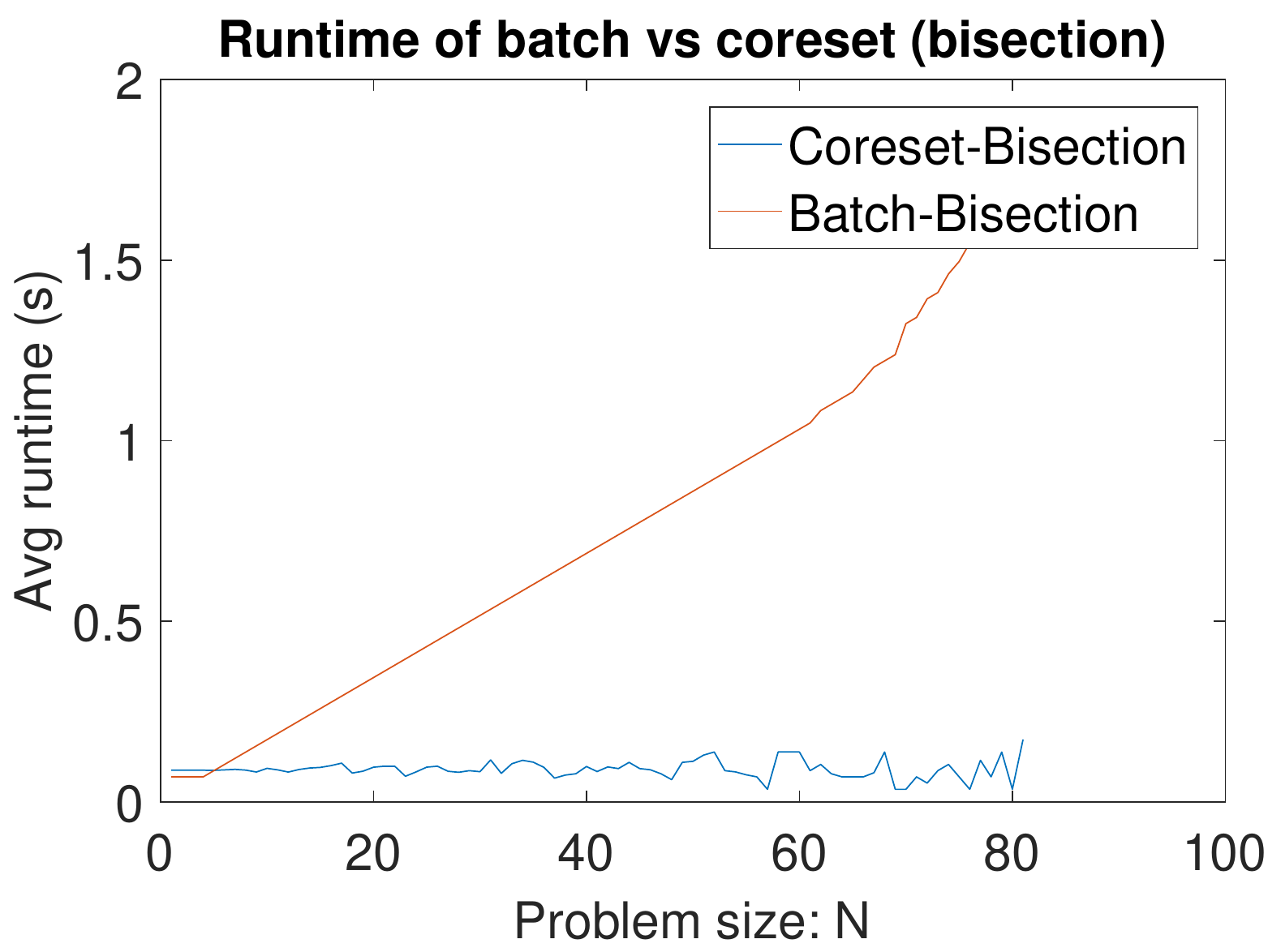}}	\hspace{2.3em}
	\subfigure{\includegraphics[width=0.38\textwidth]{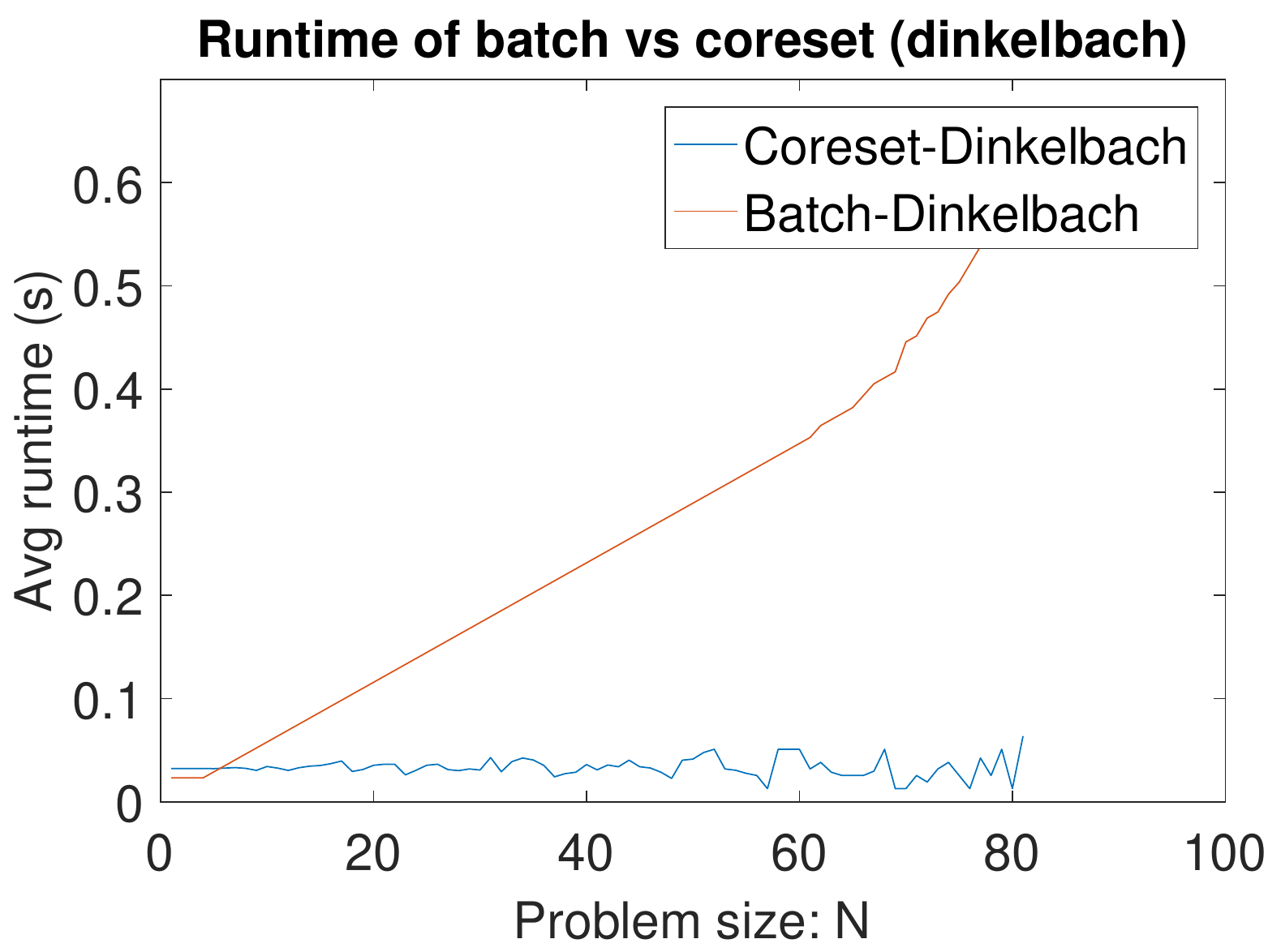}}
	\caption{Results for \textbf{Arc of Triumph}. (top left) Histogram of problem sizes. (top right) Approximation error ratio versus error ratio bound. (bottom left) Runtime of coreset vs batch, for bisection solver. (bottom right) Runtime of coreset vs batch, for Dinkelbach solver.}
	\label{fig:triumph}
\end{figure*}

\begin{figure*}[]\centering
	\subfigure{\includegraphics[width=0.38\textwidth]{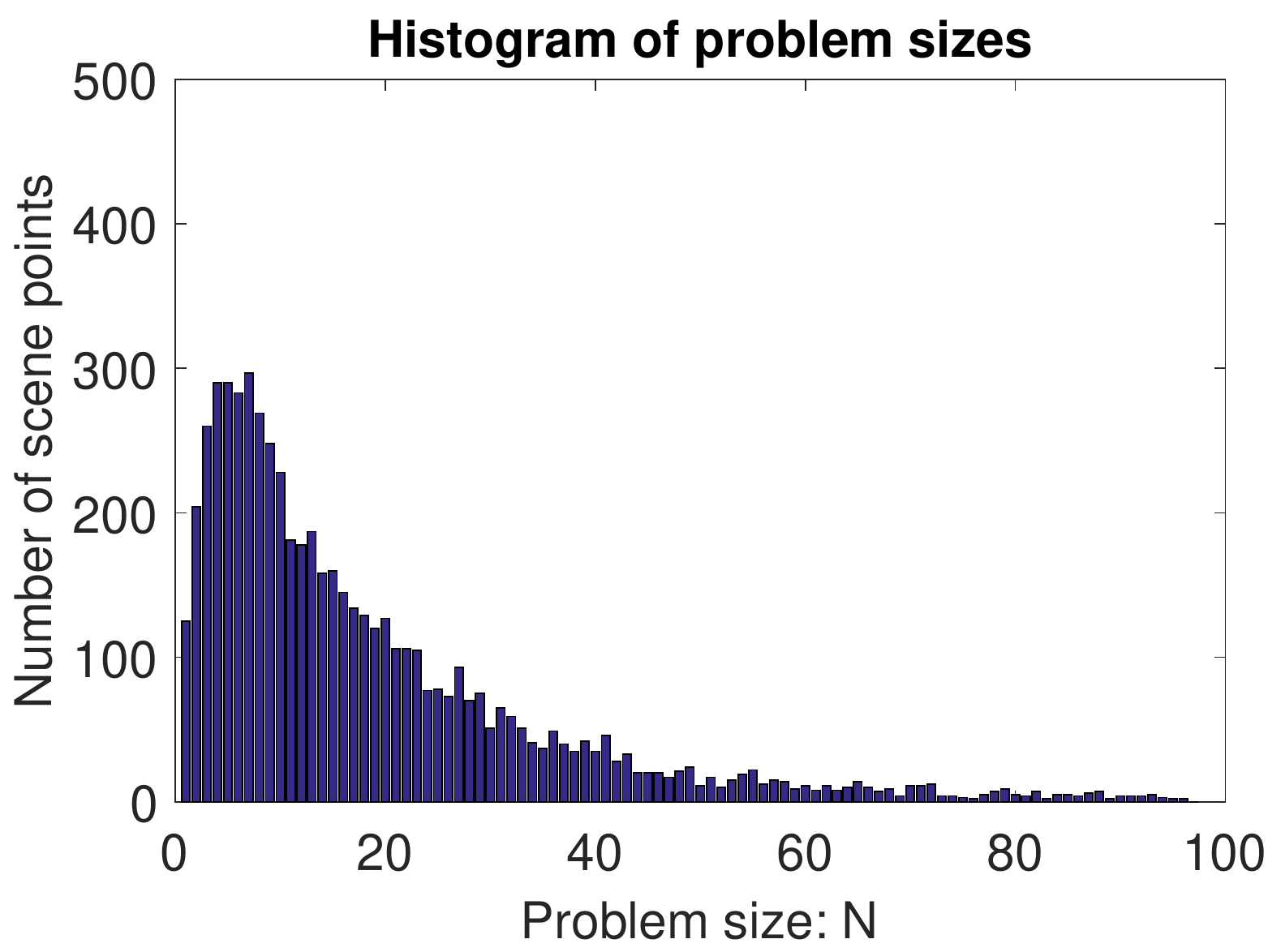}}\hspace{2.3em}
	\subfigure{\includegraphics[width=0.38\textwidth]{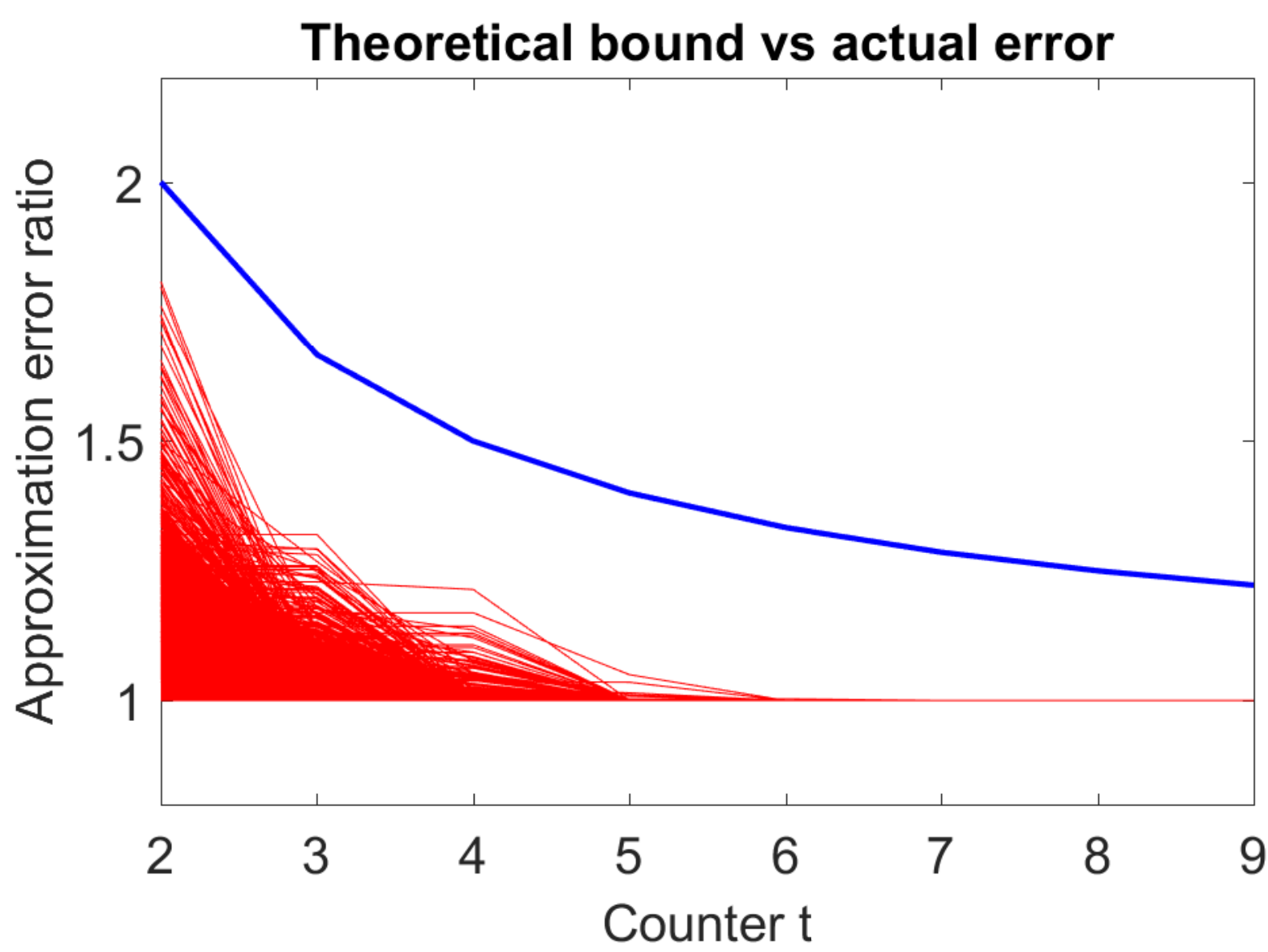}}
	\subfigure{\includegraphics[width=0.38\textwidth]{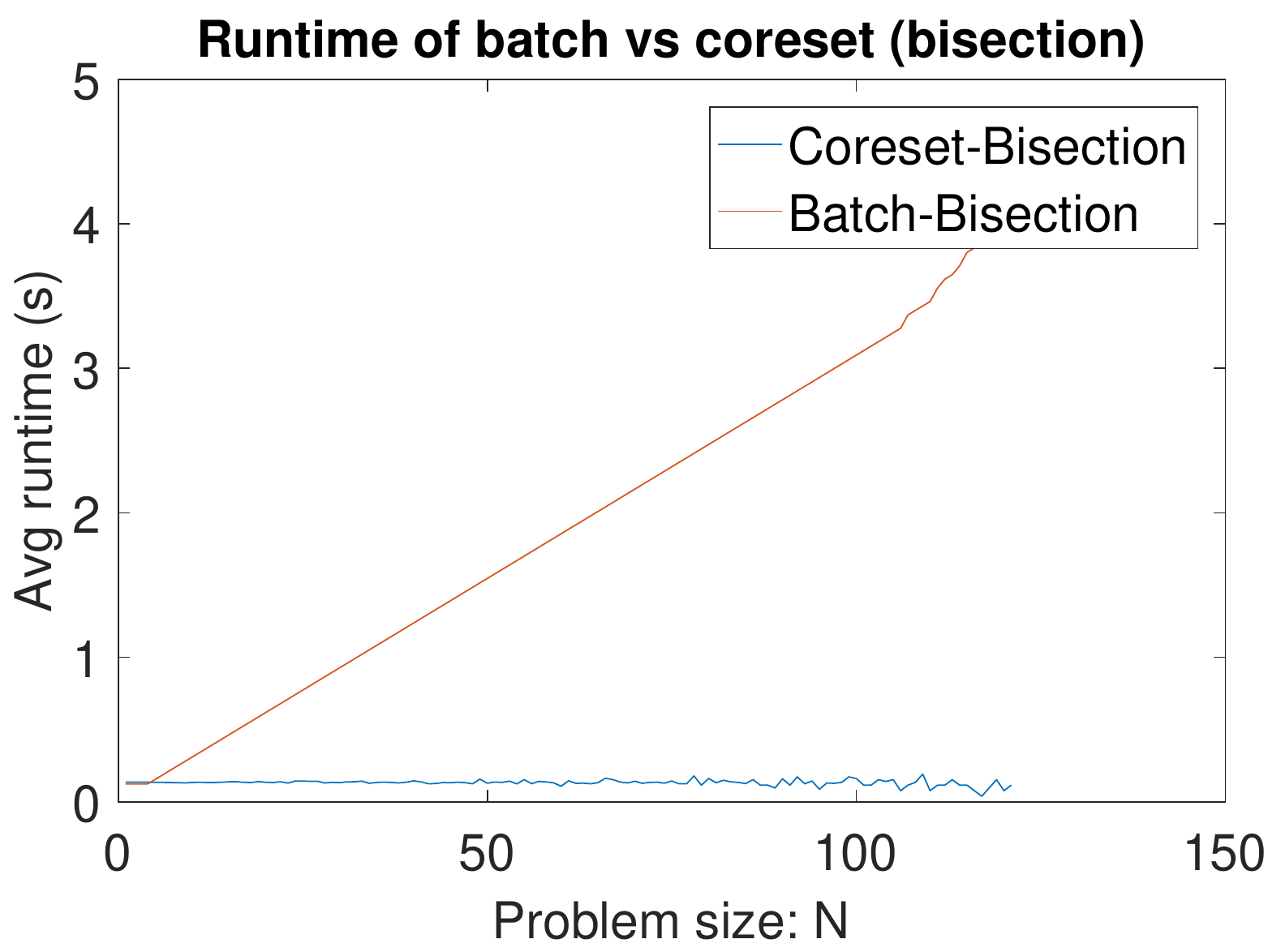}}\hspace{2.3em}	
	\subfigure{\includegraphics[width=0.38\textwidth]{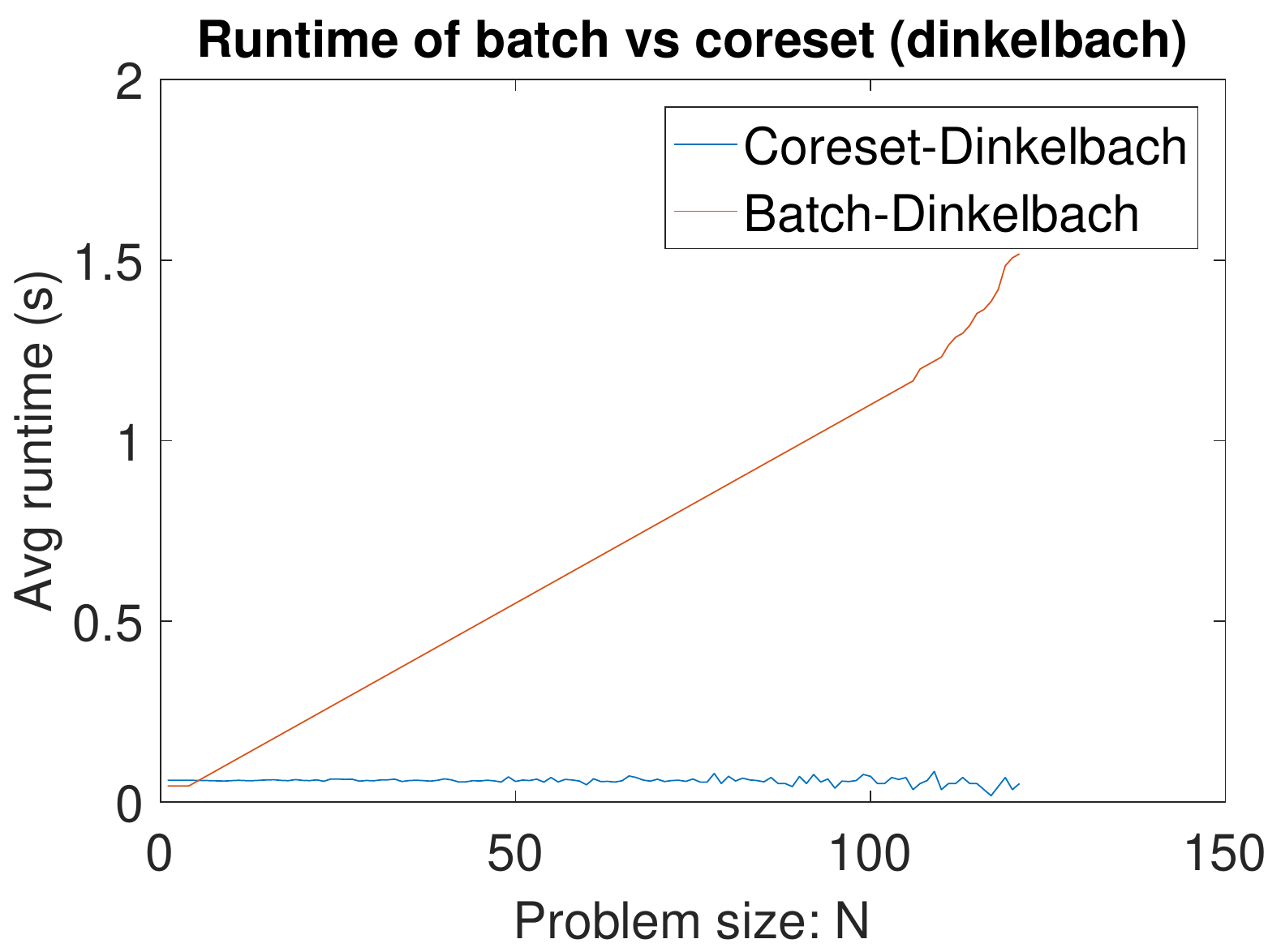}}
	\caption{Results for \textbf{\"{O}rebro Castle}. (top left) Histogram of problem sizes. (top right) Approximation error ratio versus error ratio bound. (bottom left) Runtime of coreset vs batch, for bisection solver. (bottom right) Runtime of coreset vs batch, for Dinkelbach solver.}
	\label{fig:orebrocastle}
\end{figure*}

\begin{figure*}[]\centering
	\subfigure{\includegraphics[width=0.38\textwidth]{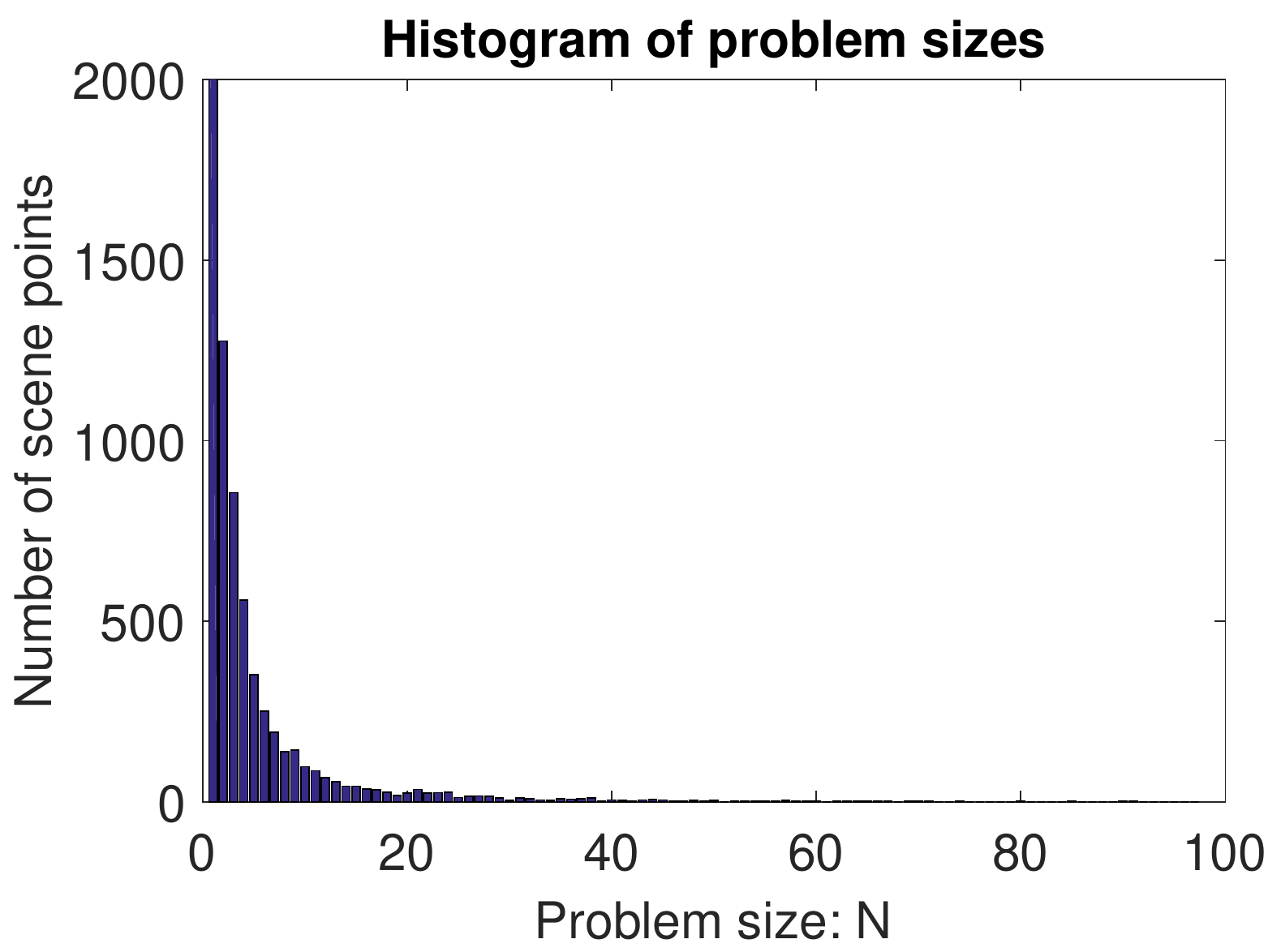}}\hspace{2.3em}
	\subfigure{\includegraphics[width=0.38\textwidth]{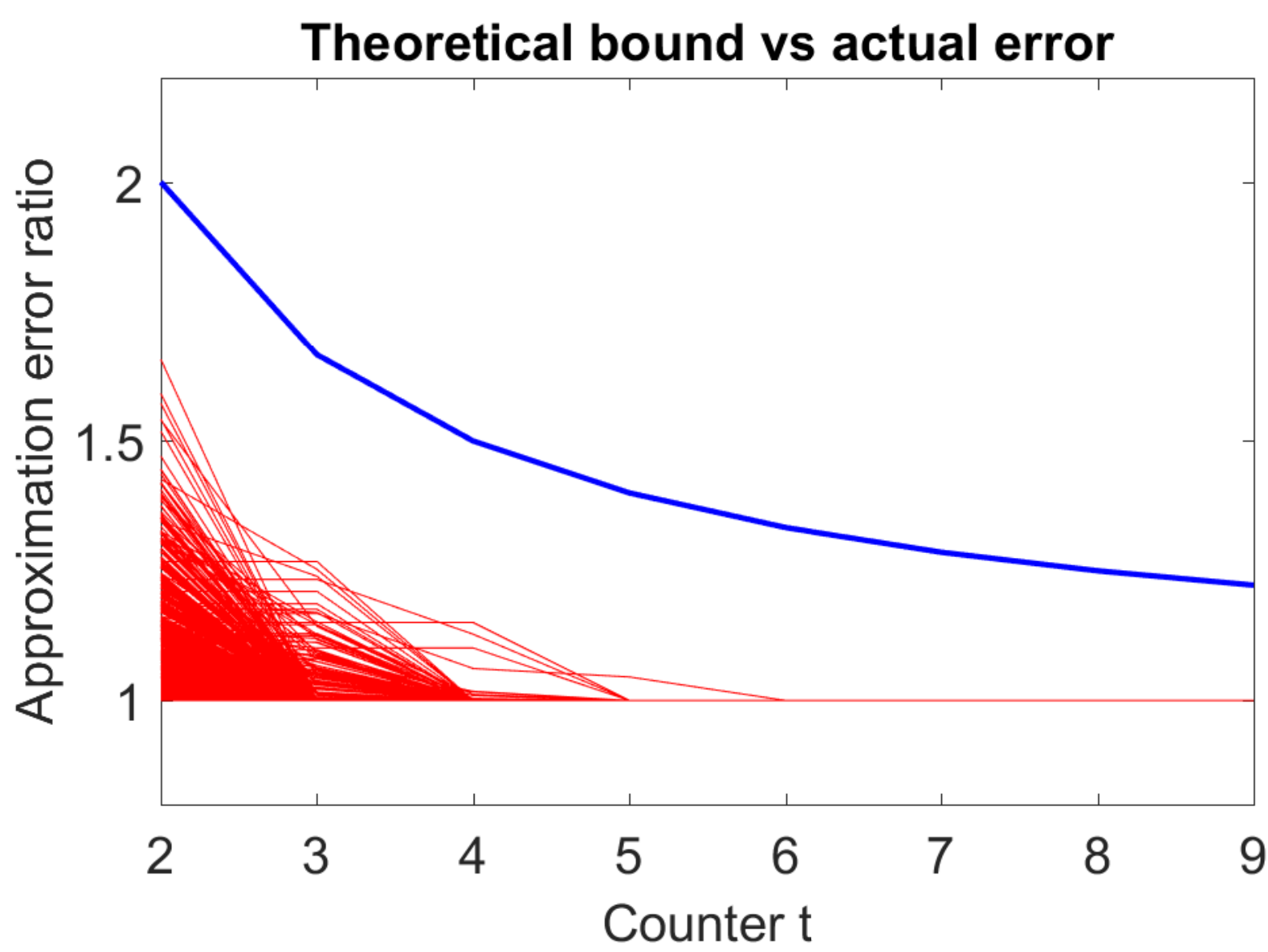}}
	\subfigure{\includegraphics[width=0.38\textwidth]{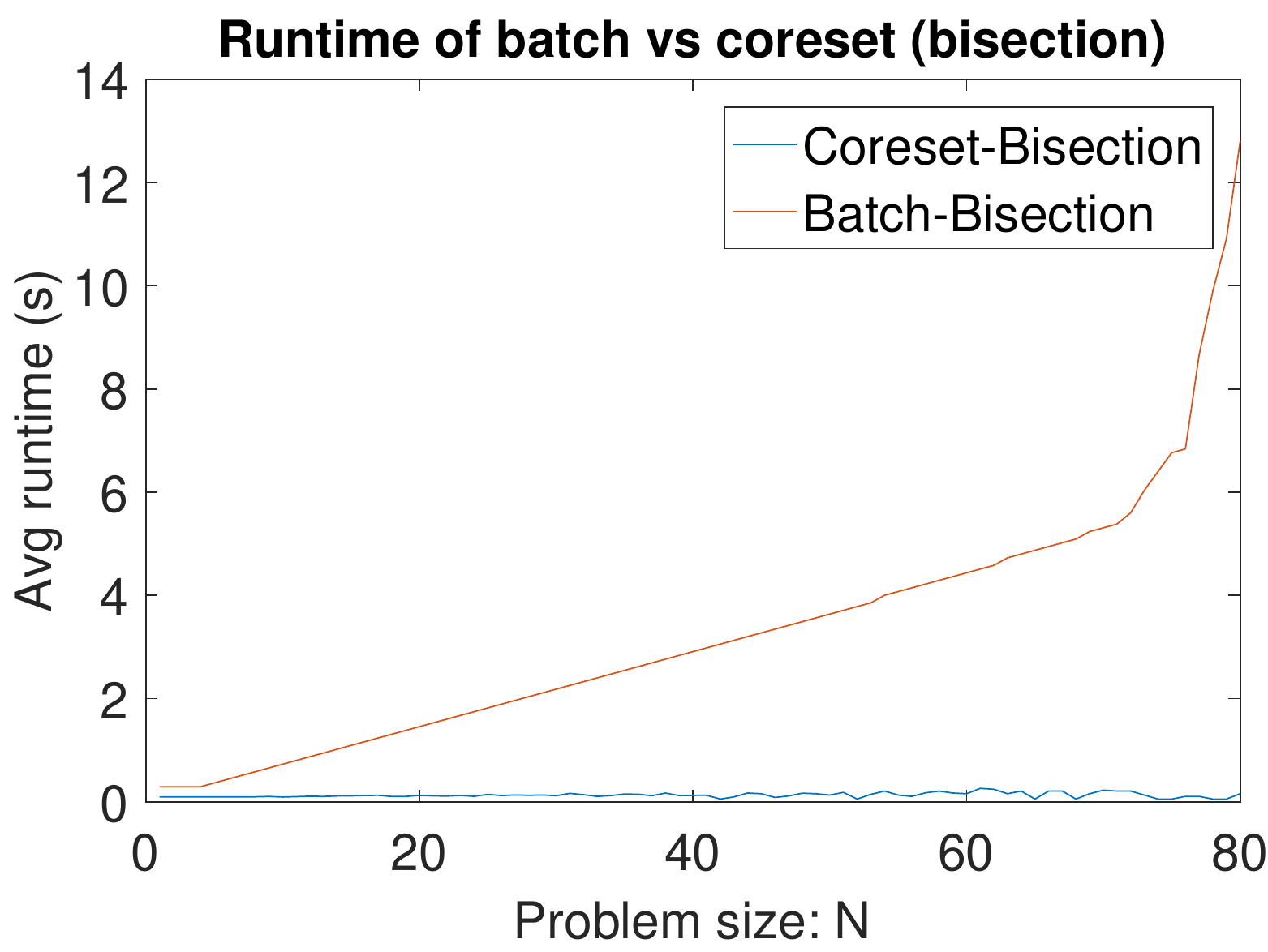}}\hspace{2.3em}	
	\subfigure{\includegraphics[width=0.38\textwidth]{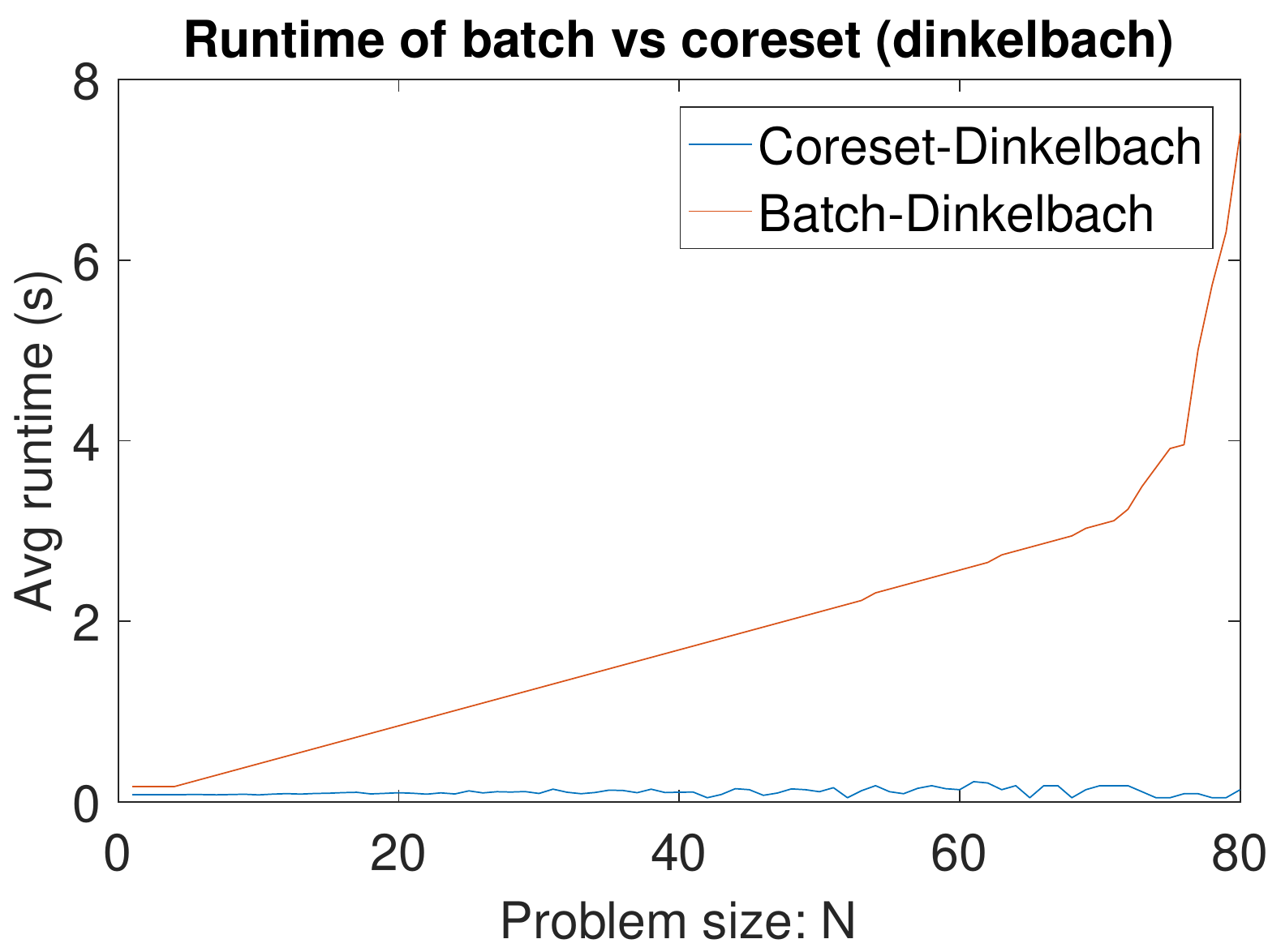}}
	\caption{Results for \textbf{Notre Dame}. (top left) Histogram of problem sizes. (top right) Approximation error ratio versus error ratio bound. (bottom left) Runtime of coreset vs batch, for bisection solver. (bottom right) Runtime of coreset vs batch, for Dinkelbach solver.}
	\label{fig:NotreDame}
\end{figure*}

\clearpage

\bibliographystyle{IEEEtran}
\bibliography{coresets}
%



%
\begin{IEEEbiography}[{\includegraphics[width=1in,height=1.15in,clip, keepaspectratio]{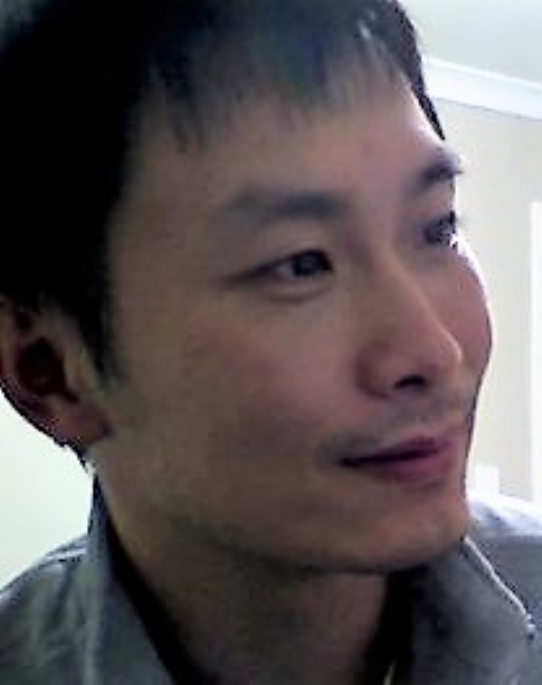}}]{Qianggong Zhang}
	received the BEng degree in computer science and techonology in 2004 and the MEng degree in computer science and techonology in 2007. Since 2015, he has been a PhD candidate at The University of Adelaide, South Australia. His primary research areas include approximation algorithms for geometric computer vision problems.
\end{IEEEbiography}

\begin{IEEEbiography}[{\includegraphics[width=1in,height=1.25in,clip,keepaspectratio]{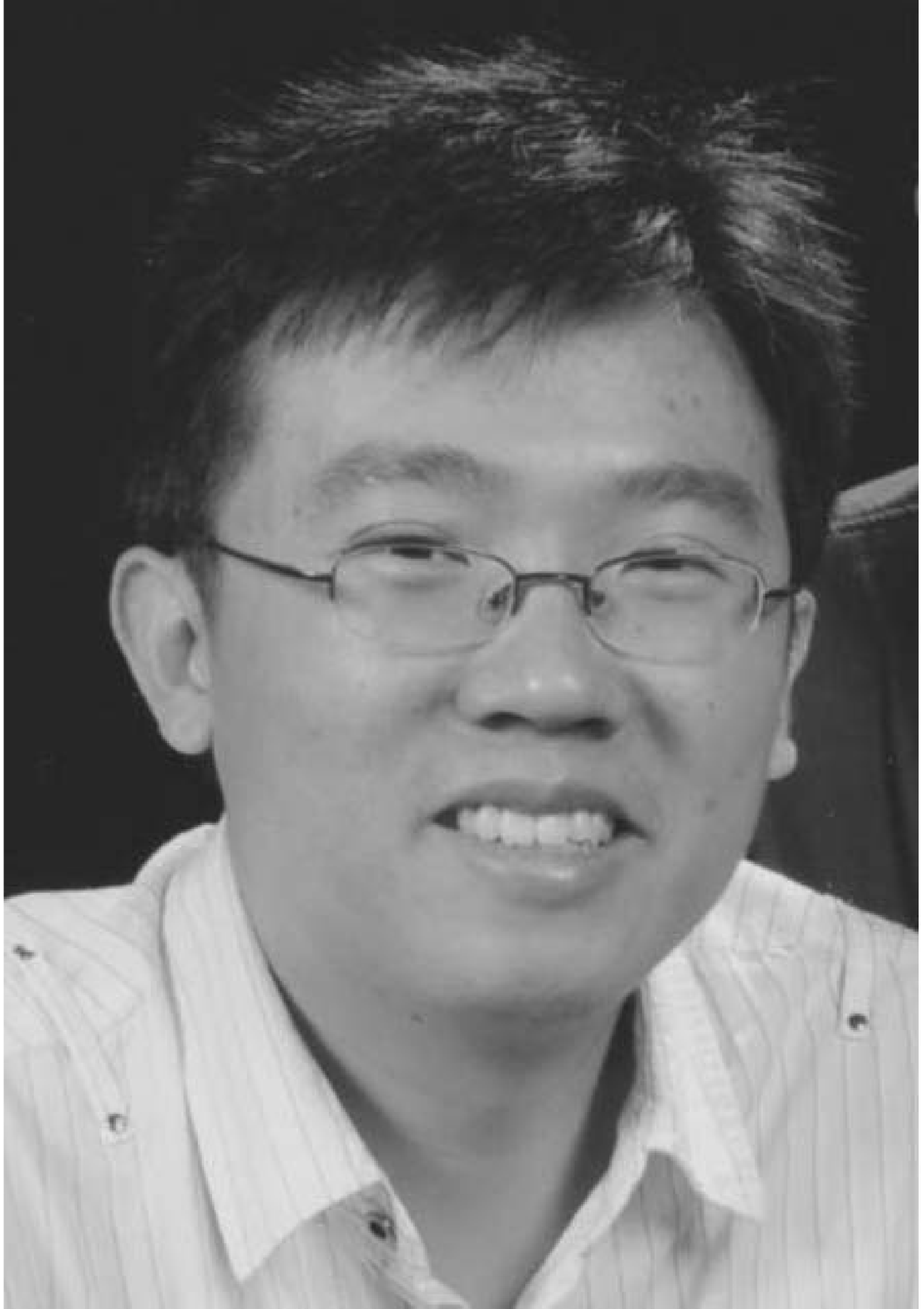}}]{Tat-Jun Chin}
	received the BEng degree in	mechatronics engineering from Universiti Teknologi Malaysia (UTM) in 2003 and the PhD degree in computer systems engineering from Monash University, Victoria, Australia, in 2007. He was a research fellow at the Institute for Infocomm Research (I2R) in Singapore from 2007 to 2008. Since 2008, he has been at The University of Adelaide, South Australia, and is now an Associate Professor. He is an Associate Editor of IPSJ Transactions on Computer Vision and Applications (CVA). His research interests include robust estimation and geometric optimisation.
\end{IEEEbiography}






\end{document}